\newif\ifdraft\draftfalse

\PassOptionsToPackage{prologue,dvipsnames}{xcolor}
\documentclass[acmsmall,screen,nonacm]{acmart}
\setcopyright{cc}
\setcctype{by}

\usepackage{amsmath}
\usepackage{amsthm}
\usepackage{bbm}
\usepackage{stmaryrd}
\usepackage{mathtools}
\usepackage{mathpartir}
\usepackage{enumitem}
\usepackage{pl-syntax}
\usepackage{doi}
\usepackage{ifthen}
\usepackage{xparse}
\usepackage{xspace}
\usepackage{xcolor}
\usepackage{suffix}
\usepackage{centernot}
\usepackage{scalerel}
\usepackage{pict2e}
\usepackage{adjustbox}
\usepackage{thmtools, thm-restate}
\usepackage[normalem]{ulem}
\usepackage[most,breakable]{tcolorbox}
\tcbuselibrary{breakable}
\tcbuselibrary{skins}
\usepackage{tikz}
\usetikzlibrary{decorations.pathreplacing,positioning,calc,arrows,arrows.meta,fit,matrix}
\usepackage{tikz-cd}
\usepackage{soul}

\usepackage{quiver}

\tikzset{
  mapsto/.style={{Bar[width=5pt]}-{Stealth[length=5pt,width=4pt]},line width=0.5pt},
  label/.style={font=\small},
  val node/.style={inner sep=0pt,outer sep=0pt},
}

\usepackage{langrules}

\input{macros}

\DefineNtwkRule{Ctx}{
  L \triangleright E_1 \ntwkstep{l} E_2
}{L \triangleright \eta[E_1] \ntwkstep{l} \eta[E_2]}

\DefineNtwkRule{Ret}{
  e_1 \localstep e_2
}{L \triangleright \Ret{e_1} \ntwkstep{\iota} \Ret{e_2}}

\DefineNtwkRule{Seq}{\val{V}}
{L \triangleright V \NtwkSeq E \ntwkstep{\iota} E}

\DefineNtwkRule{App}{
  f = \NtwkFun{F}{X}{E} \\
  \val{V}
  }{L \triangleright f~V \ntwkstep{\iota} \subst*{E}{{F}{f}{X}{V}}}

\DefineNtwkRule{TApp}{
  f = \NtwkTFun{F}{\alpha}{E} \\
  }{L \triangleright f~t \ntwkstep{\iota} \subst*{E}{{F}{f}{\alpha}{t}}}

\DefineNtwkRule{UnfoldFold}{\val{V}}
{L \triangleright \NtwkUnfold{(\NtwkFold{V})} \ntwkstep{\iota} V}

\DefineNtwkRule{FstPair}{\val{V_1} \\ \val{V_2}}
{L \triangleright \NtwkFst{(V_1,V_2)} \ntwkstep{\iota} V_1}

\DefineNtwkRule{SndPair}{\val{V_1} \\ \val{V_2}}
{L \triangleright \NtwkSnd{(V_1,V_2)} \ntwkstep{\iota} V_2}

\DefineNtwkRule{CaseInl}{\val{V}}
{L \triangleright \NtwkCase{(\NtwkInl{V})}{X}{E_1}{Y}{E_2} \ntwkstep{\iota} \subst{E_1}{X}{V}}

\DefineNtwkRule{CaseInr}{\val{V}}
{L \triangleright \NtwkCase{(\NtwkInr{V})}{X}{E_1}{Y}{E_2} \ntwkstep{\iota} \subst{E_2}{Y}{V}}

\DefineNtwkRule{LocalCaseInl}{\val{v}}
{L \triangleright \NtwkLocalCase{\Ret{\LocalInl{v}}}{x}{E_1}{y}{E_2} \ntwkstep{\iota} \subst{E_1}{x}{v}}

\DefineNtwkRule{LocalCaseInr}{\val{v}}
{L \triangleright \NtwkLocalCase{\Ret{\LocalInr{v}}}{x}{E_1}{y}{E_2} \ntwkstep{\iota} \subst{E_2}{y}{v}}

\DefineNtwkRule{Let}{\val{v}
}{L \triangleright \NtwkLetIn{x}{\Ret{v}}{E} \ntwkstep{\iota} \subst{E}{x}{v}}

\DefineNtwkRule{TyLet}{
    \val{\say{t}}
  }{L \triangleright \NtwkLetIn{\alpha \knd \kappa}{\Ret{\say{t}}}{E} \ntwkstep{\iota} \subst*{E}{{\alpha}{t}}}

\DefineNtwkRule{Send}{
  \val{v} \\
  \fv{\rho} = \varnothing
}{L \triangleright \SendTo{\Ret{v}}{\rho} \ntwkstep{\RSendNtwk{v}{\rho \setminus \{L\}}} \Ret{v}}

\DefineNtwkRule{Recv}{
  \val{v} \\
  L' \neq L
}{L \triangleright \RecvFrom{L'} \ntwkstep{\RRecvNtwk{L'}{v}} \Ret{v}}

\DefineNtwkRule{Choose}{
  \fv{\rho} = \varnothing
}{L \triangleright \ChooseFor{d}{\rho}{E} \ntwkstep{\RSendNtwk{d}{\rho \setminus \{L\}}} E}

\newtoggle{NAllowLLinebreak}
\DefineNtwkRule{AllowL}{L' \neq L}{
  L \triangleright
  \iftoggle{NAllowLLinebreak}{\AllowChoiceFlat*{L'}{E_1}{{E_2}_\bot}}{\AllowChoice{L'}{E_1}{{E_2}_\bot}}
  \ntwkstep{\RRecvNtwk{L'}{\Left}} E_1
}

\DefineNtwkRule{AllowR}{L' \neq L}{L \triangleright \AllowChoice{L'}{{E_1}_\bot}{E_2} \ntwkstep{\RRecvNtwk{L'}{\Right}} E_2}

\DefineNtwkRule{IAmIn}{L \in \rho}
{L \triangleright \AmIIn{\rho}{E_1}{E_2} \ntwkstep{\iota} E_1}

\DefineNtwkRule{IAmNotIn}{L \notin \rho}
{L \triangleright \AmIIn{\rho}{E_1}{E_2} \ntwkstep{\iota} E_2}

\DefineNtwkRule{PID}{ }
{L \triangleright \PID{\alpha}{E} \ntwkstep{\iota} \subst{E}{\alpha}{L}}

\newtoggle{NForkLinebreak}
\DefineNtwkRule{Fork}{
  {\def\arraystretch{1.1}
  \begin{array}{@{}l@{}}
    E_1' = \subst*{E_1}{{\alpha}{L'}{x}{\say{L'}}} \\
    E_2' = \subst*{E_2}{{\alpha}{L'}{x}{\say{L'}}}
  \end{array}} \\
  {\def\arraystretch{1.1}
  \begin{array}{@{}c@{}}
    \fv{E_1'} = \varnothing
  \end{array}}
}{
  L \triangleright
  \iftoggle{NForkLinebreak}{\NtwkFork{(\alpha,x)}{E_1}{E_2}}{\NtwkFork{(\alpha,x)}{E_1}{E_2}}
  \ntwkstep{\RForkNtwk{L'}{E_1'}}
  E_2'
}

\DefineNtwkRule{Exit}{ }{
  L \triangleright \NtwkExit \ntwkstep{\RExitNtwk} \NtwkUnit
}

\DefineChorRule{Ctx}{
  \conf*{C} \step[R] \conf{C'}{\Omega'}
}{\conf*{\eta[C]} \step[\ctxredex{\eta}] \conf{\eta[C']}{\Omega'}}

\DefineChorRule{Done}{
  e_1 \localstep e_2 \\\\
  \nl{\rho} \subseteq \Omega \\
  \fv{\rho} = \varnothing
}{\conf*{\rho.e_1} \step[\RDone{\rho}{e_1}{e_2}] \conf*{\rho.e_2}}

\DefineChorRule{App}{
  f = \Fun{\rho}{F}{X}{C} \\
  \val{V} \\\\
  \nl{\rho} \subseteq \Omega \\
  \fv{\rho} = \varnothing
}{\conf*{f \appchor{\rho} V} \step[\RApp{\rho}] \conf*{\subst*{C}{{F}{f}{X}{V}}}}

\DefineChorRule{TApp}{
  f = \TFun{\rho'}{F}{\alpha \knd \kappa}{C} \\
  \nl{\rho} \subseteq \Omega \\
  \fv{\rho} = \varnothing
}{\conf*{f \appchor{\rho} t} \step[\RTApp{\rho}] \conf*{\subst*{C}{{F}{f}{\alpha}{t}}}}

\DefineChorRule{UnfoldFold}{
  \val{V} \\
  \nl{\rho} \subseteq \Omega \\
  \fv{\rho} = \varnothing
}{\conf*{\Unfold{\rho}{(\Fold{\rho}{V})}} \step[\RUnfoldFold{\rho}] \conf*{V}}

\DefineChorRule{FstPair}{
  \val{V_1} \\
  \val{V_2} \\
  \nl{\rho} \subseteq \Omega \\
 \fv{\rho} = \varnothing
}{\conf*{\Fst{\rho}{(V_1,V_2)_\rho}} \step[\RFstPair{\rho}] \conf*{V_1}}

\DefineChorRule{SndPair}{
  \val{V_1} \\
  \val{V_2} \\
  \nl{\rho} \subseteq \Omega \\
  \fv{\rho} = \varnothing
}{\conf*{\Snd{\rho}{(V_1,V_2)_\rho}} \step[\RSndPair{\rho}] \conf*{V_2}}

\DefineChorRule{CaseInl}{
  \val{V} \\
  \nl{\rho} \subseteq \Omega \\
  \fv{\rho} = \varnothing
}{\conf*{\Case*{\rho}{(\Inl{\rho}{V})}{X}{C_1}{Y}{C_2}} \step[\RCaseInl{\rho}] \conf*{\subst{C_1}{X}{V}}}

\DefineChorRule{CaseInr}{
  \val{V} \\
  \nl{\rho} \subseteq \Omega \\
  \fv{\rho} = \varnothing
}{\conf*{\Case*{\rho}{(\Inr{\rho}{V})}{X}{C_1}{Y}{C_2}} \step[\RCaseInr{\rho}] \conf*{\subst{C_2}{Y}{V}}}

\DefineChorRule{LocalCaseInl}{
  \val{v} \\
  \getCase(v) = \metaInl(v_1) \\
  \nl{\rho} \subseteq \Omega \\
  \fv{\rho} = \varnothing
}{\conf*{\LocalCase*{\rho}{\rho.v}{x}{C_1}{y}{C_2}} \step[\RLocalCaseInl{\rho}] \conf*{\hsubst{C_1}{\rho}{x}{v_1}}}

\DefineChorRule{LocalCaseInr}{
  \val{v} \\
  \getCase(v) = \metaInr(v_2) \\
  \nl{\rho} \subseteq \Omega \\
  \fv{\rho} = \varnothing
}{\conf*{\LocalCase*{\rho}{\rho.v}{x}{C_1}{y}{C_2}} \step[\RLocalCaseInr{\rho}] \conf*{\hsubst{C_2}{\rho}{y}{v_2}}}

\DefineChorRule{LetV}{
  \val{v} \\
  \nl{\rho} \subseteq \Omega \\
  \fv{\rho} = \varnothing
}{\conf*{\LetIn{\rho.x}{\rho'.v}{C}} \step[\RLet{\rho}{v}] \conf*{\hsubst{C}{\rho}{x}{v}}}

\DefineChorRule{LetI}{
  \conf*{C_2} \step[R] \conf{C_2'}{\Omega'} \\
  \fv{\rho} = \varnothing \\
  \cloc{C_1} \cap \rloc{R} = \varnothing \\
  \rho \cap \rloc{R} = \varnothing
}{\conf*{\LetIn{\rho.x}{C_1}{C_2}} \step[R] \conf{\LetIn{\rho.x}{C_1}{C_2'}}{\Omega'}}

\newtoggle{TyLetVLinebreak}
\DefineChorRule{TyLetV}{
  \val{\say{t}} \\
  \nl{\rho} \subseteq \Omega \\
  \fv{\rho} = \varnothing
}{
  {\begin{array}{@{}l@{}}
    \conf*{\iftoggle{TyLetVLinebreak}
      {\LetIn*{\rho.\alpha \knd \kappa}{\rho'.\say{t}}{C}}
      {\LetIn{\rho.\alpha \knd \kappa}{\rho'.\say{t}}{C}}}
    \step[\RLet{\rho}{t}] \conf*{\subst*{C}{{\alpha}{t}}}
  \end{array}}
}

\DefineChorRule{TyLetI}{
  \conf*{C_2} \step[R] \conf{C_2'}{\Omega'} \\\\
  \cloc{C_1} \cap \rloc{R} = \varnothing \\
  \rho \cap \rloc{R} = \varnothing \\
  \fv{\rho} = \varnothing
}{\conf*{\LetIn{\rho.\alpha \knd \kappa}{C_1}{C_2}} \step[R] \conf{\LetIn{\rho.\alpha \knd \kappa}{C_1}{C_2'}}{\Omega'}}

\DefineChorRule{SendV}{
  \val{v} \\
  L_1 \in \rho_1 \\
  L_1 \in \Omega \\
  \nl{\rho_2} \subseteq \Omega \\
  \fv{\rho_2} = \varnothing
}{\conf*{\rho_1.v \ChorSend[L_1] \rho_2} \step[\RSendV{L_1}{v}{\rho_2}] \conf*{(\rho_1 \cup \rho_2).v}}

\DefineChorRule{Sync}{
  L \in \Omega \\
  \nl{\rho} \subseteq \Omega \\
  \fv{\rho} = \varnothing
}{\conf*{\syncs{L}{d}{\rho} \seq C} \step[\RSendV{L}{d}{\rho}] \conf*{C}}

\DefineChorRule{SyncI}{
  \conf*{C} \step[R] \conf{C'}{\Omega'} \\\\
  \ell \notin \rloc{R} \\
  \rho \cap \rloc{R} = \varnothing \\
  \fv{\rho} = \fv{\ell} = \varnothing
}{\conf*{\syncs{\ell}{d}{\rho} \seq C} \step[R] \conf{\syncs{\ell}{d}{\rho} \seq C'}{\Omega'}}

\DefineChorRule{CaseI}{
  \conf*{C_1} \step[R] \conf{C_1'}{\Omega'} \\
  \conf*{C_2} \step[R] \conf{C_2'}{\Omega'} \\\\
  \cloc{C} \cap \rloc{R} = \varnothing \\
  \rho \cap \rloc{R} = \varnothing \\
  \fv{\rho} = \varnothing
}{\conf*{\Case{\rho}{C}{X}{C_1}{Y}{C_2} } \step[R] \conf{\Case{\rho}{C}{X}{C_1'}{Y}{C_2'} }{\Omega'}}

\DefineChorRule{LocalCaseI}{
  \conf*{C_1} \step[R] \conf{C_1'}{\Omega'} \\
  \conf*{C_2} \step[R] \conf{C_2'}{\Omega'} \\\\
  \fv{\rho} = \varnothing \\
  \cloc{C} \cap \rloc{R} = \varnothing \\
  \rho \cap \rloc{R} = \varnothing
}{\conf*{\LocalCase{\rho}{C}{x}{C_1}{y}{C_2} } \step[R] \conf{\LocalCase{\rho}{C}{x}{C_1'}{y}{C_2'} }{\Omega'}}

\DefineChorRule{AppI}{
  \conf*{C_2} \step[R] \conf{C_2'}{\Omega'} \\\\
  \cloc{C_1} \cap \rloc{R} = \varnothing
}{\conf*{C_1 \appchor{\rho} C_2} \step[R] \conf{C_1 \appchor{\rho} C_2'}{\Omega'}}

\DefineChorRule{PairI}{
  \conf*{C_2} \step[R] \conf{C_2'}{\Omega'} \\\\
  \cloc{C_1} \cap \rloc{R} = \varnothing
}{\conf*{(C_1,C_2)_\rho} \step[R] \conf{(C_1,C_2')_\rho}{\Omega'}}

\newtoggle{CForkLinebreak}
\DefineChorRule{Fork}{
  L'~\text{globally fresh} \\
  \fv{C'} = \varnothing \\
  L \in \Omega \\\\
  C' = \subst*{C}{{\alpha}{L'}{x}{\say{L'}}}
}{{\begin{array}{@{}l@{}}
  \conf*{\Fork{(\alpha,x)}{L}{C}}
  \iftoggle{CForkLinebreak}{ \\ \quad }{}
  \step[\RFork{L}{L'}]
  \conf{\KillAfter{L'}{C'}}{\Omega \cup \{L'\}}
\end{array}}}

\DefineChorRule{ForkI}{
  \conf*{C} \step[R] \conf{C'}{\Omega'} \\
  L \notin \rloc{R} \\
}{\conf*{\Fork{(\alpha,x)}{L}{C}} \step[R] \conf{\Fork{(\alpha,x)}{L}{C'}}{\Omega'}}

\DefineChorRule{Kill}{
  \val{V} \\
  L \in \Omega
}{\conf*{\KillAfter{L}{V}} \step[\RKill{L}] \conf{V}{\Omega \setminus \{L\}}}

\DefineChorRule{KillI}{
  L \notin \cloc{C} \\
  L \in \Omega
}{\conf*{\KillAfter{L}{C}} \step[\RKill{L}] \conf{C}{\Omega \setminus \{L\}}}

\DefineKindingRule{Var}{
  \provesCol \Gamma \\
  \alpha \knd \kappa \in \Gamma
}{\chorkinded{\Gamma}{\alpha}{\kappa}}

\DefineKindingRule{SubVar}{
  \provesCol \Gamma \\
  \alpha \knd \finsetKnd \in \Gamma
}{\chorkinded{\Gamma}{\alpha}{\setKnd}}

\DefineKindingRule{Loc}{
  \Location \in \Locations
}{\chorkinded{\Gamma}{\Location}{\locKnd}}

\DefineKindingRule{Sng}{
  \chorkinded{\Gamma}{\ell}{\locKnd}
}{\chorkinded{\Gamma}{\{\ell\}}{\setKnd}}

\DefineKindingRule{SngFin}{
  \chorkinded{\Gamma}{\ell}{\locKnd}
}{\chorkinded{\Gamma}{\{\ell\}}{\finsetKnd}}

\DefineKindingRule{Union}{
  \chorkinded{\Gamma}{\rho_1}{\setKnd} \\
  \chorkinded{\Gamma}{\rho_2}{\setKnd}
}{\chorkinded{\Gamma}{\rho_1 \cup \rho_2}{\setKnd}}

\DefineKindingRule{UnionFin}{
  \chorkinded{\Gamma}{\rho_1}{\finsetKnd} \\
  \chorkinded{\Gamma}{\rho_2}{\finsetKnd}
}{\chorkinded{\Gamma}{\rho_1 \cup \rho_2}{\finsetKnd}}

\DefineKindingRule{AnySet}{
  \provesCol \Gamma
}{\chorkinded{\Gamma}{\anyLoc}{\setKnd}}

\DefineKindingRule{Local}{
  \provesCol \Gamma \\
  \localkinded{\Gamma}{t_e}
}{\chorkinded{\Gamma}{t_e}{*_e}}

\DefineKindingRule{At}{
  \chorkinded{\Gamma}{t_e}{*_e} \\
  \chorkinded{\Gamma}{\rho}{\finsetKnd}
}{\chorkinded{\Gamma}{t_e @ \rho}{\tyknd{\rho}}}

\DefineKindingRule{Arrow}{
  \chorkinded{\Gamma}{\rho}{\setKnd} \\\\
  \chorkinded{\Gamma}{\tau_1}{\tyknd{\rho_1}} \\
  \chorkinded{\Gamma}{\tau_2}{\tyknd{\rho_2}}
}{\chorkinded{\Gamma}{\tau_1 \arr{\rho} \tau_2}{\tyknd{\rho_1 \cup \rho_2 \cup \rho}}}

\DefineKindingRule{Prod}{
  \chorkinded{\Gamma}{\tau_1}{\tyknd{\rho_1}} \\
  \chorkinded{\Gamma}{\tau_2}{\tyknd{\rho_2}}
}{\chorkinded{\Gamma}{\tau_1 \times \tau_2}{\tyknd{\rho_1 \cup \rho_2}}}

\DefineKindingRule{Sum}{
  \chorkinded{\Gamma}{\tau_1}{\tyknd{\rho_1}} \\
  \chorkinded{\Gamma}{\tau_2}{\tyknd{\rho_2}} \\\\
  \chorkinded{\Gamma}{\rho}{\setKnd} \\
  \rho_1 \cup \rho_2 \subseteq \rho
}{\chorkinded{\Gamma}{\tau_1 +_\rho \tau_2}{\tyknd{\rho}}}

\DefineKindingRule{Rec}{
  \chorkinded{\Gamma, \alpha \knd \tyknd{\rho}}{\tau}{\tyknd{\rho}}
}{\chorkinded{\Gamma}{\mu_\rho \alpha \ldotp \tau}{\tyknd{\rho}}}

\DefineKindingRule{AllLoc}{
  \kappa \in \{\locKnd, \setKnd, \finsetKnd\} \\\\
  \chorkinded{\Gamma, \alpha \knd \kappa}{\tau}{\tyknd{\rho_\tau}} \\
  \chorkinded{\Gamma, \alpha \knd \kappa}{\rho}{\setKnd} \\
}{\chorkinded{\Gamma}{\allty{\alpha \knd \kappa}{\rho}{\tau}}{\tyknd{\anyLoc}}}

\DefineKindingRule{All}{
  \kappa \in \{*_e, \tyknd{\rho'}\} \\\\
  \chorkinded{\Gamma, \alpha \knd \kappa}{\tau}{\tyknd{\rho_\tau}} \\
  \chorkinded{\Gamma}{\rho}{\setKnd}
}{\chorkinded{\Gamma}{\allty{\alpha \knd \kappa}{\rho}{\tau}}{\tyknd{\rho_\tau \cup \rho}}}

  \DefineTypingRule{Var}{
    \chorkinded{\ctx}{\tau}{\tyknd{\rho}} \\
    X \ty \tau \in \ctx
  }{\chortyped{\ctx}{X_\rho}{\tau}{\varnothing}}

\DefineTypingRule{Done}{
  \proves \ctx \\
  \chorkinded{\ctx}{\rho}{\setKnd} \\
  \localtyped{\proj{\ctx}{\rho}}{e}{t_e} \\\\
  \rho' = \text{if}~\val{e}~\text{then}~\varnothing~\text{else}~\rho
}{\chortyped{\ctx}{\rho.e}{t_e @ \rho}{\rho'}}

\DefineTypingRule{Fun}{
  \chorkinded{\ctx}{\tau_1}{\tyknd{\rho_a}} \\
  \chorkinded{\ctx}{\tau_2}{\tyknd{\rho_b}} \\\\
  \rho' = \rho_a \cup \rho_b \cup \rho \\\\
  \chortyped{\ctx, F \ty \tau_1 \arr{\rho} \tau_2, X \ty \tau_1}{C}{\tau_2}{\rho}
}{\chortyped{\ctx}{\Fun{\rho'}{F}{X}{C}}{\tau_1 \arr{\rho} \tau_2}{\varnothing}}

\DefineTypingRule{App}{
  \chorkinded{\ctx}{\tau_1}{\tyknd{\rho_a}} \\
  \chorkinded{\ctx}{\tau_2}{\tyknd{\rho_b}} \\\\
  \rho' = \rho_a \cup \rho_b \cup \rho \\\\
  \chortyped{\ctx}{C_1}{\tau_1 \arr{\rho} \tau_2}{\rho_1} \\
  \chortyped{\ctx}{C_2}{\tau_1}{\rho_2}
}{\chortyped{\ctx}{C_1 \appchor{\rho'} C_2}{\tau_2}{\rho_1 \cup \rho_2 \cup \rho'}}

\DefineTypingRule{TFunLoc}{
  \kappa \in \{\locKnd, \finsetKnd\} \\\\
  \chortyped{\ctx, F \ty \forall \alpha \knd \kappa [\rho] \ldotp \tau, \alpha \knd \kappa}{C}{\tau}{\rho}
}{\chortyped{\ctx}{\TFunLoc{F}{\alpha \knd \kappa}{C}}{\forall \alpha \knd \kappa [\rho] \ldotp \tau}{\varnothing}}

\DefineTypingRule{TAppLoc}{
  \kappa \in \{\locKnd, \finsetKnd\} \\
  \chortyped{\ctx}{C}{\forall \alpha \knd \kappa [\rho] \ldotp \tau}{\rho_1} \\
  \chorkinded{\ctx}{t}{\kappa} \\\\
  \chorkinded{\ctx}{\subst{\tau}{\alpha}{t}}{\tyknd{\rho_\tau}} \\
  \rho' = \rho_\tau \cup \subst{\rho}{\alpha}{t}
}{\chortyped{\ctx}{C \appchor{\rho'} t}{\subst{\tau}{\alpha}{t}}{\rho_1 \cup \rho'}}

\DefineTypingRule{TFun}{
  \kappa \in \{*_e, \tyknd{\rho''}\} \\\\
  \chortyped{\ctx, F \ty \forall \alpha \knd \kappa [\rho] \ldotp \tau, \alpha \knd \kappa}{C}{\tau}{\rho} \\\\
  \chorkinded{\ctx, \alpha \knd \kappa}{\tau}{\tyknd{\rho_\tau}} \\
  \rho' = \rho_\tau \cup \rho
}{\chortyped{\ctx}{\TFun{\rho'}{F}{\alpha \knd \kappa}{C}}{\forall \alpha \knd \kappa [\rho] \ldotp \tau}{\varnothing}}

\DefineTypingRule{TApp}{
  \kappa \in \{*_e, \tyknd{\rho''}\} \\\\
  \chortyped{\ctx}{C}{\forall \alpha \knd \kappa [\rho] \ldotp \tau}{\rho_1} \\
  \chorkinded{\ctx}{t}{\kappa} \\\\
  \chorkinded{\ctx}{\subst{\tau}{\alpha}{t}}{\tyknd{\rho_\tau}} \\
  \rho' = \rho_\tau \cup \rho
}{\chortyped{\ctx}{C \appchor{\rho'} t}{\subst{\tau}{\alpha}{t}}{\rho_1 \cup \rho'}}

\DefineTypingRule{Pair}{
  \chortyped{\ctx}{C_1}{\tau_1}{\rho_1} \\
  \chortyped{\ctx}{C_2}{\tau_2}{\rho_2} \\\\
  \chorkinded{\ctx}{\tau_1}{\tyknd{\rho_a}} \\
  \chorkinded{\ctx}{\tau_2}{\tyknd{\rho_b}} \\\\
  \rho = \rho_a \cup \rho_b
}{\chortyped{\ctx}{(C_1,C_2)_\rho}{\tau_1 \times \tau_2}{\rho_1 \cup \rho_2}}

\DefineTypingRule{Fst}{
  \chortyped{\ctx}{C}{\tau_1 \times \tau_2}{\rho'} \\\\
  \chorkinded{\ctx}{\tau_1}{\tyknd{\rho_a}} \\
  \chorkinded{\ctx}{\tau_2}{\tyknd{\rho_b}} \\\\
  \rho = \rho_a \cup \rho_b
}{\chortyped{\ctx}{\Fst{\rho}{C}}{\tau_1}{\rho \cup \rho'}}

\DefineTypingRule{Snd}{
  \chortyped{\ctx}{C}{\tau_1 \times \tau_2}{\rho'} \\\\
  \chorkinded{\ctx}{\tau_1}{\tyknd{\rho_a}} \\
  \chorkinded{\ctx}{\tau_2}{\tyknd{\rho_b}} \\\\
  \rho = \rho_a \cup \rho_b
}{\chortyped{\ctx}{\Snd{\rho}{C}}{\tau_2}{\rho \cup \rho'}}

\DefineTypingRule{Inl}{
  \chortyped{\ctx}{C}{\tau_1}{\rho} \\
  \chorkinded{\ctx}{\rho'}{\setKnd} \\\\
  \chorkinded{\ctx}{\tau_1}{\tyknd{\rho_a}} \\
  \chorkinded{\ctx}{\tau_2}{\tyknd{\rho_b}} \\\\
  \rho_a \cup \rho_b \subseteq \rho'
}{\chortyped{\ctx}{\Inl{\rho'}{C}}{\tau_1 +_{\rho'} \tau_2}{\rho}}

\DefineTypingRule{Inr}{
  \chortyped{\ctx}{C}{\tau_2}{\rho} \\
  \chorkinded{\ctx}{\rho'}{\setKnd} \\\\
  \chorkinded{\ctx}{\tau_1}{\tyknd{\rho_a}} \\
  \chorkinded{\ctx}{\tau_2}{\tyknd{\rho_b}} \\\\
  \rho_a \cup \rho_b \subseteq \rho'
}{\chortyped{\ctx}{\Inr{\rho'}{C}}{\tau_1 +_{\rho'} \tau_2}{\rho}}

\DefineTypingRule{Case}{
  \chortyped{\ctx}{C}{\tau_1 +_{\rho'} \tau_2}{\rho} \\\\
  \chortyped{\ctx, X \ty \tau_1}{C_1}{\tau}{\rho_1} \\
  \chortyped{\ctx, Y \ty \tau_2}{C_2}{\tau}{\rho_2}
}{\chortyped{\ctx}{\Case*{\rho'}{C}{X}{C_1}{Y}{C_2}}{\tau}{\rho \cup \rho' \cup \rho_1 \cup \rho_2}}

\DefineTypingRule{LocalCase}{
  \isSum{s}{t_1}{t_2} \\
  \chortyped{\ctx}{C}{s @ \rho'}{\rho} \\\\
  \chortyped{\ctx, \rho'.x \ty t_1}{C_1}{\tau}{\rho_1} \\
  \chortyped{\ctx, \rho'.y \ty t_2}{C_2}{\tau}{\rho_2}
}{\chortyped{\ctx}{\LocalCase*{\rho'}{C}{x}{C_1}{y}{C_2}}{\tau}{\rho \cup \rho' \cup \rho_1 \cup \rho_2}}

\DefineTypingRule{Fold}{
  \chortyped{\ctx}{C}{\subst{\tau}{\alpha}{\mu_\rho \alpha \ldotp \tau}}{\rho'} \\
  \chorkinded{\ctx}{\mu_\rho \alpha \ldotp \tau}{\tyknd{\rho}}
}{\chortyped{\ctx}{\Fold{\rho}{C}}{\mu_\rho \alpha \ldotp \tau}{\rho'}}

\DefineTypingRule{Unfold}{
  \chortyped{\ctx}{C}{\mu_\rho \alpha \ldotp \tau}{\rho'} \\
  \chorkinded{\ctx}{\rho}{\setKnd}
}{\chortyped{\ctx}{\Unfold{\rho}{C}}{\subst{\tau}{\alpha}{\mu_\rho \alpha \ldotp \tau}}{\rho \cup \rho'}}

\DefineTypingRule{LetLocal}{
  \rho_1 \subseteq \rho_2 \\
  \chorkinded{\ctx}{\rho_1}{\setKnd} \\\\
  \chortyped{\ctx}{C_1}{t_e @ \rho_2}{\rho} \\
  \chortyped{\ctx, \rho_1.x \ty t_e}{C_2}{\tau}{\rho'}
}{\chortyped{\ctx}{\LetIn{\rho_1.x}{C_1}{C_2}}{\tau}{\rho \cup \rho' \cup \rho_1}}

\DefineTypingRule{LetLoc}{
  \chortyped{\ctx}{C_1}{\Loc_{\rho_1} @ \rho_3}{\rho} \\
  \rho_1 \subseteq \rho_2 \subseteq \rho_3 \\\\
  \chorkinded{\ctx}{\rho_2}{\finsetKnd} \\\\
  \chorkinded{\ctx}{\tau}{\tyknd{\rho_\tau}} \\
  \chortyped{\ctx, \alpha \knd \locKnd}{C_2}{\tau}{\rho'}
}{\chortyped{\ctx}{\LetIn{\rho_2.\alpha \knd \locKnd}{C_1}{C_2}}{\tau}{\rho \cup (\rho' \setminus \{\alpha\}) \cup \rho_2}}

\DefineTypingRule{LetLocSet}{
  \chortyped{\ctx}{C_1}{\LocSet_{\rho_1} @ \rho_3}{\rho} \\
  \rho_1 \subseteq \rho_2 \subseteq \rho_3 \\\\
  \chorkinded{\ctx}{\rho_2}{\finsetKnd} \\\\
  \chorkinded{\ctx}{\tau}{\tyknd{\rho_\tau}} \\
  \chortyped{\ctx, \alpha \knd \finsetKnd}{C_2}{\tau}{\rho'}
}{\chortyped{\ctx}{\LetIn{\rho_2.\alpha \knd \finsetKnd}{C_1}{C_2}}{\tau}{\rho \cup (\rho' \setminus \alpha) \cup \rho_2}}

\DefineTypingRule{Send}{
  \chortyped{\ctx}{C}{t_e @ \rho_1}{\rho} \\\\
  \chorkinded{\ctx}{\ell_1}{\locKnd} \\
  \ell_1 \in \rho_1 \\
  \chorkinded{\ctx}{\rho_2}{\finsetKnd}
}{\chortyped{\ctx}{C \ChorSend[\ell_1] \rho_2}{t_e @ (\rho_1 \cup \rho_2)}{\rho \cup \{\ell_1\} \cup \rho_2}}

\DefineTypingRule{Sync}{
  \chortyped{\ctx}{C}{\tau}{\rho'} \\\\
  \chorkinded{\ctx}{\ell}{\locKnd} \\
  \chorkinded{\ctx}{\rho}{\finsetKnd} \\
}{\chortyped{\ctx}{\syncs{\ell}{d}{\rho} \seq C}{\tau}{\{\ell\} \cup \rho \cup \rho'}}

\DefineTypingRule{Fork}{
  \chorkinded{\ctx}{\ell}{\locKnd} \\
  \chorkinded{\ctx}{\tau}{\tyknd{\rho_\tau}} \\\\
  \chortyped{\ctx, \alpha \knd \locKnd, \{\ell,\alpha\}.x \ty \Loc_{\alpha}}{C}{\tau}{\rho}
}{\chortyped{\ctx}{\Fork{(\alpha,x)}{\ell}{C}}{\tau}{\{\ell\} \cup (\rho \setminus \{\alpha\})}}

\DefineSpawnRule{Var}{
  \chorkinded{\ctx}{\tau}{\tyknd{\rho}} \\
  X \ty \tau \in \ctx
}{\chortypedplus{\ctx}{X_\rho}{\tau}{\varnothing}}

\DefineSpawnRule{Done}{
  \proves \ctx \\
  \chorkinded{\ctx}{\rho}{\setKnd} \\
  \localtyped{\proj{\ctx}{\rho}}{e}{t_e} \\\\
  \rho' = \text{if}~\val{e}~\text{then}~\varnothing~\text{else}~\rho
}{\chortypedplus{\ctx}{\rho.e}{t_e @ \rho}{\rho'}}

\DefineSpawnRule{Fun}{
  \chortypedplus{\ctx, F \ty \tau_1 \arr{\rho} \tau_2, X \ty \tau_1}{C}{\tau_2}{\rho} \\\\
  \chorkinded{\ctx}{\tau_1}{\tyknd{\rho_a}} \\
  \chorkinded{\ctx}{\tau_2}{\tyknd{\rho_b}} \\\\
  \rho' = \rho_a \cup \rho_b \cup \rho \\
  \spawnedlocs{C} = \varnothing
}{\chortypedplus{\ctx}{\Fun{\rho'}{F}{X}{C}}{\tau_1 \arr{\rho} \tau_2}{\varnothing}}

\DefineSpawnRule{App}{
  \chortypedplus{\ctx}{C_1}{\tau_1 \arr{\rho} \tau_2}{\rho_1} \\
  \chortypedplus{\ctx}{C_2}{\tau_1}{\rho_2} \\\\
  \chorkinded{\ctx}{\tau_1}{\tyknd{\rho_a}} \\
  \chorkinded{\ctx}{\tau_2}{\tyknd{\rho_b}} \\\\
  \rho' = \rho_a \cup \rho_b \cup \rho \\\\
  \nl{\rho_1} \cap \spawnedlocs{C_2} = \varnothing \\\\
  \nl{\rho_2} \cap \spawnedlocs{C_1} = \varnothing \\\\
  \namedlocs{\rho'} \cap (\spawnedlocs{C_1} \cup \spawnedlocs{C_2}) = \varnothing
}{\chortypedplus{\ctx}{C_1 \appchor{\rho'} C_2}{\tau_2}{\rho_1 \cup \rho_2 \cup \rho'}}

\DefineSpawnRule{TFun}{
  \kappa \in \{*_e, \tyknd{\rho''}\} \\\\
  \chortypedplus{\ctx, F \ty \forall \alpha \knd \kappa [\rho] \ldotp \tau, \alpha \knd \kappa}{C}{\tau}{\rho} \\\\
  \chorkinded{\ctx, \alpha \knd \kappa}{\tau}{\tyknd{\rho_\tau}} \\
  \rho' = \rho_\tau \cup \rho \\\\
  \spawnedlocs{C} = \varnothing
}{\chortypedplus{\ctx}{\TFun{\rho'}{F}{\alpha \knd \kappa}{C}}{\forall \alpha \knd \kappa [\rho] \ldotp \tau}{\varnothing}}

\DefineSpawnRule{TApp}{
  \kappa \in \{*_e, \tyknd{\rho''}\} \\\\
  \chortypedplus{\ctx}{C}{\forall \alpha \knd \kappa [\rho] \ldotp \tau}{\rho_1} \\
  \chorkinded{\ctx}{t}{\kappa} \\\\
  \chorkinded{\ctx}{\subst{\tau}{\alpha}{t}}{\tyknd{\rho_\tau}} \\
  \rho' = \rho_\tau \cup \rho \\\\
  \namedlocs{\rho'} \cap \spawnedlocs{C} = \varnothing
}{\chortypedplus{\ctx}{C \appchor{\rho'} t}{\subst{\tau}{\alpha}{t}}{\rho_1 \cup \rho'}}

\DefineSpawnRule{TFunLoc}{
  \kappa \in \{\locKnd, \setKnd, \finsetKnd\} \\\\
  \chortypedplus{\ctx, F \ty \forall \alpha \knd \kappa [\rho] \ldotp \tau, \alpha \knd \kappa}{C}{\tau}{\rho} \\
  \spawnedlocs{C} = \varnothing
}{\chortypedplus{\ctx}{\TFunLoc{F}{\alpha \knd \kappa}{C}}{\forall \alpha \knd \kappa [\rho] \ldotp \tau}{\varnothing}}

\DefineSpawnRule{TAppLoc}{
  \kappa \in \{\locKnd, \setKnd, \finsetKnd\} \\\\
  \chortypedplus{\ctx}{C}{\forall \alpha \knd \kappa [\rho] \ldotp \tau}{\rho_1} \\
  \chorkinded{\ctx}{t}{\kappa} \\\\
  \kappa = \setKnd \Rightarrow \chorkinded{\ctx}{t}{\finsetKnd} \\\\
  \chorkinded{\ctx}{\subst{\tau}{\alpha}{t}}{\tyknd{\rho_\tau}} \\
  \rho' = \rho_\tau \cup \subst{\rho}{\alpha}{t} \\\\
  \namedlocs{\rho'} \cap \spawnedlocs{C} = \varnothing
}{\chortypedplus{\ctx}{C \appchor{\rho'} t}{\subst{\tau}{\alpha}{t}}{\rho_1 \cup \rho'}}

\DefineSpawnRule{Pair}{
  \chortypedplus{\ctx}{C_1}{\tau_1}{\rho_1} \\
  \chortypedplus{\ctx}{C_2}{\tau_2}{\rho_2} \\\\
  \chorkinded{\ctx}{\tau_1}{\tyknd{\rho_a}} \\
  \chorkinded{\ctx}{\tau_2}{\tyknd{\rho_b}} \\\\
  \rho = \rho_a \cup \rho_b \\\\
  \nl{\rho_1} \cap \spawnedlocs{C_2} = \varnothing \\\\
  \nl{\rho_2} \cap \spawnedlocs{C_1} = \varnothing
}{\chortypedplus{\ctx}{(C_1,C_2)_\rho}{\tau_1 \times \tau_2}{\rho_1 \cup \rho_2}}

\DefineSpawnRule{Fst}{
  \chortypedplus{\ctx}{C}{\tau_1 \times \tau_2}{\rho'} \\\\
  \chorkinded{\ctx}{\tau_1}{\tyknd{\rho_a}} \\
  \chorkinded{\ctx}{\tau_2}{\tyknd{\rho_b}} \\\\
  \rho = \rho_a \cup \rho_b \\
  \namedlocs{\rho} \cap \spawnedlocs{C} = \varnothing
}{\chortypedplus{\ctx}{\Fst{\rho}{C}}{\tau_1}{\rho \cup \rho'}}

\DefineSpawnRule{Snd}{
  \chortypedplus{\ctx}{C}{\tau_1 \times \tau_2}{\rho'} \\\\
  \chorkinded{\ctx}{\tau_1}{\tyknd{\rho_a}} \\
  \chorkinded{\ctx}{\tau_2}{\tyknd{\rho_b}} \\\\
  \rho = \rho_a \cup \rho_b \\
  \namedlocs{\rho} \cap \spawnedlocs{C} = \varnothing
}{\chortypedplus{\ctx}{\Snd{\rho}{C}}{\tau_2}{\rho \cup \rho'}}

\DefineSpawnRule{Inl}{
  \chortypedplus{\ctx}{C}{\tau_1}{\rho} \\
  \chorkinded{\ctx}{\rho'}{\setKnd} \\\\
  \chorkinded{\ctx}{\tau_1}{\tyknd{\rho_a}} \\
  \chorkinded{\ctx}{\tau_2}{\tyknd{\rho_b}} \\\\
  \rho_a \cup \rho_b \subseteq \rho'
}{\chortypedplus{\ctx}{\Inl{\rho'}{C}}{\tau_1 +_{\rho'} \tau_2}{\rho}}

\DefineSpawnRule{Inr}{
  \chortypedplus{\ctx}{C}{\tau_2}{\rho} \\
  \chorkinded{\ctx}{\rho'}{\setKnd} \\\\
  \chorkinded{\ctx}{\tau_1}{\tyknd{\rho_a}} \\
  \chorkinded{\ctx}{\tau_2}{\tyknd{\rho_b}} \\\\
  \rho_a \cup \rho_b \subseteq \rho'
}{\chortypedplus{\ctx}{\Inr{\rho'}{C}}{\tau_1 +_{\rho'} \tau_2}{\rho}}

\DefineSpawnRule{Case}{
  \chortypedplus{\ctx}{C}{\tau_1 +_{\rho'} \tau_2}{\rho} \\\\
  \chortypedplus{\ctx, X \ty \tau_1}{C_1}{\tau}{\rho_1} \\
  \chortypedplus{\ctx, Y \ty \tau_2}{C_2}{\tau}{\rho_2} \\\\
  \spawnedlocs{C_1} = \spawnedlocs{C_2} \\\\
  \nl{\rho'} \cap (\spawnedlocs{C} \cup \spawnedlocs{C_1}) = \varnothing \\\\
  \nl{\rho} \cap \spawnedlocs{C_1} = \varnothing \\\\
  (\nl{\rho_1} \cup \nl{\rho_2}) \cap \spawnedlocs{C} = \varnothing
}{\chortypedplus{\ctx}{\Case*{\rho'}{C}{X}{C_1}{Y}{C_2}}{\tau}{\rho \cup \rho' \cup \rho_1 \cup \rho_2}}

\DefineSpawnRule{LocalCase}{
  \isSum{s}{t_1}{t_2} \\
  \chortypedplus{\ctx}{C}{s @ \rho'}{\rho} \\\\
  \chortypedplus{\ctx, \rho'.x \ty t_1}{C_1}{\tau}{\rho_1} \\
  \chortypedplus{\ctx, \rho'.y \ty t_2}{C_2}{\tau}{\rho_2} \\\\
  \spawnedlocs{C_1} = \spawnedlocs{C_2} \\\\
  \nl{\rho'} \cap (\spawnedlocs{C} \cup \spawnedlocs{C_1}) = \varnothing \\\\
  \nl{\rho} \cap \spawnedlocs{C_1} = \varnothing \\\\
  (\nl{\rho_1} \cup \nl{\rho_2}) \cap \spawnedlocs{C} = \varnothing
}{\chortypedplus{\ctx}{\LocalCase*{\rho'}{C}{x}{C_1}{y}{C_2}}{\tau}{\rho \cup \rho' \cup \rho_1 \cup \rho_2}}

\DefineSpawnRule{Fold}{
  \chortypedplus{\ctx}{C}{\subst{\tau}{\alpha}{\mu_\rho \alpha \ldotp \tau}}{\rho'} \\
  \chorkinded{\ctx}{\mu_\rho \alpha \ldotp \tau}{\tyknd{\rho}}
}{\chortypedplus{\ctx}{\Fold{\rho}{C}}{\mu_\rho \alpha \ldotp \tau}{\rho'}}

\DefineSpawnRule{Unfold}{
  \chortypedplus{\ctx}{C}{\mu_\rho \alpha \ldotp \tau}{\rho'} \\
  \chorkinded{\ctx}{\rho}{\setKnd} \\
  \nl{\rho} \cap \spawnedlocs{C} = \varnothing
}{\chortypedplus{\ctx}{\Unfold{\rho}{C}}{\subst{\tau}{\alpha}{\mu_\rho \alpha \ldotp \tau}}{\rho \cup \rho'}}

\DefineSpawnRule{LetLocal}{
  \rho_1 \subseteq \rho_2 \\
  \chorkinded{\ctx}{\rho_1}{\setKnd} \\\\
  \chortypedplus{\ctx}{C_1}{t_e @ \rho_2}{\rho} \\
  \chortypedplus{\ctx, \rho_1.x \ty t_e}{C_2}{\tau}{\rho'} \\\\
  \nl{\rho} \cap \spawnedlocs{C_2} = \varnothing \\
  \nl{\rho'} \cap \spawnedlocs{C_1} = \varnothing \\\\
  \nl{\rho_1} \cap (\spawnedlocs{C_1} \cup \spawnedlocs{C_2}) = \varnothing
}{\chortypedplus{\ctx}{\LetIn{\rho_1.x}{C_1}{C_2}}{\tau}{\rho \cup \rho' \cup \rho_1}}

\DefineSpawnRule{LetLoc}{
  \chortypedplus{\ctx}{C_1}{\Loc_{\rho_1} @ \rho_3}{\rho} \\
  \rho_1 \subseteq \rho_2 \subseteq \rho_3 \\\\
  \chorkinded{\ctx}{\rho_2}{\finsetKnd} \\\\
  \chorkinded{\ctx}{\tau}{\tyknd{\rho_\tau}} \\
  \chortypedplus{\ctx, \alpha \knd \locKnd}{C_2}{\tau}{\rho'} \\\\
  \nl{\rho} \cap \spawnedlocs{C_2} = \varnothing \\\\
  \nl{\rho'} \cap \spawnedlocs{C_1} = \varnothing \\\\
  \nl{\rho_2} \cap (\spawnedlocs{C_1} \cup \spawnedlocs{C_2}) = \varnothing
}{\chortypedplus{\ctx}{\LetIn{\rho_2.\alpha \knd \locKnd}{C_1}{C_2}}{\tau}{\rho \cup (\rho' \setminus \{\alpha\}) \cup \rho_2}}

\DefineSpawnRule{LetLocSet}{
  \chortypedplus{\ctx}{C_1}{\LocSet_{\rho_1} @ \rho_3}{\rho} \\
  \rho_1 \subseteq \rho_2 \subseteq \rho_3 \\\\
  \chorkinded{\ctx}{\rho_2}{\finsetKnd} \\\\
  \chorkinded{\ctx}{\tau}{\tyknd{\rho_\tau}} \\
  \chortypedplus{\ctx, \alpha \knd \finsetKnd}{C_2}{\tau}{\rho'} \\\\
  \nl{\rho} \cap \spawnedlocs{C_2} = \varnothing \\\\
  \nl{\rho'} \cap \spawnedlocs{C_1} = \varnothing \\\\
  \nl{\rho_2} \cap (\spawnedlocs{C_1} \cup \spawnedlocs{C_2}) = \varnothing
}{\chortypedplus{\ctx}{\LetIn{\rho_2.\alpha \knd \finsetKnd}{C_1}{C_2}}{\tau}{\rho \cup (\rho' \setminus \alpha) \cup \rho_2}}

\DefineSpawnRule{Send}{
  \chortypedplus{\ctx}{C}{t_e @ \rho_1}{\rho} \\\\
  \chorkinded{\ctx}{\ell_1}{\locKnd} \\
  \ell_1 \in \rho_1 \\
  \chorkinded{\ctx}{\rho_2}{\finsetKnd} \\\\
  (\nl{\ell_1} \cup \nl{\rho_2}) \cap \spawnedlocs{C} = \varnothing
}{\chortypedplus{\ctx}{C \ChorSend[\ell_1] \rho_2}{t_e @ (\rho_1 \cup \rho_2)}{\rho \cup \{\ell_1\} \cup \rho_2}}

\DefineSpawnRule{Sync}{
  \chortypedplus{\ctx}{C}{\tau}{\rho'} \\\\
  \chorkinded{\ctx}{\ell}{\locKnd} \\
  \chorkinded{\ctx}{\rho}{\finsetKnd} \\\\
  (\nl{\ell} \cup \nl{\rho}) \cap \spawnedlocs{C} = \varnothing
}{\chortypedplus{\ctx}{\syncs{\ell}{d}{\rho} \seq C}{\tau}{\{\ell\} \cup \rho \cup \rho'}}

\DefineSpawnRule{Fork}{
  \chortypedplus{\ctx, \alpha \knd \locKnd, \{\ell,\alpha\}.x \ty \Loc_{\alpha}}{C}{\tau}{\rho} \\\\
  \chorkinded{\ctx}{\ell}{\locKnd} \\
  \chorkinded{\ctx}{\tau}{\tyknd{\rho_\tau}} \\\\
  \nl{\ell} \cap \spawnedlocs{C} = \varnothing
}{\chortypedplus{\ctx}{\Fork{(\alpha,x)}{\ell}{C}}{\tau}{\{\ell\} \cup (\rho \setminus \{\alpha\})}}

\DefineSpawnRule{Kill}{
  \chortypedplus{\ctx}{C}{\tau}{\rho} \\
  L \notin \spawnedlocs{C}
}{\chortypedplus{\ctx}{\KillAfter{L}{C}}{\tau}{\rho \cup \{L\}}}

\title{Step in Tine: Forking Processes in Functional Choreographies}
\date{}

\author{Ashley Samuelson}
\orcid{0009-0001-8800-2590}
\affiliation{
  \institution{University of Wisconsin--Madison}
  \city{Madison}
  \state{Wisconsin}
  \country{USA}
}
\email{ashley.samuelson@wisc.edu}

\author{Zonghao Zou}
\orcid{0009-0003-3456-7603}
\affiliation{
  \institution{University of Wisconsin--Madison}
  \city{Madison}
  \state{Wisconsin}
  \country{USA}
}
\email{zzou39@wisc.edu}

\author{Andrew K. Hirsch}
\orcid{0000-0003-2518-614X}
\affiliation{
  \institution{University at Buffalo, SUNY}
  \city{Buffalo}
  \state{New York}
  \country{USA}
}
\email{akhirsch@buffalo.edu}

\author{Ethan Cecchetti}
\orcid{0000-0001-7900-8328}
\affiliation{
  \institution{University of Wisconsin--Madison}
  \city{Madison}
  \state{Wisconsin}
  \country{USA}
}
\email{cecchetti@wisc.edu}

\begin{document}

\begin{abstract}
  Choreographic programming allows a concurrent programmer to write one centralized program and compile it to individual programs for each system participant.
This approach simplifies critical properties like deadlock freedom,
but it complicates \emph{forking new processes}, a core concurrent primitive.
This work addresses that gap with the choreographic fork calculus~\lamfork, the first functional choreographic language with process spawning.
\lamfork~provides a deadlock-freedom guarantee while allowing programs to dynamically determine
when to spawn new processes, what they will do, and who will communicate with them.
All results are formalized in Lean.

\end{abstract}

\maketitle

\section{Introduction}
\label{sec:introduction}

As nearly every computer system now relies on parallelism for efficiency,
the challenge of writing correct concurrent code is an increasing concern.
Complex interactions between processes can lead to subtle bugs such as deadlocks,
where multiple processes wait on each other, preventing the system from progressing.
\emph{Choreographic programming}~\citep{Montesi23} is a promising new tool to address this challenge.
Instead of writing separate programs for each process,
programmers write a choreography: a single top-level program that specifies the behavior and interactions of all processes together.
A compiler then produces code for each process using a procedure called \emph{endpoint projection}~(EPP).
This global specification allows programmers to reason about the system as a whole and leads to \emph{deadlock freedom by design},
eliminating a common source of bugs in concurrent programs.

Early choreographic languages showed strong potential,
but lacked features necessary for modern software engineering~\citep{Montesi13,CarboneM13}.
A flurry of recent work has introduced features including
higher-order programming~\citep{CruzFilipeGLMP22,HirschG22,GiallorenzoMP23},
process polymorphism~\citep{GraversenHM24,SamuelsonHC25}, and multiply-located values~\citep{BatesKJSKN25,SamuelsonHC25},
but these languages still cannot dynamically \emph{fork processes}.
This shortcoming may seem surprising given choreographies' close relationship to
behavioral types such as multiparty session types,
which have always supported dynamic process and session creation~\citep{HondaVK98,HondaYC08,JacobsBK22,Kobayashi06,Padovani14}.
However, the two paradigms prevent deadlocks using fundamentally different strategies.
Behavioral types statically track protocols or obligations and enforce that communication follows them.
Choreographies rely on structural properties of EPP, only requiring a \emph{data} type system already familiar to programmers.

Though early choreographic work enabled process spawning~\citep{CarboneM13,CruzFilipeM17c},
it relied on a feature-poor calculus, lacking data structures, closures, and even simple variable binding.
These omissions do not merely simplify presentation;
these languages entirely lack the typing and EPP infrastructure needed to support such features in tandem with process spawning.
Developing these facilities, so that spawning remains deadlock-free alongside functional features, is the crux of this work.

Enabling dynamic process forking raises three important safety and implementation concerns:
\begin{enumerate}[leftmargin=*]
    \item Existing processes must be made aware of a new process so they can interact with it.
    \item Executing processes must not attempt to communicate with a terminated process.
    \item New processes must know what code to execute, and the code must match the expectations of existing processes to avoid deadlock.
\end{enumerate}

Problem~1 can be resolved using first-class process names~\citep{SamuelsonHC25},
a form of higher-order communication that enables safely transmitting process names across the network.
Problem~2 is fundamentally a resource-management problem:
spawned threads are resources that should be released when no longer needed, and using one afterward causes deadlock.
We solve this problem by lifting the principles of type-based resource management~\citep{TofteT97,LaunchburyP94} to process polymorphism.

Problem~3 requires more careful analysis.
New processes must know what code to execute, so their parent must have access to that code to initialize them.
To properly project the parent's code from a choreography, endpoint projection must be able to statically determine what code the child will need.
If the child has an unbounded dynamic lifetime, identifying this code statically is undecidable.
One might give the child the entire choreography,
but avoiding deadlock from mismatched expectations with existing processes would then pose a serious challenge.

A key insight of this work is that we can resolve Problem~3 using a simple and well-understood resource-management technique:
\emph{syntactically scoping} process lifetimes.
With lexical lifetimes,
the code for a spawned thread is just the projection of the choreographic code where the thread's name is in scope---a straightforward compile-time determination.

While lexical scoping might at first appear limiting,
it parallels a restriction behavioral type systems already impose:
protocols~\cite{GayV10,HondaYC08} or obligations~\cite{Dardha18,Kobayashi06,Padovani14,VanDenHeuvel24}
statically bound which continuations an open channel endpoint may enter,
effectively limiting the lifetime of any process communicating over that channel.
Syntactically scoping threads merely aligns this lifetime bound with the program's lexical structure.
Moreover, because threads execute concurrently
and may die after their parent,
they can terminate as soon as their own operations complete,
avoiding the stack-based discipline typically associated with lexically scoped resources.

Nevertheless, combining lexical scoping with functional choreographic programming introduces two significant complications, both stemming from closures.
Consider the following choreographies.
\begin{mathpar}
  \addtocounter{numlevels}{1}
  \def\arraystretch{1.1}
  \begin{array}{@{}l@{}}
    \ReserveWorkerGood =
    \LamN~\task\ldotp
      \begin{array}[t]{@{}r@{~}l@{}}
        \LetN & \begin{array}[t]{@{}l@{}}
          \alpha \ChorDef \Client.\ForkN()
        \end{array} \\
        \In & (\task~\alpha) \ColSend \Client
      \end{array}
  \end{array}
  \and
  \def\arraystretch{1.1}
  \begin{array}{@{}l@{}}
    \color{Red}
    \colorlet{chorcolor}{Red}
    \colorlet{loccolor}{Red}
    \ReserveWorkerBad =
    \LetIn*{\alpha}{\Client.\ForkN()}{
      \LamN~\task \ldotp (\task~\alpha) \ColSend \Client}
  \end{array}
  \addtocounter{numlevels}{-1}
\end{mathpar}
Both take a choreographic function~``\task'' that requires a worker location as an argument,
execute it at a worker thread~$\alpha$ spawned by a client process~\Client,
and send the result to~\Client.
While similar, they differ critically in the scope of~$\alpha$.
The first choreography spawns~$\alpha$ on application, executes the task, sends the result, and terminates~$\alpha$ as it falls out of scope.
The second spawns the worker \emph{first} and returns a function that closes over its name.
The latter illustrates the first complication: thread leakage.
Since~$\alpha$ has a lexical lifetime,
in~\ReserveWorkerBad it terminates before the closure is called,
so a later invocation deadlocks, with~\Client waiting on a message from a worker that is no longer running.

The second complication is projecting closures.
In both implementations, the value bound to~\task is defined outside the scope of~$\alpha$ yet must be executed by~$\alpha$.
More generally, a thread can use a value bound outside its scope,
so its parent must bind a value that behaves correctly for both itself and the thread,
even though their behaviors may differ and the parent does not know about the thread at the binding site.
For plain data, the parent can project normally and the child can ignore extraneous information.
Closures, however, contain pending computation that projects differently for different processes.
Without great care in the definition of EPP,
a function projected for~\Client might not execute correctly when given to a spawned thread~$\alpha$.

This work develops \lamfork (pronounced ``lambda-fork''), the first functional choreographic language to allow dynamic process spawning.
To resolve the two complications introduced by closures, \lamfork employs two distinct mechanisms.
First, \lamfork~avoids thread leakage
by tracking participating processes in the type system and tagging function types with \emph{latent participants}:
locations not spawned in the function body that might be required to execute it.
By prohibiting spawned threads from escaping their scope \emph{even as latent participants},
the type system correctly rejects~\ReserveWorkerBad as ill-typed while accepting~\ReserveWorkerGood.
Second, to correctly project closures for future threads, EPP projects function bodies \emph{twice}, once for the parent defining the function
and again for any child it might later spawn.
The projected code \emph{dynamically selects} between the two depending on who is executing it,
using a new~``\PIDN'' construct to dynamically bind the child's name.

The rest of the paper is structured as follows.
Section~\ref{sec:background} reviews background, and Section~\ref{sec:system_model} defines the underlying system model.
The subsequent sections present the following main contributions.
\begin{itemize}[leftmargin=*]
\item Section~\ref{sec:language} introduces \lamfork,
  a functional choreographic language supporting process spawning.
\item Section~\ref{sec:static-semantics} formulates a type system for \lamfork that tracks which processes might participate in a choreography
  and ensures no process needs to perform computation after it dies.
  The type system is decidable (with minimal syntactic annotations) via a bidirectional algorithm.
\item Section~\ref{sec:endpoint-projection} defines our endpoint projection~(EPP) procedure
  to compile a choreography into a target-language program (Section~\ref{sec:network-lang}) for each process,
  and characterizes its correctness with respect to a top-level operational semantics.
  This result combines with choreographic type soundness to prove that executing a projected system never deadlocks.
\item We formalize and verify all results in the Lean proof assistant~\citep{deMouraKADR15}.
  The development is available on Zenodo~\citep{SamuelsonZHC25}.
\end{itemize}
Section~\ref{sec:related-work} reviews related work, and Section~\ref{sec:conclusion} concludes.

\section{Background}
\label{sec:background}
To lay the groundwork for our contributions, we first review the design and features of the choreographic languages on which we build.

\label{sec:bg-functional-choreo}
Our language~\lamfork primarily extends \lamqc~\citep{SamuelsonHC25} and Pirouette~\citep{HirschG22}.
Like prior work, \lamfork~prefixes local operations with the process that performs them.
For example, $\Alice.(1 \LocalPlus 3) \ColSend \Bob$ specifies that~\Alice computes a sum and sends the result to~\Bob.
Typographically, we write choreographic syntax in \FontChoreo{blue},
processes in \LocFont{red}, and local computations in \LocalLangFont{green}.
Local programs such as ``$1 \LocalPlus 3$'' are written in a user-chosen language,
whereas choreographic constructs such as send~($\ColSend$) are fixed.

Although choreographic and local operations are separate,
a choreography can sequence local computations using features such as \LetN-expressions.
For example, the output of~$\Alice.(1 \LocalPlus 3) \ColSend \Bob$ is an integer located at~\Bob,
which \Bob can then use in a subsequent local computation by binding it to a (local) variable~$x$ as follows:
$\LetIn{\Bob.x}{\Alice.(1 \LocalPlus 3) \ColSend \Bob}{\Bob.(x \LocalMinus 2)}$.

\lamfork also supports \emph{multiply-located values}~(MLVs)~\citep{BatesKJSKN25,SweetDHEHH23},
which specify a local value or computation at multiple locations simultaneously and ensure they agree.
For instance, the multiply-located computation $\{\Alice,\Bob\}.(3 \LocalGreater 1) \ChorSend[\Alice] \Client$
instructs both~\Alice and~\Bob to locally compute~$3 \LocalGreater 1$,
and then~\Alice to send the result to~\Client.
The result is the MLV~$\{\Alice,\Bob,\Client\}.\True$ known to all three parties.

Choreography-level branching---implemented by $\ITEBase{C}{C_1}{C_2}$---enables the entire set of parties who know a MLV's value (e.g., $\{\Alice,\Bob,\Client\}$ in the MLV above) to branch.
If a choreography branches in this manner, locations unaware of the decision must be informed in order to participate.
This is accomplished with the \emph{selection statement} $\syncs{\ell}{d}{\rho'} \seq C$, which instructs location~$\ell$ to broadcast its chosen branch $d \in \{\Left,\Right\}$ to a set of observers~$\rho'$, keeping all participants synchronized.

\paragraph{Process Polymorphism}
\lamfork provides two forms of process polymorphism.
First, \emph{process abstractions}~$F = \TLam{\alpha}{C}$, introduced by~\citet{GraversenHM24},
allow defining process-generic choreographies
and instantiating them with any concrete process via an application form~$F~\Alice$.

Second, the language supports \emph{first-class} process names~\citep{SamuelsonHC25}:
local representations of process names (e.g., $\say{\Bob} = \texttt{``Bob''}$) can be evaluated and sent as first-class data.
A \emph{type-let} expression can then bind this dynamically decided name.
For instance, the program below allows \Alice to select between~\Bob and~\Client based on the local condition~$e$.
\Alice~shares the result with~\Bob and~\Client so they know who should perform the subsequent computation,
and they all bind that name to variable~$\alpha$.
\[
  \LetIn{\{\Alice,\Bob,\Client\}.\alpha}{\Alice.(\LocalITE{e}{\say{\Bob}}{\say{\Client}}) \ColSend \{\Bob,\Client\}}{(\alpha.(1 \LocalPlus 3) \ColSend \Alice)}
\]

\paragraph{Semantics}
Like most choreographic languages, \lamfork defines two semantics:
a top-level choreographic semantics,
and \emph{endpoint projection}~(EPP), which compiles a choreography to a concurrent system.
The choreographic semantics supports out-of-order execution while preserving each participant's operation order.
For instance, in
$D = \Alice.(2 \LocalTimes 4) \ColSend \Client \seq \Bob.(3 \LocalPlus 2) \ColSend \Client$,
processes~\Alice and~\Bob each compute a value and then send it to~\Client.
Here~\Alice can compute~$8$ and~\Bob can compute~$5$ in either order, even though~\Bob's computation appears syntactically later.
\Client, however, must receive~$8$ before~$5$.

EPP, by contrast, separates a choreography into per-location programs by
extracting the actions of each location.
For instance, in~$D$ above, \Alice~and~\Bob each compute and send a value, while~\Client receives both.
Thus, the projections for those three locations, denoted~$\epp{D}{\ell}$ for location~$\ell$, are:
\begin{mathpar}
  \epp{D}{\Alice} = \SendTo{\Ret{2 \LocalTimes 4}}{\Client}
  \and
  \epp{D}{\Bob} = \SendTo{\Ret{3 \LocalPlus 2}}{\Client}
  \and
  \epp{D}{\Client} =
  \mkern2mu
  \begin{array}[m]{@{}l@{}}
    \RecvFrom{\Alice} \NtwkSeq \\
    \RecvFrom{\Bob}
  \end{array}
\end{mathpar}
The target (network) language is written using an $\FontNtwk{orange~teletype}$ font.
 
As the top-level and EPP-based semantics provide separate interpretations of the same program,
it is important that their results are equivalent.
Prior work proves a bisimulation between the semantics~\citep{HirschG22,SamuelsonHC25},
which we loosen to eliminate harsh semantic restrictions.
However, our guarantees remain strong:
they ensure that both semantics produce the same result
and that any concurrent execution of the projected choreography is deadlock-free.

\section{System Model}
\label{sec:system_model}

\lamfork assumes a set~\Locations of potential computational units, which we refer to interchangeably as \emph{locations} or \emph{processes}.
We use the term \emph{thread} specifically for a location that was spawned and will terminate during execution.
The number of locations executing at any time is finite,
but programs may spawn an unbounded number of new locations with fresh names,
so~\Locations must be infinite.

\subsection{Local Language}
\label{sec:local-lang}
\lamfork allows local programs in nearly any language.
To support formal reasoning about the choreography, this \emph{local language} must satisfy a few common properties:
it must consist of a set of expressions with a distinguished set of values, a small-step operational semantics, and a type system.
We make a few assumptions on the semantics and type system, which slightly generalize those of prior work.

\paragraph{Local Operational Semantics}
\label{sec:local-oper-semant}
We write $e_1 \localstep e_2$ to denote that the expression~$e_1$ steps to~$e_2$ in the local language's operational semantics.
The semantics must satisfy the following two properties.
\begin{enumerate}[topsep=3pt,itemsep=2pt]
  \item\label{prop:li:values} Values cannot step: if $\val{v}$, then there is no~$e$ such that $v \localstep e$.
  \item\label{prop:li:confluence} Confluence:
    if $e_1 \localsteps e_2$ and $e_1 \localsteps e_3$,
    then there is an~$e_4$ where $e_2 \localsteps e_4$ and $e_3 \localsteps e_4$.
\end{enumerate}
Property~(\ref{prop:li:values}) is a standard assumption about the set of values in the language,
and property~(\ref{prop:li:confluence}) ensures multiply-located computations produce the same result
at each location where they execute.

\paragraph{Local Type System}
\label{sec:local-types}
The local language must provide kinding ($\localkinded{\Gamma}{t}$) and typing ($\localtyped{\Gamma;\Delta}{e}{t}$) judgments,
where~$\Gamma$ is a kinding context and~$\Delta$ is a typing context.
This allows, but does not require, the local type system to be polymorphic.
These judgments use a \LocalColor{green} double-vertical turnstile~$\eprovesCol$ to distinguish them from the choreographic type system.
For simplicity we assume a single local kind~$*_e$, but our formalization allows multiple kinds.
The type system must be sound---satisfy the standard progress and preservation properties---as well as respect well-formed (type) substitutions.

Generalizing prior work that requires local booleans,
we implement choreographic branching by allowing a user-defined predicate
$\isSum{s}{t_1}{t_2}$ indicating that local values of type~$s$ can be interpreted as
either a~$t_1$ or a~$t_2$.
For example, one could specify that~$\isSum{\Bool}{\Unit}{\Unit}$ holds,
and interpret~$\True$ as~$\metaInl{\LocalLangFont{()}}$ and~$\False$ as~$\metaInr{\LocalLangFont{()}}$.
This interpretation is operationalized by a user-provided partial function~$\getCase$ called the \emph{extraction function}.
To ensure well-typedness,
we require that, if $\isSum{s}{t_1}{t_2}$ and~$\localtyped{}{v}{s}$ for a value~$v$,
then either $\getCase(v) = \metaInl(v_1)$ with $\localemptyped{v_1}{t_1}$
or $\getCase(v) = \metaInr(v_2)$ with $\localemptyped{v_2}{t_2}$.
On other expressions, it may be undefined.

As in~\lamqc, we support first-class process names by requiring local types~$\Loc_\rho$ and~$\LocSet_\rho$
that define first-class \emph{representations} of location names and sets of locations, respectively.
It is vital to distinguish between a location constant---a logical name such as~$\Alice \in \Locations$---and its local representation---a value in the local language, such as~$\say{\Alice} = \texttt{``Alice''}$.
Representations may use any underlying data type (e.g., (sets of) strings),
but, as with sums above, the local language must be able to reify any well-typed representation into a corresponding location (set).
The index~$\rho$ of these types statically bounds the set of locations a representation might evaluate to.
For example, a conditional returning either $\say{\Alice}$ or $\say{\Bob}$ can type as $\Loc_{\{\Alice,\Bob\}}$.
Since a new thread may take any (fresh) name,
we also require each location~$L \in \Locations$ to have at least one representation~$\say{L}$ so spawning cannot get stuck.

\begin{ex}[System F]\label{ex:system-f}
    \begin{figure}[t]
    \begin{syntax}
      \category[Local Types]{t}
      \alternative{\alpha}
      \alternative{\Unit}
      \alternative{\Int}
      \alternative{\String}
      \alternative{\Loc_\rho}
      \alternative{\LocSet_\rho} \\
      \alternative{t_1 \to t_2}
      \alternative{t_1 + t_2}
      \alternative{t_1 \times t_2}
      \alternative{\forall \alpha \ldotp t}
      \alternative{\mu \alpha \ldotp t}

      \category[Local Expressions]{e}
      \alternative{x}
      \alternative{()}
      \alternative{n \in \mathbb{Z}}
      \alternative{s \in \mathit{str}}
      \alternative{\LocalFun{f}{x}{e}}
      \alternative{e_1~e_2}
      \alternative{\LocalColor{\Lambda} \alpha \ldotp e}
      \alternative{e~t} \\
      \alternative{\LocalInl{e}}
      \alternative{\LocalInr{e}}
      \alternative{\LocalLangCase{e}{\LocalInl{x}}{e_1}{\LocalInr{y}}{e_2}} \\
      \leftnote{(\LocalBinop \in \{\LocalPlus, \LocalMinus, \LocalTimes, \LocalEq, \LocalLess\})}
      \alternative{(e_1,e_2)}
      \alternative{\LocalFst{e}}
      \alternative{\LocalSnd{e}}
      \alternative{\LocalFold{e}}
      \alternative{\LocalUnfold{e}}
      \alternative{e_1 \LocalBinop e_2}
    \end{syntax}
    \caption{Example Local Language Syntax.}
    \label{fig:local-lang-syntax}
  \end{figure}

  Because our assumptions generalize prior choreographic calculi, many standard $\lambda$-calculi can serve as the local language.
  We use an exemplar based on System~F (Figure~\ref{fig:local-lang-syntax}),
  extended with algebraic data types, recursive types, integers, strings, and recursive functions.
  Strings represent locations, location sets are lists of strings, and sum types use standard injections;
  the extraction function pattern matches on these injections and is undefined otherwise.
  That is, $\text{isSum}$ is true precisely when $\isSum{t_1 + t_2}{t_1}{t_2}$,
  and the extraction function is defined as
  $\getCase(\LocalInl{v_1}) = \metaInl(v_1)$ and $\getCase(\LocalInr{v_2}) = \metaInr(v_2)$,
  failing to be defined for other expressions.
\end{ex}

\section[The Lambda-Fork Language]{The \lamfork Language}
\label{sec:language}

We now present \lamfork, the first functional choreographic language that can dynamically spawn threads.
We inherit its core from \lamqc~\citep{SamuelsonHC25}, including algebraic and recursive data types,
multiply-located computations, process polymorphism, and first-class process names,
while retaining the traditional deadlock-freedom guarantee of choreographic languages.

\subsection[Lambda-Fork Syntax]{\lamfork Syntax}
\label{sec:syntax}

\begin{figure}[b]
  \begin{syntax}
    \category[Selection Labels]{d}
    \alternative{\Left}
    \alternative{\Right}

    \category[Choreographies]{C}
    \alternative{X} \alternative{\rho.e}
    \alternative{\LetIn{\rho.x}{C_1}{C_2}}
    \alternative{\LetIn{\rho.\alpha \knd \kappa}{C_1}{C_2}}\\
    \alternative{C \ChorSend[\ell] \rho}
    \alternative{\syncs{\ell}{d}{\rho} \seq C} \\
    \alternative{\Fun{\rho}{F}{X}{C}}
    \alternative{C_1~C_2}
    \alternative{\TFun{\rho}{F}{\alpha \knd \kappa}{C}}
    \alternative{C~t} \\
    \alternative{\LocalCase{}{C}{x}{C_1}{y}{C_2}}\\
    \alternative{\Inl{\rho}{C}}
    \alternative{\Inr{\rho}{C}}
    \alternative{\Case{}{C}{X}{C_1}{Y}{C_2}} \\
    \alternative{\Fork{(\alpha,x)}{\ell}{C}}
    \alternative{\KillAfter{L}{C}}
    \alternative{\cdots}
  \end{syntax}

  \caption{Selected Choreographic Syntax of \lamfork.}
  \label{fig:abstract-syntax}
\end{figure}

Figure~\ref{fig:abstract-syntax} presents the core syntax of \lamfork.
We write choreographic program variables in uppercase Roman characters ($X, Y, F, \dotsc$),
local program variables in lowercase Roman characters ($x, y, f, \dotsc$),
and type variables in lowercase Greek characters ($\alpha, \beta, \dotsc$).
Informally, the metavariable~$\ell$ denotes a location, $\rho$~a~set of locations, $\tau$~a~choreographic type, $t_e$~a~local type, $\kappa$~a~kind, and $t$~a~type of any kind.
To support polymorphism, these metavariables all formally represent a type- or kind-expression (e.g., the open set~$\rho = \alpha \cup \{\beta,\Alice\}$ is valid),
whose syntax we define in Section~\ref{sec:static-semantics}.

Most constructs in \lamfork are standard functional operations generalized to choreographies,
but there are some key differences.
For instance, the expression~$\rho.e$ executes local program~$e$ at every location in the (non-empty) set~$\rho$.
In cases where~$\rho = \{\ell\}$ is a singleton, we use the shorthand~$\ell.e$.

Local programs like~$e$ can use local variables bound in the choreography, which are scoped by the location(s) binding them:
$\Alice.x$ is distinct from $\Bob.x$.
Variables bound across multiple locations are written~$\rho.x$.
Choreographies bind local (resp. type) variables via a \LetN expression~$\LetIn{\rho.x}{C_1}{C_2}$
(resp.~$\LetIn{\rho.\alpha \knd \kappa}{C_1}{C_2}$),
where~$\rho$ is a subset of locations aware of $C_1$'s output.

The send operation~$C \ChorSend[\ell] \rho$ evaluates~$C$ to a local value known to~$\ell$ (and possibly others)
and directs~$\ell$ to send it to all locations in~$\rho$.
We abbreviate this to $C \ColSend \rho$ when only~$\ell$ initially knows the value.
This operation is \emph{collecting}: the result is a multiply-located value known to the sender, recipients, and any previous observers.
For example, $\{\Alice,\Bob\}.(4 \LocalMinus 2) \ChorSend[\Alice] \{\Client\}$ yields $\{\Alice, \Bob, \Client\}.2$.

Named recursive functions~$\Fun{\rho}{F}{X}{C}$ and recursive type functions~$\TFun{\rho}{F}{\alpha \knd \kappa}{C}$
require an annotation~$\rho$ denoting their participating locations, since operations do not always specify their participants.
For instance, $(\Fun{\rho}{F}{X}{F~X})~V$ will send all locations in~$\rho$ into an infinite loop.
Without the annotation, it is ambiguous who should loop and who should not.
When a type function abstracts over a location or location set ($\kappa \in \{\locKnd,\setKnd,\finsetKnd\}$),
it may be applied to any location (set), so all locations must know its definition.
In such cases, $\rho$ is universal, and we elide it ($\TFunLoc{F}{\alpha \knd \kappa}{C}$).
Applications and type applications are written~$C_1~C_2$ and~$C~t$, respectively.

\lamfork distinguishes between two forms of sum types:
\emph{local} sums ($\rho.(\LocalInl{v_1})$ or~$\rho.(\LocalInr{v_2})$), representing choices evaluated within a local program,
and \emph{choreographic} sums ($\Inl{\rho}{V_1}$ or~$\Inr{\rho}{V_2}$),
which allow data on either side of the branch to reside asymmetrically at different locations.
Consequently, \lamfork provides two branching constructs.
\emph{Local} case expressions ($\LocalCase{}{C}{x}{C_1}{y}{C_2}$) branch on a local sum~$C$,
which must produce a value of a local type~$s$ where $\isSum{s}{t_1}{t_2}$ holds for some~$t_1$ and~$t_2$.
Locations ignorant of~$C$'s value can participate either by receiving~$C$ prior to branching
or via selection statements ($\syncs{\ell}{d}{\rho} \seq C$) within the branches of the \LocalCaseN,
wherein~$\ell$ notifies~$\rho$ of the chosen branch direction~$d \in \{\Left,\Right\}$.
The latter avoids sending the entire value of~$C$ to the additional locations,
which may be desirable for security or performance reasons.
\emph{Choreographic} case expressions ($\Case{}{C}{X}{C_1}{Y}{C_2}$) similarly branch on choreographic sums and also support selection statements.

Pairs and recursive data types act like their usual functional-programming counterparts.

\paragraph{Fork and Kill Expressions}
The key addition of \lamfork is the~\ForkN expression, which allows dynamic spawning of new locations.
Specifically, $\Fork{(\alpha,x)}{\ell}{C}$ instructs~$\ell$ to spawn a new child process,
and everyone---including the child---to continue on~$C$.
The new process's name is bound to the type variable~$\alpha$, while~$x$ binds a local,
first-class representation of~$\alpha$'s name known to both~$\ell$ and~$\alpha$.
This representation can be sent to other locations to facilitate direct communication with~$\alpha$.

We include two notational shortcuts for~\ForkN.
First, we write~$\Fork{\alpha}{\ell}{C}$ when~$x$ is not free in~$C$.
Second, the notation~$\LetFork{\alpha}{\ell}{\rho}{C}$
is sugar for~$\LetIn{(\_,x)}{\ell.\ForkN()}{(\LetIn{\alpha}{x \ChorSend[\ell] \rho}{C})}$,
which shares the thread's name with everyone in~$\rho$.

The administrative term~$\KillAfter{L}{C}$ serves as a dual to~\ForkN and tracks thread lifetimes.
It is not written by the programmer;
rather, when a thread~$L$ is spawned by a~\ForkN with body~$C$,
the choreographic operational semantics wraps the executing body~$C$ in a~\KillAfterN
to denote that~$L$ will die once~$C$ finishes execution.
This makes the~\ForkN's scope explicit during reduction,
giving the type system and EPP the syntactic handle they need to track the thread's lifetime.

Example~\ref{ex:load-balancer} below combines~\ForkN with case expressions and (first-class) process polymorphism.

\begin{ex}[Load Balancer]
\label{ex:load-balancer}
  Consider a cloud computing application where a client~\Client wishes to outsource an expensive computation~$F$ with input~$X$.
  The function~\WithWorker below shows how~\Client can run~$F$ on a generic worker~$W$ using process polymorphism.
  Once the worker has computed~$F~X$, it informs a manager process~\Mngr, which then executes a callback function~\programfont{onFinish}.
  \[\rulefiguresize
    \addtocounter{numlevels}{1}
    \begin{array}{@{}l@{}}
      \WithWorker~W~F~X~\programfont{onFinish} = \\
      \def\arraystretch{1.1}
      \quad\LetMany{{W.f}{F \ColSend W}{W.x}{X \ColSend W}{\Client.\mathit{res}}{W.(f~x) \ColSend \Client}}
              {\begin{array}[t]{@{}l@{}}
                  W.\LocalLangFont{``done"} \ColSend \Mngr \seq \\
                  \Mngr.(\programfont{onFinish}~()) \seq \\
                  \Client.\mathit{res}
                \end{array}}
    \end{array}
    \qquad
    \begin{array}{@{}l@{}}
        \HandleRequest~F~X = \\
        \def\arraystretch{1.1}
        \quad\begin{array}[t]{@{}l@{}}
        \LocalCaseN~(\Mngr.\AcquireWorker() \ColSend \{\Client\} \cup \Pool) ~ \OfN \\
          {}\mid \LocalSome{w} \Rightarrow
          \LetMany{{W}{w}{\Mngr.f}{\Mngr.(\LocalLam \_ \ldotp \ReleaseWorker~w)}}{\WithWorker~W~F~X~\Mngr.f} \\
          {}\mid \LocalNone \Rightarrow
          \LetIn*{W}{\Mngr.\ForkN() \ColSend \Client}{\WithWorker~W~F~X~\Mngr.(\LocalLam \_ \ldotp ())}
        \end{array}
    \end{array}
    \addtocounter{numlevels}{-1}
  \]

  However, \WithWorker does not select a worker;
  that job falls to~\Mngr, which maintains a pool of permanent workers, selects an available one for each request,
  and, if all are busy, spawns, using~\ForkN, an ephemeral worker that terminates after a single job.
  The \HandleRequest function above implements this functionality.
  In \HandleRequest, \Mngr~uses~\AcquireWorker to find a free worker~$w$, returning~$\LocalSome{w}$ if one exists and~$\LocalNone$ otherwise.
  Once all relevant parties learn the result, the choreography branches.
  If the job is run on a free worker, \Mngr~releases it afterward;
  otherwise the newly spawned process simply falls out of scope and terminates.

  Note that this example critically relies on first-class location names, inherited from \lamqc~\citep{SamuelsonHC25}.
  Both the output of \AcquireWorker and the spawned location name are sent as messages,
  and the final worker identity is bound to a type-level location.
\end{ex}

\subsection{Operational Semantics}
\label{sec:semantics}

The operational semantics of \lamfork evaluates \emph{fully elaborated terms} produced by type inference (Section~\ref{sec:type-checking}). 
In this syntax, control-flow constructs are explicitly annotated with participant-tracking information.
For instance, function application (resp. local branching) becomes~$C_1 \appchor{\rho} C_2$ (resp.~$\LocalCaseN_\rho~C$)
to denote that the locations in set~$\rho$ will perform the application or branch.
These annotations enable precisely tracking which locations perform a given reduction.

The semantics is defined by a labeled transition system
of the form~$\langle C_1 , \Omega_1 \rangle \step[R] \langle C_2 , \Omega_2 \rangle$.
The parameters~$\Omega_1, \Omega_2 \subset \Locations$ record the locations who are executing the choreography before and after the step, respectively,
and are used to track which locations remain alive.
The label~$R$, which we call a \emph{redex}, tracks the specific reduction occurring.
Because choreographies describe concurrent computation, a primary function of~$R$ is to support \emph{out-of-order} reductions,
allowing non-interacting locations to execute independently,
so long as the order of operations for each location is respected.

\begin{figure}
  \rulefiguresize
  \begin{syntax}
    \category[Redices]{R}
    \alternative{\RDone{\rho}{e_1}{e_2}}
    \alternative{\RApp{\rho}}
    \alternative{\RSendV{L}{m}{\rho}}
    \alternative{\RFork{L}{L'}}
    \alternative{\cdots}
  \end{syntax}
  \[
    \rulefiguresize
    \def\arraystretch{1.1}
    \begin{array}{l@{\quad}l}
      \fbox{$\rloc{R}$} & \fbox{$\cloc{C}$} \\[3pt]
      \begin{array}{r@{\hspace{0.5em}}c@{\hspace{0.5em}}l}
        \rloc{\RDone{\rho}{e_1}{e_2}} & \defeq & \rho \\
        \rloc{\RApp{\rho}} & \defeq & \rho \\
        \rloc{\RSendV{L}{m}{\rho}} & \defeq & \{L\} \cup \rho \\
        \rloc{\RFork{L}{L'}} & \defeq & \{L\}
      \end{array}
      &
      \begin{array}{{@{}r@{~}c@{~}l@{}}}
        \multicolumn{3}{c}{
          \hspace{-0.55em}
          \cloc{V} \defeq \varnothing
          \hspace{1.2em}
          \cloc{\rho.e} \defeq \rho
          \hspace{1.2em}
          \cloc{C \ChorSend[\ell] \rho} \defeq \cloc{C} \cup \{\ell\} \cup \rho
        }
        \\
        \cloc{C_1 \appchor{\rho} C_2} = \cloc{\LetIn{\rho.x}{C_1}{C_2}} & \defeq & \cloc{C_1} \cup \cloc{C_2} \cup \rho
        \\
        \cloc{\Fork{(\alpha,x)}{\ell}{C}} & \defeq & \{\ell\} \cup (\cloc{C} \setminus \{\alpha\})
        \\
        \cloc{\KillAfter{L}{C}} & \defeq & \{L\} \cup \cloc{C}
      \end{array}
    \end{array}
  \]
  \caption{
    Selected Redices and Location Function Rules.
    Here~$m$ is either a local value~$v$ or a selection label~$d$.
  }
  \label{fig:redices}
\end{figure}

We compute the locations involved in a step using the \emph{redex locations} function~$\rloc{R}$,
and those (possibly) involved in an entire choreography using the \emph{choreography locations} function~$\cloc{C}$.
For example, redex~$\RSendV{\Alice}{v}{\Bob}$ denotes that~\Alice sends~$v$ to~\Bob.
Since precisely~\Alice and~\Bob participate in this step, $\rloc{\RSendV{\Alice}{v}{\Bob}} = \{\Alice, \Bob\}$.
Redices~$\RDone{\rho}{e_1}{e_2}$ and $\RApp{\rho}$,
which label stepping a local program and $\beta$-reduction, respectively,
involve the locations in~$\rho$, so $\rloc{R} = \rho$.
Redex~$\RFork{L}{L'}$ appears when~$L$ spawns a thread named~$L'$.
Since~$L'$ is not alive until after the step occurs, only~$L$ participates in the step:
$\rloc{\RFork{L}{L'}} = \{L\}$.

The function~$\cloc{C}$ captures all locations that may eventually participate in~$C$, not only the next step.
For instance,
$\cloc{\LetIn{\Alice.x}{\Alice.2}{(\Alice.(2 \LocalPlus x) \ColSend \Bob)}} = \{\Alice,\Bob\}$,
although~\Alice takes several steps before~\Bob is involved.
Figure~\ref{fig:redices} shows selected redices and definitions for both location functions.

These two functions together determine when it is safe to reorder steps.
Specifically, a step~$R$ can execute before an entire computation~$C$
if the participants in the two are disjoint---$\cloc{C} \cap \rloc{R} = \varnothing$---even if a standard in-order semantics would execute~$C$ to completion before~$R$.
The following out-of-order rule for \LetN-expressions is an example of such a step.
\[
  \rulefiguresize
  \CLetIRule[left]
\]
This rule also prohibits the out-of-order step~$R$ in the body from including locations binding a variable in the~\LetN,
and ensures that all locations binding the \LetN have been resolved by requiring~$\fv{\rho} = \varnothing$.
The latter requirement prevents a location variable from later resolving to a location appearing in the step,
whose operations \ruleref{C-LetI} would then have reordered.

Out-of-order execution can similarly occur within~\CaseN branches before evaluating the scrutinee,
with reordering restrictions similar to \ruleref{C-LetI},
and the additional requirement that the step will occur regardless of the branch taken.
The latter point is enforced by requiring identical redices~$R$ and new executing locations~$\Omega'$ in both branches,
as in~\ruleref{C-CaseI} below and its~\LocalCaseN analogue.
\[
  \rulefiguresize
  \CCaseIRule[lab]
\]
\begin{figure}
  \rulefiguresize
  \begin{ruleset}
    \CDoneRule
    \and
    \CAppRule
    \and
    \CSendVRule
    \and
    \toggletrue{CForkLinebreak}
    \CForkRule
    \and
    \CKillRule
  \end{ruleset}
  \caption{Selected \lamfork Operational Semantics}
  \label{fig:semantics}
\end{figure}

Figure~\ref{fig:semantics} contains a selection of additional rules.
The remaining rules are either very similar to those presented here,
or are nearly identical to ones for the corresponding construct in standard call-by-value $\lambda$-calculus.
The complete semantics can be found in
 Appendix~\ref{sec:full-chor-sem}.

\ruleref{C-Done} lifts the local-language semantics to choreographies,
\ruleref{C-App} applies a function to its argument,
and \ruleref{C-SendV} formalizes the multiply-located semantics of message-passing.
These steps require that all of the ``named locations'' in~$\rho$ (written~$\nl{\rho}$) are running, formalized as $\nl{\rho} \subseteq \Omega$.
Thus, every process which needs to perform the action can do so.
Formally, the named locations function~$\nl{\rho}$ homomorphically collects all concrete location names appearing in the set.
Its full definition can be found in

  Appendix~\ref{sec:pn-def}.

The final two rules formalize thread lifecycles.
To spawn a thread, \ruleref{C-Fork} selects a globally fresh name~$L'$---one not in the running locations~$\Omega$ or the current choreography, \emph{and not equal to any name used earlier}---and binds the type variable~$\alpha$ to~$L'$ and the local variable~$x$ to its representation~$\say{L'}$.
It ensures the substituted body~$C'$ is closed, adds~$L'$ to the active set~$\Omega$, and wraps~$C'$ in a~\KillAfterN.
Once the thread's computation completes, \ruleref{C-Kill} removes~$L'$ from~$\Omega$.
An out-of-order variant of~\ruleref{C-Kill} allows a thread~$L$ to terminate
when it has no remaining computation in~$C$ (i.e., $L \notin \cloc{C}$).

Requiring new thread names be globally fresh, rather than not currently running, may seem superfluous,
but is necessary to correctly simulate the projected system (Section~\ref{sec:endpoint-projection}).
Projected functions reference the name of the process that projected them
and can be passed to a spawned thread, which must determine whether it is that process.
Thus, names of dead threads may persist in projected functions,
and reusing such a name could cause a thread to execute incorrect code.

\subsection{Static Semantics}
\label{sec:static-semantics}
The static semantics of our language is defined by a kinding judgment and a typing judgment.

\subsubsection{Kinding System}
\label{sec:kind-system}

\begin{figure}[t]
  \begin{syntax}
    \category[Kinds]{\kappa}
    \alternative{*_e}
    \alternative{\locKnd}
    \alternative{\setKnd}
    \alternative{\finsetKnd}
    \alternative{\tyknd{\rho}}

    \category[Local Types]{t_e}
    \alternative{\alpha}
    \alternative{\Loc_\rho}
    \alternative{\LocSet_\rho}
    \alternative{\ldots}

    \categoryFromSet[Locations]{L}{\Locations}

    \vspace{-3pt}
    \category[Types]{\ell, \rho, \tau, t}
    \alternative{\alpha}
    \alternative{t_e @ \rho}
    \alternative{\tau_1 \arr{\rho} \tau_2}
    \alternative{\forall \alpha \knd \kappa [\rho] \ldotp \tau}
    \alternative{\tau_1 \times \tau_2}
    \alternative{\tau_1 +_\rho \tau_2}
    \alternative{\mu_\rho \alpha \ldotp \tau}
    \\
    \alternative{L}
    \alternative{\{\ell\}}
    \alternative{\rho_1 \cup \rho_2}
    \alternative{\anyLoc}
  \end{syntax}

  \caption{Syntax of Types and Kinds. Here~$\alpha$ is a type variable.}
  \label{fig:types}
\end{figure}

To support type and process polymorphism, the kinding judgment~$\chorkinded{\Gamma}{t}{\kappa}$ classifies a type~$t$ as a local type ($*_e$),
a location ($\locKnd$), a set of locations ($\setKnd$),
a finite set of locations ($\finsetKnd$), or a choreographic program type ($\tyknd{\rho}$).
Figure~\ref{fig:types} presents the syntax for these types and kinds.

The kind~$\locKnd$ classifies locations, which may be concrete~$L \in \Locations$ or a type variable~$\alpha$.
Types of kind~$\setKnd$ are (non-empty) sets of locations,
including type variables~$\alpha$, singleton sets~$\{\ell\}$, unions~$\rho_1 \cup \rho_2$,
and the infinite set~$\anyLoc$ of all locations.
The set~$\anyLoc$ is crucial for typing process-polymorphic functions, as they may involve any location.
However, actual computation (like sending messages) cannot involve infinite sets.
Thus, the sub-kind~$\finsetKnd$ classifies sets which are \emph{finite} (do not contain~$\anyLoc$).

These location sets are equipped with syntactically defined containment and subset relations, and a set difference function.
Note that, as~$\anyLoc$ represents all locations, $\ell \in \anyLoc$ for all~$\ell$, including unresolved type variables.
The formal definition of these operations can be found in

  Appendix~\ref{sec:set-relations}.

The final kind~$\tyknd{\rho}$ of program types is similar to the kind~$*$ of System~F and \lamqc,
but includes a set~$\rho$ to bound who might know any part of a value with a type of this kind.
For instance, a located type~$t_e@\rho$ has kind $\tyknd{\rho}$.
A value~$\rho.v$ of that type is known to everyone in~$\rho$, but nobody else.
This tracking ensures everyone performing an operation (e.g., an application)
possesses the required information (e.g., the function body),
even when this information is not evident from the syntax.

The \ruleref{K-Prod} and \ruleref{K-Sum} rules below give two examples of how compound types are kinded.
\begin{mathpar}[\rulefiguresize]
  \KProdRule
  \and
  \KSumRule
\end{mathpar}
In \ruleref{K-Prod}, if a location knows (part of) either side of a pair, then they know part of the entire pair.

One may expect \ruleref{K-Sum} to follow a similar pattern: collect the annotations on each side.
However, sums specify not only the information in their underlying types, but also whether the value is an~\InlN or an~\InrN.
The~$\rho$ on the plus describes \emph{who knows which side the value is on}, which may include more locations than know the data on each side.
For instance, a value of type $(\Int@\Alice +_{\{\Alice,\Bob,\Client\}} \Int@\Bob) \kndCol \tyknd{\{\Alice,\Bob,\Client\}}$
could be an~\Int at either~\Alice or~\Bob, but all of~\Alice, \Bob, and~\Client know which.
Requiring $\rho_1 \cup \rho_2 \subseteq \rho$ ensures that everyone who might hold data knows whether they need to hold that data.

A function and each application must be known to anyone aware of its inputs or outputs, or involved in computing its body---the \emph{latent participants}.
We track a function's latent participants with an annotation~$\rho$ on its type.
The kinding rule \ruleref{K-Arrow} below gives a function type a kind including all three sets of locations in its bound.

A forall type~$\allty{\alpha \knd \kappa}{\rho}{\tau}$ similarly tracks latent participants with the set~$\rho$.
When not abstracting over a location (set), the participants are fixed.
However, when $\kappa \in \{\locKnd, \setKnd, \finsetKnd\}$, $\alpha$ is bound in both~$\tau$ \emph{and~$\rho$}.
As described above, there is no finite bound on who might need to know such a process-polymorphic function,
so \ruleref{K-AllLoc} below gives it kind~$\tyknd{\anyLoc}$.
\begin{mathpar}[\rulefiguresize]
  \KArrowRule
  \and
  \KAllLocRule
\end{mathpar}

\subsubsection{Type System}
\label{sec:type-system}

Typing judgments in \lamfork take the form~$\chortyped{\ctx}{C}{\tau}{\rho}$,
where $\ctx = \Gamma;\Delta_e;\Delta$ is a three-part context of type, local, and choreographic variables,
$C$~is a choreography, $\tau$~is a program type,
and $\rho$~is a set of locations which bounds who might be involved in computing~$C$.

\paragraph{Participant Tracking}
\label{ex:participant-set-type}
To see how (latent) participant tracking ensures that a killed thread is not asked to perform further computation,
consider the following program:
\[
  \LetIn{F}{\left(
    \begin{array}[c]{@{\,}r@{~}l@{\,}}
      \LetN  & \alpha \ChorDef \Alice.\ForkN() \\
      \In & (\LamN \_ \ldotp \LetIn{\Alice.x}{(\alpha.(1 \LocalPlus 2) \ColSend \Alice)}{\Alice.x})
    \end{array}
  \right)}{F~\Alice.\LocalLangFont{()}}
\]
Here \Alice spawns a thread~$\alpha$ who immediately dies,
as the body of the \ForkN expression---the $\LamN$-abstraction---is a value.
Yet the returned abstraction closes over~$\alpha$ and, if applied,
would leave \Alice waiting forever for a message from the dead~$\alpha$.
Participant tracking prevents this: the $\LamN$-abstraction's body types with~$\rho = \{\alpha,\Alice\}$,
which is recorded in the abstraction's function type:
\[ \chortyped{\alpha \knd \locKnd}{(\LamN \_ \ldotp \LetIn{\Alice.x}{(\alpha.(1 \LocalPlus 2) \ColSend \Alice)}{\Alice.x})}{\Unit@\Alice \arr{\smash{\{\alpha,\Alice\}}} \Int@\Alice}{\varnothing} \]
Because the spawned thread~$\alpha$ appears in the latent participants~$\{\alpha,\Alice\}$ of the function type,
our typing rule rejects the enclosing~\ForkN expression.

By contrast, in the program below, the~\ForkN's body has type~$\Unit@\Alice \arr{\raisebox{-1.5pt}[0pt][0pt]{$\scriptstyle\{\Alice,\Bob\}$}} \Int@\Alice$,
which does not contain~$\alpha$.
It is thus safe to use the value after~$\alpha$ is killed, so the whole term is well-typed.
\[
  \provesCol
  \LetIn{F}{\left(
    \begin{array}[c]{@{\,}r@{~}l@{\,}}
      \LetN  & \begin{array}[t]{@{}l@{}l@{}}
        \alpha & {}\ChorDef \Alice.\ForkN() \ColSend \Bob \\
        \Bob.y & {}\ChorDef \alpha.(1 \LocalPlus 2) \ColSend \Bob
      \end{array} \\
      \In & (\LamN \_ \ldotp \LetIn{\Alice.x}{(\Bob.y \ColSend \Alice)}{\Alice.x})
    \end{array}
  \right)}{F~\Alice.\LocalLangFont{()}}
  \tyCol \Int@\Alice \locsCol \{\Alice,\Bob\}
\]

\paragraph{Typing Rules}

\begin{figure}
  \begin{mathpar}[\rulefiguresize]
    \SetRuleLabelLoc{lab}
    \TFunRule \and \TAppRule
    \and
    \TTFunLocRule \and \TTAppLocRule
    \and
    \TLetLocalRule \and \TForkRule
  \end{mathpar}
  \caption{Selected \lamfork Typing Rules}
  \label{fig:typing}
\end{figure}

Figure~\ref{fig:typing} presents a selection of key~\lamfork typing rules.
\ruleref{T-Fun} and~\ruleref{T-App} formalize participant tracking:
the latent participants on the function type must match the body's participants,
and everyone involved in the input, output, or body must perform the application.

\ruleref{T-TFunLoc} for process-polymorphic functions is similar to the rule for standard functions,
but since~$\alpha$ is bound in the function's body, it may appear freely in the type~$\tau$ and latent participants~$\rho$.
\ruleref{T-TAppLoc} also mirrors the standard application rule,
but it substitutes~$\alpha$ for the now-resolved variable.
Thus, although everyone must know the function's definition,
only those locations who know the output or are involved in the function body \emph{after resolving~$\alpha$} must perform the application.

\ruleref{T-LetLocal} types \LetN-expressions by ensuring the head produces a multiply-located value at all locations in~$\rho_2$,
a subset~$\rho_1 \subseteq \rho_2$ of which binds local variable~$x$ in the body,
whose type is given to the entire expression.
The participants are collected from the head, body, and anyone who binds~$x$.

Finally, the new~\ForkN expression is typed by~\ruleref{T-Fork}.
It requires the body to be well-typed with~$\alpha$ and~$x$ in scope,
but crucially requires the resulting type~$\tau$ to be well-kinded \emph{without~$\alpha$ in scope},
preventing spawned threads from being referenced after termination.
The overall participants~$\{\ell\} \cup (\rho \setminus \{\alpha\})$ include the parent~$\ell$ and the body's participants~$\rho$ but exclude the thread~$\alpha$,
which cannot execute outside the~\ForkN, so the surrounding scope need not know about it.

This rule mirrors \texttt{letregion} of \citet{TofteT97} and \texttt{runST} of \citet{LaunchburyP94},
where a bound region or state thread may not occur free in the result type or environment,
and is discharged from the effect;
latent participants play the role of latent effects~\citep{Lucassen88},
blocking escape even through pending computation.
We transplant this reasoning to choreographies in order to manage threads rather than storage,
and where escape causes deadlock rather than dangling.

Note that there is no typing rule for \KillAfterN~expressions,
as the type system is only intended to handle surface-language programs.
In Section~\ref{sec:type-sound} below we introduce an augmented type system which can handle~\KillAfterN,
allowing us to prove type soundness.

The full typing rules can be found in  Appendix~\ref{sec:full-types}.

\subsubsection{Type Soundness}
\label{sec:type-sound}
The type system above enjoys two notions of soundness:
$\rho$ captures every location that might step,
and the typical guarantee that well-typed choreographies do not get stuck.

The following theorem formalizes the first property;
as with all results here, it is verified in Lean.
Recall that each step is labeled by a redex~$R$,
and that~$\rloc{R}$ computes the locations involved in~$R$.
\begin{restatable}[Sound Participants]{thm}{participantsSound}
  \label{lem:participants-sound}
  If $\chortyped{\ctx}{C}{\tau}{\rho}$ and $\conf{C}{\Omega} \step[\smash{R}] \conf{C'}{\Omega'}$,
  then $\rloc{R} \subseteq \rho$.
\end{restatable}
This soundness does not directly extend to multiple steps,
as a step that spawns a thread may increase the participants in a choreography.
However, by type preservation (Theorem~\ref{thm:preservation} below),
$C'$ is typed at \emph{some} set~$\rho'$ whose new locations (those in~$\rho' \setminus \rho$) are bounded by those spawned by the step~$R$ (in~$\Omega' \setminus \Omega$),
so~$\rho$ tracks all future steps made by \emph{currently spawned} locations.

Our traditional type soundness result follows from a standard progress-and-preservation argument.
However, types are not preserved in the surface type system above, as stepping~\ForkN immediately creates a~\KillAfterN term.
Not only is there no typing rule for~\KillAfterN,
but a na\"{i}ve one would be unsound, as it cannot prevent \emph{external scopes} from referencing the ephemeral thread.
We thus introduce an augmented judgment~$\chortypedplus{\ctx}{C}{\tau}{\rho}$ which enforces this well-scoping by ensuring a spawned thread never
performs computation outside the scope of the associated~\KillAfterN.
Its rules extend the surface rules with extra premises.
For example, \ruleref{S-LetLocal} below types \LetN-expressions.
\begin{mathpar}[\rulefiguresize]
  \SetRuleLabelLoc{left}
  \SLetLocalRule
  \hspace{10pt}
  \SKillRule
\end{mathpar}
The first four premises of~\ruleref{S-LetLocal} match those of~\ruleref{T-LetLocal}.
The last three newly added premises rule out ill-scoping possibilities where
threads spawned in the head perform computation in the body,
threads spawned in the body perform computation in the head,
or threads spawned in either bind~$x$.
\ruleref{S-Kill} similarly prevents a thread from appearing in two different~\KillAfterN expressions.
This well-scoping is enforced using participant tracking to determine who performs computation in each subexpression,
and the spawned-locations function~$\spawnedlocs{C}$ to compute each thread's lexical scope;
this function collects all locations~$L$ appearing in a subterm~$\KillAfter{L}{C'}$ of~$C$.
The other rules of the augmented typing judgment can be found in  Appendix~\ref{sec:full_chor_wellscoped}{}.

This augmented typing judgment is fully sound with respect to the operational semantics,
as demonstrated by the progress and preservation theorems below.
\begin{restatable}[Progress]{thm}{progress}
\label{thm:progress}
  If $\choremptypedplus{C}{\tau}{\rho}$ and $\namedlocs{\rho} \subseteq \Omega$,
  then either~$C$ is a value or~$\conf{C}{\Omega}$ can step.
\end{restatable}
\vspace{-0.5\baselineskip}\begin{restatable}[Preservation]{thm}{preservation}
\label{thm:preservation}
  If $\chortypedplus{\ctx}{C}{\tau}{\rho}$,
  $\namedlocs{\rho} \subseteq \Omega$, and
  $\langle C , \Omega \rangle \steps{} \langle C' , \Omega' \rangle$,
  then there is some $\rho'$ such that
  $\chortypedplus{\ctx}{C'}{\tau}{\rho'}$,
  $\namedlocs{\rho'} \subseteq \Omega'$, and
  $\nl{\rho'} \setminus \nl{\rho} \subseteq \Omega' \setminus \Omega$.
\end{restatable}

The premise~$\nl{\rho} \subseteq \Omega$ requires all potential participants to be running.
The use of the named locations function~$\nl{\rho}$ is a technicality needed to handle when~$\rho$ contains~$\anyLoc$.
Preservation maintains this invariant and ensures that only newly spawned locations materialize in~$\rho'$.

Finally, for choreographies without \KillAfterN,
the two type systems are equivalent,
yielding soundness of the standard type system when combined with progress and preservation.
\begin{restatable}[Typing Equivalence]{thm}{typeEquiv}
\label{thm:type-equiv}
  $\chortyped{\ctx}{C}{\tau}{\rho}$ if and only if~$\chortypedplus{\ctx}{C}{\tau}{\rho}$ and~$C$ has no~\KillAfterN.
\end{restatable}

\subsection{Algorithmic Typing}
\label{sec:type-checking}

While the declarative type system is powerful, a practical implementation of \lamfork needs a type checking algorithm.
Moreover, endpoint projection relies on type annotations which impair developer ergonomics.
To solve both problems, we implement a bidirectional formulation~\citep{PierceT00} of our surface type system in Lean.
This procedure relies on a new assumption:
the local language must provide its own bidirectional typing judgments for both inference and checking.

Recall that our declarative surface judgment is of the form~$\chortyped{\ctx}{C}{\tau}{\rho}$.
As~$\tau$ and~$\rho$ can each be considered either ``input'' or ``output'' to the judgment,
our algorithmic formulation splits this into four mixed modes to separately check~($\checkarrow$) or synthesize/infer~($\syntharrow$) each parameter,
and additionally outputs a fully elaborated version of the choreography~$C$, as needed by the operational semantics and endpoint projection.
For instance, applications~$C_1~C_2$ are elaborated as~$C_1 \appchor{\rho} C_2$ to denote that locations in~$\rho$ perform the application,
and variables~$X$ are elaborated as~$X_\rho$ to denote that only locations in~$\rho$ know any part of a value that~$X$ resolves to.
We use~$\chortypedsynthany{\ctx}{C}{\tau}{\rho}{C'}$ to denote any combined mode.
For example, $\chortypedsynth{\ctx}{C}{c}{\tau}{i}{\rho}{C'}$ checks against~$\tau$ but infers~$\rho$, and outputs the elaborated program~$C'$.
Below, we highlight the mixed modes and how this system uses the local inference procedure.
\begin{ruleset}
    \SetRuleLabelLoc{lab}
    \Rule{TI-Done}{
        \chorempctxwf{\ctx} \\
        \chorkinded{\ctx}{\rho}{\setKnd} \\\\
        \localtypedsynth{\proj{\ctx}{\rho}}{e}{t} \\\\
        \rho' = \text{if}~\val{e}~\text{then}~\varnothing~\text{else}~\rho
    }{
        \chortypedsynth{\ctx}{\rho.e}{i}{t @ \rho}{i}{\rho'}{\rho.e}
    }
    \hfill
    \Rule{TC-Done}{
        \chorempctxwf{\ctx} \\
        \chorkinded{\ctx}{\rho}{\setKnd} \\\\
        \localtypedcheck{\proj{\ctx}{\rho}}{e}{t} \\\\
        \rho' = \text{if}~\val{e}~\text{then}~\varnothing~\text{else}~\rho
    }{
        \chortypedsynth{\ctx}{\rho.e}{c}{t @ \rho}{i}{\rho'}{\rho.e}
    }
    \hfill
    \Rule{TI-App}{
        \chortypedsynth{\ctx}{C_1}{i}{\tau_1 \arr{\rho} \tau_2}{i}{\rho_1}{C_1'} \\\\
        \chortypedsynth{\ctx}{C_2}{c}{\tau_1}{i}{\rho_2}{C_2'} \\
        \chorkinded{\ctx}{\tau_1}{\tyknd{\rho_a}} \\\\
        \chorkinded{\ctx}{\tau_2}{\tyknd{\rho_b}} \\
        \rho' = \rho_a \cup \rho_b \cup \rho
    }{
        \chortypedsynth{\ctx}{C_1~C_2}{i}{\tau_2}{i}{\rho_1 \cup \rho_2 \cup \rho'}{C_1' \appchor{\rho'} C_2'}
    }
\end{ruleset}

Rules~\ruleref{TI-Done} and~\ruleref{TC-Done} leverage the local-language--provided algorithmic judgments $\localtypedsynth{\Gamma;\Delta_e}{e}{t}$ and $\localtypedcheck{\Gamma;\Delta_e}{e}{t}$.
The~\emph{context projection}~$\proj{\ctx}{\rho}$ restricts~$\ctx$ to only those (local) variables available to all locations in~$\rho$.
Crucially, \ruleref{TC-Done} demonstrates the utility of mixed modes:
even when type~$t @ \rho$ is provided as a checking input, the participant set~$\rho'$ can still be synthesized directly from the syntax~$\rho.e$.
Rule~\ruleref{TI-App} synthesizes the type of applications by propagating the inferred function type to check the argument against.
The participants~$\rho_2$ of the argument, however, must be inferred.
This rule also demonstrates how to add the annotation~$\appchor{\rho'}$ on applications.

This system soundly provides an algorithmic checker for our surface-level type system.
However, completeness requires some developer-provided annotations.
For instance, we cannot always infer the opposite branch type~$\tau_2$ for~\InlN or the latent participants~$\rho$ of a function.
The following syntax of \emph{annotated terms} contains all of the additional hints needed to guarantee completeness;
the unlisted constructs require no additional annotations, matching exactly the surface syntax of Figure~\ref{fig:abstract-syntax}.
\begin{syntax}
    \rulefiguresize
    \category{C}
    \alternative{\Inl{\rho;\tau_2}{C}}
    \alternative{\FunAnnot{}{F}{X \ty \tau_1}{\tau_2 \mathbin{\locs} \rho}{C}}
    \alternative{\TFunAnnot{}{F}{\alpha \knd \kappa}{\tau \mathbin{\locs} \rho}{C}}
    \alternative{\Fold{\mu_\rho \alpha \ldotp \tau}{C}}
    \alternative{\cdots}
\end{syntax}

This algorithmic formulation is mechanically verified in Lean to be computable, sound on all terms, and complete with respect to the declarative system on fully annotated terms.

\begin{thm}[Algorithmic Typing is Sound and Complete]
\label{thm:bid-soundComplete}
  If~$\chortypedsynthany{\ctx}{C}{\tau}{\rho}{C'}$,
  then~$\chortyped{\ctx}{C}{\tau}{\rho}$,
  and~$C'$ is a fully elaborated version of~$C$.
  Conversely, if~$\chortyped{\ctx}{C}{\tau}{\rho}$ and~$C$ is fully annotated,
  then there exists~$C'$ where~$\chortypedsynthany{\ctx}{C}{\tau}{\rho}{C'}$.
\end{thm}

\begin{thm}[Algorithmic Typing is Computable]
\label{thm:bid-compute}
    Given~$\ctx$, a fully annotated~$C$, and inputs for $\checkarrow$~positions,
    it is decidable whether there exist~$C'$ and outputs for $\syntharrow$~positions where
    $\chortypedsynthany{\ctx}{C}{\tau}{\rho}{C'}$.
\end{thm}

\section{Network Language}
\label{sec:network-lang}
To compile a choreography into a concurrent system,
we need to specify the target language that the system will run.
This \emph{network language} describes the actions of each location in the system,
and gives a concurrent operational semantics to describe the execution of the entire system.

\subsection{Network Language Syntax}
\label{sec:network-syntax}
\begin{figure}[t]
  \begin{syntax}
    \category[{\smash{\parbox[t]{\widthof{Programs}}{\centering Network Programs}}}]{E}
    \alternative{X} \alternative{\Ret{e}} \alternative{\NtwkUnit} \alternative{\SendTo{E}{\rho}} \alternative{\RecvFrom{\ell}} \\
    \alternative{E_1 \NtwkSeq E_2}
    \alternative{\NtwkFun{F}{X}{E}} \alternative{E_1~E_2} \alternative{\NtwkTFun{F}{\alpha}{E}} \alternative{E~t} \\
    \alternative{\NtwkLetIn{x}{E_1}{E_2}} \alternative{\NtwkLetIn{\alpha \knd \kappa}{E_1}{E_2}} \\
    \alternative{\AllowOneChoice{\ell}{d}{E}}
    \alternative{\AllowChoice{\ell}{E_1}{E_2}}\\
    \alternative{\ChooseFor{d}{\rho}{E}}
    \alternative{\AmIIn{\rho}{E_1}{E_2}} \\
    \alternative{\PID{\alpha}{E}}
    \alternative{\NtwkFork{(\alpha,x)}{E_1}{E_2}} \alternative{\NtwkExit} \alternative{\cdots}
    \category[Systems]{\Pi}
    \alternative{L_1 \triangleright E_1 \mathrel{\parallel} \ldots \mathrel{\parallel} L_n \triangleright E_n}\\
  \end{syntax}

  \caption{Selected Network Program Syntax. Here $L \in \Locations$ is a concrete location name.}
  \label{fig:network-lang-syntax}
\end{figure}
The network language (Figure~\ref{fig:network-lang-syntax}) is a concurrent $\lambda$-calculus with messages from the same space as in choreographies---local values and selection messages.
Its syntax closely mirrors the choreographic syntax, but splits message-passing into two distinct constructs for the sender and recipient(s).

The return expression~$\Ret{e}$ executes and yields the output of a local expression~$e$,
mirroring the choreography~$\rho.e$.
To account for a location~$L \notin \rho$---who should not execute~$e$---and other scenarios
where a location should do nothing or does not know a value,
we include an inert unit value~$\NtwkUnit$.

Message passing uses two separate constructs for the sender and recipient(s) of a message.
Specifically, $\SendTo{E}{\rho}$ multicasts the result of~$E$ to every location in~$\rho$
(and also yields the result at the sender),
while~$\RecvFrom{\ell}$ waits to receive a local value from~$\ell$.
Since only local values can be sent, \SendN may only send a value~$\Ret{v}$,
while other forms such as $\SendTo{\NtwkUnit}{\rho}$ will be stuck.

(Type) let-expressions, (type) functions, \FontNtwk{case}, \FontNtwk{localCase}, and algebraic and recursive data types
(data type operations are omitted from Figure~\ref{fig:network-lang-syntax}) are included as in choreographies,
less some unneeded location (set) annotations.
The sequencing construct $E_1 \NtwkSeq E_2$ is standard.

Two further constructs are standard in (process-polymorphic) choreographies.
First, recall that a location can participate in a~\CaseN or~\LocalCaseN if it either knows the data being branched on \emph{or selection messages are inserted telling it which way to go}.
In the latter case, the location branches on the result of waiting for a selection message,
which we represent using~\AllowChoiceN expressions:
$\AllowChoice{\ell}{E_1}{E_2}$ waits for a selection message from $\ell$,
continuing as~$E_1$ if it is $\Left$ and as~$E_2$ otherwise.
When an isolated selection statement appears in a choreography,
the recipient statically knows what message will be received,
so we provide an~\AllowChoiceN with only one branch;
such a term receiving the wrong branch direction becomes stuck.
The counterpart~\ChooseForN construct simply multicasts the chosen direction~$d$ to~$\rho$.

Next, the ``AmI-In''~expression, introduced by \lamqc as a generalization of PolyChor$\lambda$'s ``AmI''~\cite{GraversenHM24},
intuitively represents a process's knowledge of its own name.
In particular, $\AmIIn{\rho}{E_1}{E_2}$ continues as~$E_1$ if the process running it is in~$\rho$, and as~$E_2$ otherwise.
EPP uses $\AmIInN$ to implement process polymorphism,
since, for instance, the behavior of a location~$L$ when executing its part of~$\alpha.(1 \LocalPlus 2)$ depends on whether~$\alpha$ resolves to a location set containing~$L$.
In order to simplify notation, we write~$\AmI{\ell}{E_1}{E_2}$ as a shorthand for~$\AmIIn{\{\ell\}}{E_1}{E_2}$.

We introduce three novel features to the network language: $\PIDN$, $\NtwkForkN$, and $\NtwkExit$.
Similar to~$\AmIInN$, the construct~$\PID{\alpha}{E}$ reflects a process's knowledge of its own name,
allowing it to be bound to a location variable~$\alpha$ in~$E$.
The~\NtwkForkN and \NtwkExit~expressions implement process spawning.
Specifically, $\NtwkFork{(\alpha,x)}{E_1}{E_2}$ spawns a thread with network code~$E_1$, called the \emph{thread task}.
The thread name is bound to~$\alpha$, and its local representation to~$x$, as in the choreographic~\ForkN.
Note that, unlike in a choreography, the thread task is explicitly specified, rather than implicit from scoping.
Dually, the \NtwkExit~command halts execution, removing the location from the system.

\subsection{Network Language Operational Semantics}
\label{sec:network-semantics}

\begin{figure}[t]
  \rulefiguresize
  \begin{syntax}
    \category[Transition Labels]{l}
    \alternative{\iota}
    \alternative{\RSendNtwk{m}{\rho}}
    \alternative{\RRecvNtwk{L}{m}}
    \alternative{\RForkNtwk{L}{E}}
    \alternative{\RExitNtwk}
  \end{syntax}

  \begin{ruleset}
    \SetRuleLabelLoc{left}
    \NRetRule
    \and
    \NAppRule
    \and
    \NSendRule
    \and
    \NRecvRule
    \and
    \NChooseRule
    \and
    \toggletrue{NAllowLLinebreak}
    \NAllowLRule
    \and
    \toggletrue{NForkLinebreak}
    \NForkRule
    \and
    \NExitRule
  \end{ruleset}
  \caption{Selected Network Language Operational Semantics}
  \label{fig:ntwk-semantics}
\end{figure}

The labeled transition system
$L \triangleright E_1 \ntwkstep{\smash{\scriptstyle l}} E_2$
gives the isolated single-location operational semantics of the network language,
where~$L$ is the location executing the program and~$l$ is the label on the step.
Selected transition labels and rules are shown in Figure~\ref{fig:ntwk-semantics}.

There are five forms of transition labels, corresponding to five different sorts of steps.
First, the ``iota'' label~$\iota$ denotes an internal step,
in which the network program reduces without interacting with other locations,
such as in reducing a local program (\ruleref{N-Ret})
or $\beta$-reduction (\ruleref{N-App}).

The send~($\RSendNtwk{m}{\rho}$) and receive~($\RRecvNtwk{L}{m}$) labels account for message-passing steps, including selection messages.
As recipients cannot know the contents of a message in advance, the rules \ruleref{N-Recv} and \ruleref{N-AllowL} (and the omitted \textsc{N-AllowR})
are non-deterministic, allowing any value to arrive.
The sender follows \ruleref{N-Send} and \ruleref{N-Choose}, which ensure that the contents of the label match the message and recipients specified by the program.
The system semantics defined below ensures the sender and recipient agree on the message, resolving the recipient's non-determinism.

The label~$\RForkNtwk{L}{E}$ indicates the spawning of a new thread~$L$ with task~$E$.
Unlike the choreographic redex~$\RFork{L}{L'}$, this label contains~$E$
to ensure the child is spawned with the correct code.
While \ruleref{N-Fork} chooses the new thread's name non-deterministically,
the system semantics ensures there is no collision with an existing name.
The dual label~$\NtwkExit$ denotes when a thread is killed;
\ruleref{N-Exit} has no special effect here, but the system semantics removes the thread from the system.

\subsubsection{Network Systems}
\label{sec:system-semantics}

\begin{figure}[t]
  \rulefiguresize
  \begin{syntax}
    \category[System Transition Labels]{l_S}
    \alternative{\iota_L}
    \alternative{L_1.m \sendsto \rho}
    \alternative{\RForkSys{L}{L'}{E}}
    \alternative{\RKillSys{L}}
  \end{syntax}
  \begin{align*}
    &\fbox{$\loc{l_S}$} \\
    &\loc{\iota_L} = \loc{\RForkSys{L}{L'}{E}} = \loc{\RKillSys{L}} \defeq \{L\} &&
    \loc{L_1.m \sendsto \rho} \defeq \{L_1\} \cup \rho
  \end{align*}
  \smallskip
  \begin{ruleset}
    \Rule{Internal}{
      L \triangleright \Pi(L) \ntwkstep{\iota} E}
    {\Pi \systemstep[\iota_L] \subst{\Pi}{L}{E}}
    \and
    \Rule{Comm}{
      L_1 \notin \rho \\
      L_1 \triangleright \Pi(L_1) \ntwkstep{\RSendNtwk{m}{\rho}} E_1 \\\\
      \forall L \in \rho\ldotp \Big(L \triangleright \Pi(L) \ntwkstep{\RRecvNtwk{L_1}{m}} E_L\Big) \\
    }{\Pi \systemstep[L_1.m \sendsto \rho] \subst*{\Pi}{{L_1}{E_1}{L \in \rho}{E_L}}}
    \and
    \Rule{Fork}{
      L'~\text{globally fresh} \\
      \fv{E_1} = \varnothing \\
      L \triangleright \Pi(L) \ntwkstep{\RForkNtwk{L'}{E_1}} E_2
    }{\Pi \systemstep[\RForkSys{L}{L'}{E_1}] \subst*{\Pi}{{L'}{(E_1 \NtwkSeq \NtwkExit)}{L}{E_2}}}
    \and
    \and
    \Rule{Kill}{
      L \triangleright \Pi(L) \ntwkstep{\RExitNtwk} E
    }{\Pi \systemstep[\RKillSys{L}] \Pi \setminus L}
  \end{ruleset}
  \caption{System Semantics and Labels}
  \label{fig:system-semantics}
\end{figure}

Since a network program represents the isolated execution of a single program at a given location,
while a choreography represents an entire concurrent system,
we need to lift the semantics of network programs to model an entire system.
Formally, we represent a \emph{system} $\Pi = {\parallel_{L \in \SysLocs} (L \triangleright E_L)}$
as a map from each location~$L$ in a finite set $\SysLocs \subset \Locations$ of running locations to the network program~$E_L$ it is currently executing.

The operational semantics of systems are shown in Figure~\ref{fig:system-semantics}.
These rules lift the single-location semantics into a concurrent composition using four rules.
The \ruleref{Internal} rule allows one location to independently take an internal step.
\ruleref{Comm} models message passing and requires the sender and all recipients to simultaneously step with the same message value.
\ruleref{Fork} spawns a new child thread with a globally fresh name~$L'$ and ensures the child runs the (closed) code specified by the parent in its step label.
Finally, \ruleref{Kill} kills a thread by removing it from the system.
Notationally, $\subst{\Pi}{L \in \rho}{E_L}$ denotes the updated system mapping~$L$ to~$E_L$ if $L \in \rho$ and $\Pi(L)$ otherwise,
and $\Pi \setminus L$ denotes~$\Pi$ with~$L$ removed from its domain.

\section{Endpoint Projection}
\label{sec:endpoint-projection}
With the target language in hand, we now define and prove correct the endpoint projection~(EPP) procedure
to compile a choreography into a system of concurrently executing network programs.

\subsection{Defining Projection}
\label{sec:epp-defined}

\begin{figure}[t]
  \rulefiguresize
  $\epp{X_\rho}{\ell} \defeq
    \begin{cases}
      X & \ifText \ell \in \rho\\
      \NtwkUnit & \owText
    \end{cases}$
  \hfill
  $\epp{\rho.e}{\ell} \defeq
    \begin{cases}
      \Ret{e} & \ifText \ell \in \rho\\
      \NtwkUnit & \owText
    \end{cases}$
  \hfill
  $\epp{C_1 \appchor{\rho} C_2}{\ell} \defeq \begin{cases}
    \epp{C_1}{\ell}~\epp{C_2}{\ell} & \ifText \ell \in \rho \\
    \epp{C_1}{\ell} \seqfun \epp{C_2}{\ell} \seqfun \NtwkUnit & \owText
  \end{cases}$
  \begin{align*}
     \epp{C \ChorSend[\ell'] \rho}{\ell} & \defeq
     \begin{cases}
       \SendTo{\epp{C}{\ell}}{\rho} & \ifText \ell = \ell'\\
       \epp{C}{\ell} \seqfun \RecvFrom{\ell} & \ifText \ell \neq \ell' \andText \ell \in \rho\\
       \epp{C}{\ell} & \owText
     \end{cases}
     \\[\vertrulegap]
     \epp{\syncs{\ell'}{d}{\rho} \seq C}{\ell} & \defeq
     \begin{cases}
       \ChooseFor{d}{\rho}{\epp{C}{\ell}} & \ifText \ell = \ell' \\
       \AllowOneChoice{\ell}{d}{\epp{C}{\ell}} & \ifText \ell \neq \ell' \andText \ell \in \rho \\
       \epp{C}{\ell} & \owText
     \end{cases}
     \\[\vertrulegap]
    \epp{\LocalCase*{\rho}{C}{x}{C_1}{y}{C_2}}{\ell} & \defeq
    \begin{cases}
      \mathrlap{\NtwkLocalCase{\epp{C}{\ell}}{x}{\epp{C_1}{\ell}}{y}{\epp{C_2}{\ell}} \quad\ifText \ell \in \rho} & \\
      \epp{C}{\ell} \seqfun (\epp{C_1}{\ell} \Merge \epp{C_2}{\ell}) & \ifText \ell \notin \rho \andText x \notin \fv{\epp{C_1}{\ell}} \andText y \notin \fv{\epp{C_2}{\ell}} \\
      \Undef & \owText
    \end{cases}
    \\[\vertrulegap]
    \epp{\Fork{(\alpha,x)}{\ell'}{C}}{\ell} & \defeq
    \begin{cases}
      \NtwkFork{(\alpha,x)}{\epp{C}{\alpha}}{\epp{C}{\ell}} & \ifText \ell = \ell' \\
      \epp{C}{\ell} & \ifText \ell \neq \ell' \andText \alpha, x \notin \fv{\epp{C}{\ell}} \\
      \Undef & \owText
    \end{cases}
    \\[\vertrulegap]
    \epp{\KillAfter{L}{C}}{\ell} & \defeq
    \begin{cases}
      \Undef & \ifText \ell = L \\
      \epp{C}{\ell} & \owText
    \end{cases}
  \end{align*}
  \caption{Selected EPP Definitions}
  \label{fig:epp-defs}
\end{figure}

The endpoint projection~$\epp{C}{\ell}$ is the network program for~$\ell$'s computation in~$C$.
EPP is a partial function over elaborated choreographic syntax
that exists only when~$\ell$ knows all relevant branch conditions and binds all variables it uses.
Its definition largely follows \lamqc~\citep{SamuelsonHC25},
mapping most terms to their network equivalents.
EPP in~\lamfork must also account for annotated participants in functions and applications
and for as-yet-unspawned threads in function bodies (Section~\ref{sec:epp-functions} below).
Figure~\ref{fig:epp-defs} presents illustrative cases.

Recall from Section~\ref{sec:type-checking} that
elaboration adds a location set~$\rho$ to choreographic variables and function applications,
indicating who knows about (part of) the value or application.
For variables~$X_\rho$ and local computations~$\rho.e$,
only locations in~$\rho$ retain anything; others project to~\NtwkUnit.
At application sites, involved parties apply the function,
while others simply sequence the function and argument before returning~\NtwkUnit.
For data and selection messages, the sender projects to~\SendN and~\ChooseForN, respectively,
while the recipients project to~\RecvN and~\AllowChoiceN.

For~$\LocalCaseN_\rho$, all locations execute the guard,
but only those in~$\rho$ branch.
For other locations, the merge operator $E_1 \Merge E_2$ combines the branches.
As is standard~\citep{Montesi23}, merge is a partial function
that collects disjoint~\AllowChoiceN branches, recursively merges shared ones, and is otherwise structurally homomorphic.
For instance,
\[\rulefiguresize
  \begin{array}{@{}l@{}}
    \left(\AllowChoice*{\ell}{E_1}{E_2}\right)
    \Merge
    \left(\AllowOneChoice*{\ell}{\NtwkLeft}{E_1'}\right)
    \defeq
    \AllowChoice*{\ell}{E_1 \Merge E_1'}{E_2}
  \end{array}
\]
 Its full definition can be found in Appendix~\ref{sec:ntwk-merge-def}. 
Since non-branching locations do not know the scrutinee,
EPP also checks that~$x$ and~$y$ are not free in the projected branches for these locations.
Projecting \CaseN~expressions is analogous.

To project the new $\Fork{(\alpha,x)}{\ell'}{C}$, the corresponding network~\NtwkForkN requires two pieces of code:
the thread task for~$\alpha$ and the body for~$\ell'$.
The former is the projection of~$C$ for~$\alpha$,
while the latter is the projection of~$C$ for~$\ell'$.
The parent thus projects~$C$ twice, once for each role.
Other locations simply project~$C$ and ensure neither~$\alpha$ nor~$x$ is free.

For~\KillAfterN, everyone projects the body except the thread~$L$, whose result is~\Undef.
Since~$L$ has access to the~\KillAfterN body, \emph{but not any surrounding scope},
EPP must project only this code.
We implement this behavior below with a separate EPP function
built on top of this one.

Note that these rules do not directly employ the sequencing primitive~$E_1 \NtwkSeq E_2$,
but instead use a \emph{collapsing sequencing function}~$E_1 \seqfun E_2$~\citep{SamuelsonHC25},
which returns~$E_2$ when~$E_1$ is a value, and otherwise returns~$E_1 \NtwkSeq E_2$.
This both increases expressivity by allowing more choreographies to project,
and ensures projected programs can correctly simulate out-of-order choreographic steps.

\subsubsection{Projecting (Type) Functions}
\label{sec:epp-functions}
Projecting (type) functions is more complex, as thread tasks may reference (choreographic) variables bound outside the \ForkN.
In the projected network system, these will resolve to values in the \emph{parent process's projection}.
The parent's values must thus work for both parent and child.
To understand the challenge, consider the following choreography.
\begin{equation}\label{eqn:proj-ex}
  \rulefiguresize
  \LetMany{    {F}{\LamN_\anyLoc X \ldotp \Case{}{X}{Y}{\Bob.1 \ColSend \Alice}{\_}{\{\Alice,\Bob\}.0}}    {Z}{\Inl{\anyLoc}{\Alice.42}}    {\alpha}{\Alice.\ForkN()}}{F~Z}
\end{equation}
Here everyone, including~$\alpha$, applies~$F$ to~$Z$.
They should all branch left, and \Bob~should then send~$1$ to~\Alice and return without~$\alpha$ doing more.
However, if we simply defined $\epp{\LamN_\anyLoc X \ldotp C}{\ell} = \NtwkLamN X \ldotp \epp{C}{\ell}$, as in prior work~\citep{HirschG22,SamuelsonHC25},
the projected system would misbehave.
In~\Alice's projection, the thread task for~$\alpha$ is simply~$F~Z$,
so~$\alpha$ receives~\Alice's projection of~$F$.
Using the na\"ive definition above, $\alpha$ would (correctly) branch left,
but then incorrectly run~\Alice's projection~$\RecvFrom{\Bob}$ of the left \CaseN~branch.
As~\Bob only sends to~\Alice, $\alpha$~will deadlock waiting on~\Bob, breaking a foundational guarantee of choreographies.

The problem is that the projection of~$F$ contains only code for~\Alice, which differs from the code needed for~$\alpha$.
We therefore project each (type) function twice:
once for the projecting location~$\ell$,
and again for any thread~$\ell$ might spawn \emph{in the future}.
As in prior work on process polymorphism~\citep{GraversenHM24,SamuelsonHC25},
the result is a function that might be executed by multiple different locations,
so we use the same~\AmIN construct to select between these projections:
one branch when~$\ell$ applies the function, and the other when a spawned thread does.
In the latter, \PIDN~binds the thread's name, allowing it to reference the name of a thread to be spawned later.
We formally define projection of functions as follows.
\[\rulefiguresize
  \begin{array}{@{}l@{}}   \setbox0=\hbox{$\ifText \ell \in \rho$}  \setbox1=\hbox{$\ifText \ell \notin \rho \andText \epp{C}{\ell} ~\text{is defined}$}\epp{\Fun{\rho}{F}{X}{C}}{\ell} \defeq \begin{cases}
    \NtwkFun{F}{X}{\AmI{\ell}{\epp{C}{\ell}}{(\PID{\alpha}{\epp{C}{\alpha}})}} & \usebox0 \\
    \NtwkUnit & \makebox[\wd0][r]{\usebox1} \\
    \Undef & \makebox[\wd0][r]{\makebox[\wd1][l]{$\owText$}}
  \end{cases}
  \end{array}
\]

While involved locations include code for potential child processes,
uninvolved locations~($\ell \notin \rho$) project to~\NtwkUnit, as they and their children never apply the function.
However, we still require that the body's projection~$\epp{C}{\ell}$ be \emph{defined}.
This requirement may seem unnecessary; if~$\ell$ does not participate in the body,
one might hope that the body would always project, preferably to~\NtwkUnit.
Unfortunately, this is not so. For instance, consider the function
$G = \LamN_\Alice X \ldotp (\ITE[\Alice]{X}{\{\Alice,\Bob\}.1}{\{\Alice,\Bob\}.2}) \seq \Alice.3$
whose computation requires only~\Alice but whose body fails to project for~\Bob.
Applying~$G$ results in a choreography that does not project for~\Bob.
To prevent an initially projectable program from stepping to a non-projectable one,
we instead say~$G$ does not project for~\Bob, as its body fails to project.

Finally, process-polymorphic functions must behave differently
depending on whether the function's argument~$\alpha$ resolves to the executing location~$\ell$.
Following prior work~\citep{GraversenHM24,SamuelsonHC25}, we use~\AmIN to branch on the value of~$\alpha$ in the function body.
In the \NtwkThen~branch, when~$\alpha = \ell$, we substitute~$\alpha$ with~$\ell$ before projection;
in the \NtwkElse~branch we project directly, which correctly treats~$\alpha \neq \ell$.
The same concern as with regular functions also arises, so we layer the solutions.
An outer~\AmIN distinguishes parent from child,
and each branch performs the standard process-polymorphic location resolution, replacing~$\ell$ with~$\beta$ when a thread executes.
\[\rulefiguresize
  \begin{array}{@{}l@{}} \addtocounter{numlevels}{1}
  \epp{\TFunLoc{F}{\alpha \knd \locKnd}{C}}{\ell} \defeq
    \NtwkTFun{F}{\alpha}{
      \AmI*{\ell}
        {(\AmI{\alpha}{\epp{\subst{C}{\alpha}{\ell}}{\ell}}{\epp{C}{\ell}})}
        {(\PID{\beta}{(\AmI{\alpha}{\epp{\subst{C}{\alpha}{\beta}}{\beta}}{\epp{C}{\beta}})})}
    }
  \addtocounter{numlevels}{-1}
  \end{array}
\]

\subsubsection{Projecting Systems and Active Threads}
\label{sec:epp-threads}
As mentioned above, EPP must be handled separately for spawned threads.
Non-ephemeral locations compile the entire choreography before executing it,
so we can simply take the projection as defined above.
For spawned threads, however, we must restrict projection to the subexpression provided by the parent.
Specifically, we define a modified EPP~$\eppfork{C}{L}$ below
which projects a spawned thread using only its corresponding~\KillAfterN body:
\[\rulefiguresize
  \begin{array}{@{}l@{}} \eppfork{C}{L} \defeq \begin{cases}
    (\epp{C_1}{L} \ntwkmerge \cdots \ntwkmerge \epp{C_n}{L}) \seqfun \NtwkExit &
      \ifText C ~\text{contains}~n \geq 1~\text{subterms}~\KillAfter{L}{C_1}, \ldots, \KillAfter{L}{C_n} \\
    \epp{C}{L} & \owText
  \end{cases}
  \end{array}
\]
Because out-of-order execution allows threads to be spawned in both branches of a conditional,
the~\KillAfterN corresponding to a thread is only unique \emph{up to branching}.
Thus, $\eppfork{C}{L}$ extracts and merges all~\KillAfterN bodies corresponding to thread~$L$, terminating the merged body with an~\NtwkExit.
 The full definition of~$\eppfork{C}{L}$ can be found in Appendix~\ref{sec:proj-def-thread}.

To compile the full system, we then lift this modified EPP point-wise over each running location in~$\SysLocs$:
$\eppfork{C}{\SysLocs} \defeq {\parallel_{L \in \SysLocs} (L \triangleright \eppfork{C}{L})}$.
Note that $\eppfork{C}{L}$ must be defined for all $L \in \SysLocs$ for $\eppfork{C}{\SysLocs}$ to be defined.

\iffalse
Matching the definition in Section~\ref{sec:network-lang} of a system of network programs running concurrently,
we aim to lift EPP point-wise to a finite set of locations $\SysLocs \subset \Locations$.
For locations with access to the full choreography, this approach works well.
For spawned locations, however, it fails to recognize that they only have code for their thread task, not the whole choreography.
For instance, after the reduction below, the newly-spawned location~\Bob should only have access to
the code in the first part of the pair corresponding to the body of the~\KillAfterN.
\[
  C_1 = \left(\left(\Fork*{\alpha}{\Alice}{\alpha.(1 \LocalPlus 1) \ColSend \Alice}\right), \TLamN \ell \knd \locKnd \ldotp \ell.3\right)_\top
  \step[]
  \left(\left(\KillAfter*{\Bob}{(\Bob.(1 \LocalPlus 1) \ColSend \Alice)}\right), \TLamN \ell \knd \locKnd \ldotp \ell.3\right)_\top = C_2
\]
However, the single-location EPP above incorrectly projects the choreography after the step
to $\epp{C_2}{\Bob} = (\SendTo{\Ret{1 \LocalPlus 1}}{\Alice} \NtwkSeq \NtwkExit, \NtwkTLamN \ell \ldotp \AmI{\ell}{\Ret{3}}{\NtwkUnit})$ for~\Bob,
including the entire pair rather than just~\Bob's part of the first component.

To account for this, we modify EPP to only extract the body of~\KillAfterN expressions when projecting currently-spawned threads.
The modified definition, denoted $\eppfork{C}{L}$, is defined as follows.
\[
  \eppfork{C}{L} = \begin{cases}
    \epp{C'}{\ell} \seqfun \NtwkExit & \ifText C ~\text{contains a unique subterm}~\KillAfter{L}{C'} \\
    \epp{C}{\ell} & \owText.
  \end{cases}
\]
In the example above, this modified EPP will maintain the connection between the choreography and its projection
by correctly projecting $\eppfork{C_2}{\Bob} = \SendTo{\Ret{1 \LocalPlus 1}}{\Alice} \NtwkSeq \NtwkExit$.

To project a full system, we can now lift projection point-wise to each running location using this modified EPP.
That is, $\eppfork{C}{\SysLocs} = {\parallel_{L \in \SysLocs} (L \triangleright \eppfork{C}{L})}$ for a finite $\SysLocs \subset \Locations$.
Note that $\eppfork{C}{L}$ must be defined for all $L \in \SysLocs$ for $\eppfork{C}{\SysLocs}$ to be defined.
\fi

\subsection{Completeness, Weak Soundness, and Deadlock Freedom}
\label{sec:proj-correct}
We have provided two semantics for \lamfork:
the top-level choreographic semantics (Section~\ref{sec:semantics}),
and the semantics given by EPP.
We now show these semantics are deeply related,
and use this relationship to translate choreographic type soundness to a deadlock-freedom guarantee for compiled systems.

\paragraph{Simulation Relation}
\label{sec:bisim}
To relate the two semantics, we must decide which systems relate to which choreographies.
One might expect a choreography~$C$ to only relate to its projection~$\eppfork{C}{\SysLocs}$,
but reductions fail to preserve this property for three main reasons.
First, choreographic branching discards the branch not taken,
while the projection for non-branching locations retains it while awaiting a selection message.
Second, collapsing sequencing~$E_1 \seqfun E_2$ discards evaluated terms from the projection.
Finally, in the projected system, a value substituted into a thread task uses the parent's projection, which syntactically differs from the child's own projection of the value.

We account for these differences via a relation~$E_1 \lessthan[\ell] E_2$ on network programs.
Such a relation is standard~\citep{Montesi23}, but we additionally index ours by the current location~$\ell$.
Specifically, we define $E_1 \lessthan[\ell] E_2$ to be the smallest structurally compatible partial order admitting the following rules.
\begin{mathpar}[\rulefiguresize]
  \infer{ }{\AllowOneChoice{\ell}{d}{E_d} \lessthan[\ell] \AllowChoice{\ell}{E_\Left}{E_\Right}}
  \and
  \infer{ }{E \lessthan[\ell] V \NtwkSeq E}
  \and
  \infer{ }{\NtwkUnit \lessthan[\ell] V}
  \\
  \infer{ }{E[\alpha \mapsto \ell] \lessthan[\ell] \PID{\alpha}{E}}
  \and
  \infer{
    E_1 \lessthan[\ell] E_2' \\
    \ell \neq \ell'
  }{
    {\begin{array}{@{}r@{}}
      \AmI{\ell}{E_1}{E_2} \lessthan[\ell] \quad \\
      \AmI{\ell'}{E_1'}{E_2'}
    \end{array}}
  }
  \and
  \infer{E_1 \lessthan[\alpha] E_1' \\
    E_2 \lessthan[\ell] E_2'
  }{
    {\begin{array}{@{}r@{}}
      \NtwkFork{(\alpha,x)}{E_1}{E_2} \lessthan[\ell] \quad \\
      \NtwkFork{(\alpha,x)}{E_1'}{E_2'}
    \end{array}}
  }
\end{mathpar}
The first two rules address branching and collapsing sequencing, respectively, as in prior work~\citep{HirschG22,Montesi23,SamuelsonHC25}.
The remaining rules address the parent--child discrepancy.
To understand the third rule, consider the example in Equation~(\ref{eqn:proj-ex}) above.
In the projected system, \Alice~substitutes their projection of~$F$ and~$Z$ into the thread task,
but projecting the choreography after the fork gives~$\alpha$ its own projections.
For~$Z$, these are $\NtwkInl{\Ret{42}}$ and $\NtwkInl{\NtwkUnit}$, respectively;
\Alice's~projection has extra (extraneous) information that~$\alpha$ neither needs nor inspects.
The rule $\NtwkUnit \lessthan[\ell] V$ lets us ignore it.

The~\PIDN and \AmIN~rules handle the function~$F$ and use the location parameter~$\ell$.
Recall that the projection of~$F$ for~\Alice is $\NtwkLamN X \ldotp \AmI{\Alice}{\epp{C}{\Alice}}{(\PID{\beta}{\epp{C}{\beta}})}$,
whereas the projection for~$\alpha$ is $\NtwkLamN X \ldotp \AmI{\alpha}{\epp{C}{\alpha}}{(\PID{\beta}{\epp{C}{\beta}})}$,
where~$C$ is the body of~$F$.
The network code for~$\alpha$ is \emph{in the~\NtwkElse branch of~\Alice's projection but the~\NtwkThen branch of~$\alpha$'s}.
Thus, when considering~$\alpha$'s view, we must relate those branches
when the process in the~\AmIN switches from~\Alice to~$\alpha$,
and resolve the~\PIDN in the~\NtwkElse branch of~\Alice's projection.
Together, the~\PIDN and~\AmIN rules properly relate the projection of the stepped choreography
to the result of stepping the projected system.

The final \NtwkForkN~rule is the only rule shown where the index~$\ell$ changes
(the congruence rule for~\PIDN likewise compares the bodies at the bound variable).
Specifically, since~$\alpha$ runs the thread tasks~$E_1$ and~$E_1'$,
this variable must index the relation when comparing these programs.

This relation is then lifted point-wise to entire systems as
$\Pi_1 \lessthan \Pi_2 ~\defeq~ \forall L \in \SysLocs \ldotp \Pi_1(L) \lessthan[L] \Pi_2(L)$,
which suffices to prove that the projected semantics simulate the choreographic semantics.

Since EPP is partial and typing does not guarantee projectability,
the theorems below assume~$\eppfork{C}{\Omega}$ is defined,
which a compiler checks by running EPP;
completeness shows that steps preserve this.
\begin{restatable}[EPP Completeness]{thm}{completenessSimple}
  \label{thm:simple-completeness}
  If $\choremptyped{C}{\tau}{\rho}$ and~$\nl{\rho} \subseteq \Omega$,
  then whenever $\conf*{C} \steps{} \conf{C'}{\Omega'}$,
  there is some $\Pi'$ such that $\eppfork{C}{\Omega} \systemsteps{} \Pi'$ and $\eppfork{C'}{\Omega'} \lessthan \Pi'$.
\end{restatable}

In addition to completeness, one might expect a corresponding soundness result.
However, a standard soundness theorem is precluded by the combination of divergence and multiply-located computations.
For instance, in $\{\Alice,\Bob\}.\Loop \seq \{\Bob,\Client\}.(1 \LocalPlus 2)$, 
the choreography can only execute the~\Loop,
whereas~\Client can independently evaluate~$1 \LocalPlus 2$ in the projected system.
Prior work circumvents this by requiring global synchronization at function boundaries and
forbidding non-terminating local computation~\citep{HirschG22,SamuelsonHC25}.
A similar approach fails for~\lamfork:
global synchronization would require locations to synchronize with threads they may not know about,
while choreographic functions let some locations loop as others proceed,
forcing such functions to be forbidden.

Rather than imposing these artificial restrictions,
we establish a robust \emph{divergence-weakened soundness} property that tolerates the divergence described above.
This property guarantees that the choreography can match any system step
that is not blocked at the choreographic level by a computation that has not yet finished in the system.
In the above example, \Client~computing~$3$ is just such a blocked step, as~$\{\Alice,\Bob\}.\Loop$ has not completed---and never will---at the system level.
Furthermore, since one choreographic step may correspond to several system steps,
the projection may be ``ahead'' at some locations and ``behind'' at others,
but it will only be behind at locations where the next choreographic step is blocked.
For instance, in~$\{\Alice,\Bob\}.\Loop \seq \{\Bob,\Client\}.(1 \LocalPlus 2)$,
if~\Alice steps~\Loop and~\Client computes~$3$,
but~\Bob does nothing, the choreography cannot match~\Client's step,
and while it can match~\Alice's steps, doing so will place it ahead at~\Bob.

To formalize this notion, we say if $\eppfork{C}{\Omega} \systemsteps \Pi$,
then~$C$ must step to some~$C'$ whose projection is behind~$\Pi$ \emph{only at locations that cannot step~$C'$},
but might be ahead of~$\Pi$ elsewhere.

\begin{thm}[Divergence-Weakened Soundness]\label{thm:partial-sound}
  If $\choremptyped{C}{\tau}{\rho}$,
  $\nl{\rho} \subseteq \SysLocs$,
  and~$\eppfork{C}{\Omega} \systemsteps \Pi$,
  then there exist~$C'$, $\Omega'$, $\Pi'$, $\Pi''$, and label trace~$\trc$ such that
  
    \begin{enumerate}[leftmargin=*,label={\arabic*.},ref={\arabic*}]
      \item\label{li:sound:c-step}
        $\conf{C}{\Omega} \steps{} \conf{C'}{\Omega'}$ and $\eppfork{C}{\Omega} \systemsteps \Pi'$ with $\eppfork{C'}{\Omega'} \lessthan \Pi'$,
      \item\label{li:sound:confl}
        $\Pi \systemsteps \Pi''$ and $\Pi' \systemsteps[\trc] \Pi''$, and
      \item\label{li:sound:disjoint}
        for all~$C''$, $\Omega''$, and~$R$, if $\conf{C'}{\Omega'} \step[\smash{R}] \conf{C''}{\Omega''}$, then  $\rloc{R} \cap \loc{\trc} = \varnothing$.
    \end{enumerate}
  
\end{thm}

The following diagram visualizes the relationships provided by conditions~\ref{li:sound:c-step}~and~\ref{li:sound:confl}.
\begin{center}
  \newlength{\PiPrimeShiftLen}
  \setlength{\PiPrimeShiftLen}{\widthof{$\Pi'$}}
  \addtolength{\PiPrimeShiftLen}{4pt}
  \begin{tikzpicture}[every node/.style={inner sep=4pt}]
    \node (C) {$\conf{C}{\Omega}$};
    \node[below=1.5em of C] (C') {$\conf{C'}{\Omega'}$};

    \node[right=2.5em of C] (CProj) {$\eppfork{C}{\Omega}$};
    \node[anchor=west] (Pi') at (C'-|CProj.west) {$\eppfork{C'}{\Omega'} \lessthan \Pi'$};

    \node[right=6em of CProj] (Pi) {$\Pi$};
    \node[right=6em of Pi'] (Pi'') {$\Pi''$};

    \draw[mapsto] (C) -- (CProj) node[label,midway,above,yshift=1.5pt]{$\scriptstyle\eppfork{\cdot}{}$};
    \draw[mapsto] (C') -- (Pi') node[label,midway,above,yshift=1.5pt]{$\scriptstyle\eppfork{\cdot}{}$};
    \draw[-implies,double equal sign distance] (CProj) -- (Pi) node[pos=1,above,yshift=-2pt]{$\scriptstyle*$};     \draw[dashed,-implies,double equal sign distance] (C) -- (C') node[pos=1,right,xshift=-2pt]{$\scriptstyle*$};     \draw[dashed,-implies,double equal sign distance] (CProj) -- ($(Pi'.east|-Pi''.north)-(\PiPrimeShiftLen,0pt)$) node[pos=1,above right,xshift=-3pt,yshift=-3pt]{$\scriptstyle*$};     \draw[dashed,-implies,double equal sign distance] (Pi) -- (Pi'') node[pos=1,above right,xshift=-3pt,yshift=-3pt]{$\scriptstyle*$};     \draw[dashed,-implies,double equal sign distance] (Pi') -- (Pi'') node[midway,above]{$\scriptstyle \trc$} node[pos=1,above,yshift=-2pt]{$\scriptstyle*$}; \end{tikzpicture}
\end{center}

Conditions~\ref{li:sound:c-step} and~\ref{li:sound:confl} allow~$\Pi$ and~$\Pi'$ to converge to a common system~$\Pi''$, accounting for locations where $C'$ is ahead of or behind $\Pi$.
Crucially, condition~\ref{li:sound:disjoint} enforces that~$\Pi'$---and thus~$C'$---can only be behind~$\Pi$ at locations where~$C'$ is blocked and cannot (yet) step.

Although weaker than prior results~\citep{HirschG22,SamuelsonHC25},
this property holds with multiply-located computations and no implicit synchronization,
and suffices for deadlock freedom.
Specifically, by combining Theorem~\ref{thm:partial-sound},
EPP completeness,
and choreographic type soundness,
we arrive at our main result.

\begin{restatable}[Deadlock Freedom]{thm}{deadlockFreedom}
\label{thm:deadlock-freedom}
  If $\choremptyped{C}{\tau}{\rho}$ and $\nl{\rho} \subseteq \Omega$,
  then whenever $\eppfork{C}{\Omega} \systemsteps \Pi$,
  either~$\Pi$ is final---every location in~$\Pi$ is a value---or~$\Pi$ can step.
\end{restatable}

\begin{proof}[Proof Sketch]
  Begin by applying Theorem~\ref{thm:partial-sound} to get~$C'$, $\Omega'$, $\Pi'$, $\Pi''$, and~$\trc$, as specified in the theorem.
  Since $\Pi \systemsteps \Pi''$, if it requires one or more steps, then~$\Pi$ steps and we are done.
  Otherwise, we get $\Pi = \Pi''$.
  Type soundness ensures either~$C'$ is a value or $\conf{C'}{\Omega'}$ can step.
  If~$C'$ is a value, then~$\eppfork{C'}{\Omega'}$ is final.
  Since~$\eppfork{C'}{\Omega'} \lessthan \Pi'$, $\Pi'$~can step to a final system, so~$\Pi = \Pi''$ can too.

  Finally, if $\conf{C'}{\Omega'} \step[{\raisebox{-1pt}[2pt]{$\scriptstyle R$}}] \conf{C''}{\Omega''}$,
  EPP completeness ensures that~$\Pi'$ can step with system label(s) corresponding to~$R$.
  However, condition~(\ref{li:sound:disjoint}) of Theorem~\ref{thm:partial-sound} guarantees that $\rloc{R} \cap \loc{t} = \varnothing$.
  That is, the step(s) from~$R$ involve a disjoint set of locations from those taking~$\Pi'$ to~$\Pi''$,
  so the parallel confluence lemma says that we can take them in either order to get the same result.
  Thus,~$\Pi = \Pi''$ can take the $R$ step(s) and is not deadlocked.
\end{proof}

\section{Related Work}
\label{sec:related-work}

While \lamfork is the first functional choreographic language with process forking,
it builds on a rich literature of choreographic programming
and relates to the extensive work on behavioral type systems including multiparty session types.

\paragraph{Functional Choreographic Programming}
\label{sec:related-choreographies}
Early work on choreographic programming focused on lower-order imperative programming~\cite[see e.g.,][]{CarboneM13,Montesi13,CarboneMS14,Cruz-FilipeMP18,Cruz-FilipeM17,LaneseMZ13,CruzFilipeM17c}.
Pirouette~\cite{HirschG22} and Chor$\lambda$~\cite{CruzFilipeGLMP22} were the first \emph{functional} choreographic languages, allowing higher-order programming.
\citet{BatesKJSKN25} extended Chor$\lambda$ with multiply-located values,
\citet{GraversenHM24} introduced process polymorphism,
and~\citet{SamuelsonHC25} developed \lamqc---the language \lamfork primarily extends---to provide process\emph{-set} polymorphism and first-class location names.

Pirouette and \lamqc relied on global synchronization at every function boundary to prove projection sound.
\citet{CruzFilipeGLMP23} provided an out-of-order semantics for Chor$\lambda$,
but relied on rewrite rules that fail in the presence of named recursive functions~\citep{SamuelsonHC25}.
\citet{PlyukhinQM25} introduced \emph{semilenient} semantics,
using principled rewrite rules to convert choreographies into an $A$-normal form~\cite{FlanaganSDF93} on the fly,
though it is unclear how to generalize this strategy to multiply-located computations.
We instead introduce \emph{divergence-weakened soundness},
which holds for~\lamfork's multiply-located computations while remaining sufficient to prove deadlock freedom.

\paragraph{Process Spawning in Choreographies}
\label{sec:related-chor-spawning}
Two previous (lower-order) choreographic languages consider process spawning.
\citet{CarboneM13} allow a fixed set of top-level service threads
to wait for a notification from a parent process to begin execution.
\citet{CruzFilipeM17c} provide a calculus tailored to parallel divide-and-conquer algorithms with dynamic process spawning,
but without first-class functions or variable binding;
parent processes can communicate the names of spawned threads, but names are otherwise second class.
Like~\lamfork, they use a type system to ensure processes are introduced before they can communicate,
but it is \emph{behavioral} and sidesteps the challenges addressed here,
such as closing over dead process names and invoking closures defined outside a thread's lexical scope.

\paragraph{Process Spawning in Behavioral Type Systems}
\label{sec:related-spawning}
Process calculi have long featured dynamic process and channel creation~\cite{MilnerPW92,Milner80},
and a vast literature on behavioral types has emerged to ensure deadlock freedom in these settings~\cite{Gay22}.
Type systems for binary session calculi and linear $\pi$-calculi ensure communication safety between two endpoints~\cite{GayV10,CairesP10,Wadler12}.
To permit dynamic thread and channel creation or cyclic network topologies without deadlocks,
binary systems typically rely on channel obligations and capabilities~\cite{Kobayashi06,Padovani14}
or priority-based session types~\cite{Dardha18,VanDenHeuvel24}.
With level polymorphism or recursion, these systems can type recursive processes in cyclic networks of unbounded size,
such as Milner's cyclic scheduler~\cite{Milner80,Padovani14,VanDenHeuvel24}.
Of these systems, only session types extend to the multiparty scenarios choreographies support~\cite{HondaVK98}.

Multiparty session types (MPSTs)~\cite{HondaYC08} extend session types to interactions among more than two participants,
but dynamic process and session spawning remain challenging.
Traditional frameworks separate session initialization from thread spawning and parallel composition,
precluding new processes from joining active sessions and requiring global consistency checks when sessions interleave~\cite{Bettini08,Coppo16}.
Parameterised MPSTs~\cite{DenielouYBH12} choose the number of participants at runtime but fix it at session start,
while dynamic multirole session types~\cite{DenielouY11} let participants join and leave.
\citet{LeBrunFD25} use replication for a limited form of intra-session process spawning,
while \maty~\cite{FowlerH26} adapts MPSTs to actors with spawning and supervision.

Multiparty GV (MPGV)~\cite{JacobsBK22} adds MPST to a linear functional language via an $n$-ary \texttt{fork} combining spawning and session creation.
It permits cyclic sessions with static arities,
but requires tree-structured inter-session interactions,
simplifying type-checking at the cost of expressivity.

Both~\lamfork and behavioral type systems rule out deadlocks and enable process spawning, but neither subsumes the other.
Behavioral systems directly reason about concurrency with types,
typically fixing a session's structure statically.
Choreographies reason with the program itself:
which threads exist, who communicates, and in what order can depend on arbitrary computation.
Deadlock freedom emerges from the structure of EPP combined with \emph{sequential} type soundness.
As a result, \lamfork can express, for instance, a ring of dynamic size inspired by Milner's cyclic scheduler without special support
(our Lean artifact types, projects, and runs it),
whereas behavioral systems needed tailored features to type it.
Conversely, choreographies fix each process's order of interactions,
so they cannot inherently express nondeterministic contention, such as clients acquiring a shared server in whatever order they arrive~\cite{BalzerTP19},
or check independently written endpoints against a common protocol~\cite{CastroHJNY19}.
Adding such features to choreographies would be valuable future work.

\section{Conclusion}
\label{sec:conclusion}

We have introduced \lamfork, the first functional choreographic language to support process forking.
We integrated process spawning with higher-order programming, process polymorphism, multiply-located computations, and first-class process names.
Implemented na\"ively, this combination can introduce bugs,
such as calling a closure over a killed thread's name, or a thread executing incorrect code for a closure.
\lamfork prevents these bugs by tracking participants in its type system
and dynamically selecting code depending on who executes a closure,
guaranteeing deadlock freedom.

\section*{Use of Generative AI}

Harmonic's Aristotle AI and Anthropic's Claude Sonnet 5.5 AI generated the Lean artifact~\citep{SamuelsonZHC25}
from paper proofs and a Rocq-based type inference and elaboration procedure written entirely by the human authors.
In the process, the AI identified some errors in the original definitions and proofs.
The human authors corrected these flaws before using the AI to complete the Lean formalization.

\section*{Data-Availability Statement}
This work contains an artifact providing Lean~4 proofs of all our theorem statements,
including a verified implementation of the type inference and elaboration pass described in Section~\ref{sec:type-checking}.
The Lean development is available on Zenodo~\citep{SamuelsonZHC25}.

The development contains no \texttt{sorry} and declares no axioms of its own.
Auditing every main theorem with \texttt{\#print axioms} shows that each depends only on Lean's three standard axioms:
\begin{itemize}
  \item \emph{Propositional extensionality} (\texttt{propext}): logically equivalent propositions are equal.
  \item \emph{Quotient soundness} (\texttt{Quot.sound}): elements related in a quotient type become equal in it, which also yields function extensionality.
  \item \emph{Classical choice} (\texttt{Classical.choice}): a type known to be nonempty has an element, which yields the law of excluded middle and thus classical reasoning.
\end{itemize}
These axioms are standard: they are built into Lean's core library, are assumed throughout Mathlib~\citep{Mathlib20} (which our development imports),
and are the accepted foundation of classical mathematics, so they do not weaken the correspondence between our theorem statements and the paper's claims.

\section*{Acknowledgments}

We would like to thank Rahul Krishnan for help editing.
Support for this research was provided by the University of Wisconsin--Madison
Office of the Vice Chancellor for Research with funding from the Wisconsin Alumni Research Foundation.

\bibliography{main}

\appendix
\section*{Appendices}
\section{Choreography Operational Semantics}
\label{sec:full-chor-semantics}
\subsection{Choreography Syntax}
\begin{syntax}
  \category[Selection Labels]{d}
  \alternative{\Left}
  \alternative{\Right}

  \category[Choreographies]{C}
  \alternative{X} \alternative{\rho.e}
  \alternative{\LetIn{\rho.x}{C_1}{C_2}}
  \alternative{\LetIn{\rho.\alpha \knd \kappa}{C_1}{C_2}}\\
  \alternative{C \ChorSend[\ell] \rho}
  \alternative{\syncs{\ell}{d}{\rho} \seq C} \\
  \alternative{\Fun{\rho}{F}{X}{C}}
  \alternative{C_1~C_2}
  \alternative{\TFun{\rho}{F}{\alpha \knd \kappa}{C}}
  \alternative{C~t} \\
  \alternative{\LocalCase{}{C}{x}{C_1}{y}{C_2}}\\
  \alternative{\Inl{\rho}{C}}
  \alternative{\Inr{\rho}{C}}
  \alternative{\Case{}{C}{X}{C_1}{Y}{C_2}} \\
  \alternative{\FoldN~C}
  \alternative{\UnfoldN~C}
  \alternative{(C_1,C_2)}
  \alternative{\FstN~C}
  \alternative{\SndN~C} \\
  \alternative{\Fork{(\alpha,x)}{\ell}{C}}
  \alternative{\KillAfter{L}{C}}
  \\
  \category[Values]{V}
  \alternative{\rho.v} \alternative{\Fun{\rho}{F}{X}{C}} \alternative{\TFun{\rho}{F}{\alpha \knd \kappa}{C}}\\
  \alternative{(V_1,V_2)_\rho} \alternative{\Inl{\rho}{V}} \alternative{\Inr{\rho}{V}} \alternative{\Fold{\rho}{V}}
\end{syntax}

\subsection{Elaborated Choreographic Syntax}
\label{sec:elaborated-syntax}
The annotations on application, type application, injections, cases, folds,
unfolds, pairs, projections, and local computations record the locations that perform the corresponding operation.
The annotation~$\rho$ on a choreographic variable~$X_\rho$ records the locations that know its value.
\begin{syntax}
  \category[Elaborated Forms]{C}
  \alternative{X_\rho} \alternative{\rho.e}
  \alternative{\LetIn{\rho.x}{C_1}{C_2}}
  \alternative{\LetIn{\rho.\alpha \knd \kappa}{C_1}{C_2}}\\
  \alternative{C \ChorSend[\ell] \rho}
  \alternative{\syncs{\ell}{d}{\rho} \seq C}\\
  \alternative{\Fun{\rho}{F}{X}{C}}
  \alternative{C_1 \appchor{\rho} C_2}
  \alternative{\TFun{\rho}{F}{\alpha \knd \kappa}{C}}
  \alternative{C \appchor{\rho} t}\\
  \alternative{\LocalCase{\rho}{C}{x}{C_1}{y}{C_2}}\\
  \alternative{\Inl{\rho}{C}} \alternative{\Inr{\rho}{C}}
  \alternative{\Case{\rho}{C}{X}{C_1}{Y}{C_2}}\\
  \alternative{\Fold{\rho}{C}} \alternative{\Unfold{\rho}{C}}
  \alternative{(C_1,C_2)_\rho}
  \alternative{\Fst{\rho}{C}} \alternative{\Snd{\rho}{C}}\\
  \alternative{\Fork{(\alpha,x)}{\ell}{C}}
  \alternative{\KillAfter{L}{C}}
\end{syntax}

\subsection{Fully Annotated Choreographic Syntax}
\label{sec:fully-annotated-syntax}
\begin{syntax}
  \rulefiguresize
  \category[Annotated Forms]{C}
  \alternative{X} \alternative{\rho.e_{\mathsf{fa}}}
  \alternative{\LetIn{\rho.x}{C_1}{C_2}}
  \alternative{\LetIn{\rho.\alpha \knd \kappa}{C_1}{C_2}}\\
  \alternative{C \ChorSend[\ell] \rho}
  \alternative{\syncs{\ell}{d}{\rho} \seq C}\\
  \alternative{\FunAnnot{}{F}{X \ty \tau_1}{\tau_2 \mathbin{\locs} \rho}{C}}
  \alternative{C_1~C_2}
  \alternative{\TFunAnnot{}{F}{\alpha \knd \kappa}{\tau \mathbin{\locs} \rho}{C}}
  \alternative{C~t}\\
  \alternative{\LocalCase{}{C}{x}{C_1}{y}{C_2}}\\
  \alternative{\Inl{\rho;\tau_2}{C}}
  \alternative{\Inr{\rho;\tau_1}{C}}
  \alternative{\Case{}{C}{X}{C_1}{Y}{C_2}}\\
  \alternative{\Fold{\rho;\tau}{C}}
  \alternative{\UnfoldN~C}
  \alternative{(C_1,C_2)}
  \alternative{\FstN~C}
  \alternative{\SndN~C}\\
  \alternative{\Fork{(\alpha,x)}{\ell}{C}}
\end{syntax}
Here~$e_{\mathsf{fa}}$ is a fully annotated local expression.
The function and type-function forms provide the argument/body types and latent participant sets.
The injection forms carry the type of the other branch, and the fold form carries the recursive body type.
These annotations ensure completeness of the inference procedure.

\subsection{Redices and Evaluation Contexts}
\begin{syntax}
  \category[Messages]{m}
  \alternative{v} \alternative{d}

  \category[Redices]{R}
  \alternative{\RDone{\rho}{e_1}{e_2}}
  \alternative{\RFun{R}}
  \alternative{\RArg{R}}
  \alternative{\RApp{\rho}}
  \alternative{\RTApp{\rho}}
  \alternative{\RUnfoldFold{\rho}} \\
  \alternative{\RPairL{R}}
  \alternative{\RPairR{R}}
  \alternative{\RFstPair{\rho}}
  \alternative{\RSndPair{\rho}}
  \alternative{\RCaseInl{\rho}}
  \alternative{\RCaseInr{\rho}} \\
  \alternative{\RLocalCaseInl{\rho}}
  \alternative{\RLocalCaseInr{\rho}} \\
  \alternative{\RLet{\rho}{v}}
  \alternative{\RLet{\rho}{t}}
  \alternative{\RSendV{L}{m}{\rho}}
  \alternative{\RFork{L_1}{L_2}}
  \alternative{\RKill{L}} \\

  \category[Evaluation Contexts]{\eta}
  \alternative{\hole}
  \alternative{\eta \appchor{\rho} C}
  \alternative{V \appchor{\rho} \eta}
  \alternative{\eta \appchor{\rho} t}
  \alternative{\Fold{\rho}{\eta}}
  \alternative{\Unfold{\rho}{\eta}}\\
  \alternative{(\eta,C)_\rho}
  \alternative{(V,\eta)_\rho}
  \alternative{\Fst{\rho}{\eta}}
  \alternative{\Snd{\rho}{\eta}} \\
  \alternative{\Inl{\rho}{\eta}}
  \alternative{\Inr{\rho}{\eta}}
  \alternative{\Case{\rho}{\eta}{X}{C_1}{Y}{C_2}}\\
  \alternative{\LocalCase{\rho}{\eta}{x}{C_1}{y}{C_2}}\\
  \alternative{\LetIn{\rho.x}{\eta}{C_2}}
  \alternative{\LetIn{\rho.\alpha \knd \kappa}{\eta}{C_2}}\\
  \alternative{\eta \ChorSend[\ell] \rho}
  \alternative{\KillAfter{L}{\eta}}
\end{syntax}

\subsection{Redex Blocked Locations}
\begin{mathparpagebreakable}
  \rloc{\RDone{\rho}{e_1}{e_2}} = \rho \and
  \rloc{\RFun{R}} = \rloc{R} \and
  \rloc{\RArg{R}} = \rloc{R} \and
  \rloc{\RApp{\rho}} = \rho \and
  \rloc{\RTApp{\rho}} = \rho \and
  \rloc{\RUnfoldFold{\rho}} = \rho \and
  \rloc{\RPairL{R}} = \rloc{R} \and
  \rloc{\RPairR{R}} = \rloc{R} \and
  \rloc{\RFstPair{\rho}} = \rho \and
  \rloc{\RSndPair{\rho}} = \rho \and
  \rloc{\RCaseInl{\rho}} = \rho \and
  \rloc{\RCaseInr{\rho}} = \rho \and
  \rloc{\RLocalCaseInl{\rho}} = \rho \and
  \rloc{\RLocalCaseInr{\rho}} = \rho \and
  \rloc{\RLet{\rho}{v}} = \rho \and
  \rloc{\RLet{\rho}{t}} = \rho \and
  \rloc{\RSendV{L}{m}{\rho}} = \{L\} \cup \rho \and
  \rloc{\RFork{L_1}{L_2}} = \{L_1\} \and
  \rloc{\RKill{L}} = \{L\} \and
\end{mathparpagebreakable}

\subsection{Choreography Blocked Locations}
\begin{mathparpagebreakable}
  \cloc{X} =\varnothing
  \and
  \cloc{\rho.e} = \rho
  \and
  \cloc{\Fun{\rho}{F}{X}{C}} = \varnothing
  \and
  \cloc{C_1 \appchor{\rho} C_2} = \cloc{C_1} \cup \cloc{C_2} \cup \rho
  \and
  \cloc{\TFun{\rho}{F}{\alpha \knd \kappa}{C}} = \varnothing
  \and
  \cloc{C \appchor{\rho} t} = \cloc{C} \cup \rho
  \and
  \cloc{\Fold{\rho}{C}} = \cloc{C}
  \and
  \cloc{\Unfold{\rho}{C}} = \cloc{C} \cup \rho
  \and
  \cloc{(C_1,C_2)_\rho} = \cloc{C_1} \cup \cloc{C_2}
  \and
  \cloc{\Fst{\rho}{C}} = \cloc{C} \cup \rho
  \and
  \cloc{\Snd{\rho}{C}} = \cloc{C} \cup \rho
  \and
  \cloc{\Inl{\rho}{C}} = \cloc{C}
  \and
  \cloc{\Inr{\rho}{C}} = \cloc{C}
  \and
  \cloc{\Case{\rho}{C}{X}{C_1}{Y}{C_2}} = \cloc{C} \cup \cloc{C_1} \cup \cloc{C_2} \cup \rho
  \and
  \cloc{\LocalCase{\rho}{C}{x}{C_1}{y}{C_2}} = \cloc{C} \cup \cloc{C_1} \cup \cloc{C_2} \cup \rho
  \and
  \cloc{\LetIn{\rho.x}{C_1}{C_2}} = \cloc{C_1} \cup \cloc{C_2} \cup \rho
  \and
  \cloc{\LetIn{\rho.\alpha \knd \kappa}{C_1}{C_2}} = \cloc{C_1} \cup (\cloc{C_2} \setminus \alpha) \cup \rho
  \and
  \cloc{C \ChorSend[\ell] \rho} = \cloc{C} \cup \{\ell\} \cup \rho 
  \and
  \cloc{\syncs{\ell}{d}{\rho} \seq C} = \{\ell\} \cup \rho \cup \cloc{C}
  \and
  \cloc{\Fork{(\alpha,x)}{\ell}{C}} = \{\ell\} \cup (\cloc{C} \setminus \alpha)
  \and
  \cloc{\KillAfter{L}{C}} = \{L\} \cup \cloc{C}
\end{mathparpagebreakable}

\subsection{Redex for an Evaluation Context}
If $\eta$ is an evaluation context and $R$ is a redex, we define $\eta[R]$ to be the redex which corresponds to making the reduction given by $R$ in the context $\eta$.
\begin{mathparpagebreakable}
    \ctxredex{\hole} \defeq R
    \and
    \ctxredex{(\eta \appchor{\rho} C)} \defeq \RFun{\ctxredex{\eta}}
    \and
    \ctxredex{(V \appchor{\rho} \eta)} \defeq \RArg{\ctxredex{\eta}}
    \and
    \ctxredex{(\eta \appchor{\rho} t)} \defeq \ctxredex{\eta}
    \and
    \ctxredex{(\Fold{\rho}{\eta})} = \ctxredex{(\Unfold{\rho}{\eta})} \defeq \ctxredex{\eta}
    \and
    \ctxredex{(\eta,C)_\rho} \defeq \RPairL{\ctxredex{\eta}}
    \and
    \ctxredex{(V,\eta)_\rho} \defeq \RPairR{\ctxredex{\eta}}
    \and
    \ctxredex{(\Fst{\rho}{\eta})} = \ctxredex{(\Snd{\rho}{\eta})} \defeq \ctxredex{\eta}
    \and
    \ctxredex{(\Inl{\rho}{\eta})} = \ctxredex{(\Inr{\rho}{\eta})} \defeq \ctxredex{\eta}
    \and
    \ctxredex{(\Case{\rho}{\eta}{X}{C_1}{Y}{C_2})} \defeq \ctxredex{\eta}
    \and
    \ctxredex{(\LocalCase{\rho}{\eta}{x}{C_1}{y}{C_2})} \defeq \ctxredex{\eta}
    \and
    \ctxredex{(\eta \ChorSend[\ell] \rho)} \defeq \ctxredex{\eta}
    \and
    \ctxredex{(\LetIn{\rho.x}{\eta}{C})} = \ctxredex{(\LetIn{\rho.\alpha \knd \kappa}{\eta}{C})} \defeq \ctxredex{\eta}
    \and
    \ctxredex{\KillAfter{L}{\eta}} \defeq \ctxredex{\eta}
\end{mathparpagebreakable}

\subsection{Location Set Operations}
\label{sec:set-relations}
Here we define the containment~$\ell \in \rho$, disjointness~$\rho_1 \cap \rho_2 = \varnothing$, and subset $\rho_1 \subseteq \rho_2$ relations;
and the set-difference function~$\rho_1 \setminus \rho_2$,
with special care given to how they are defined when the locations and sets in question are open terms.

The principle for how the containment and subset relations behave is that of a modality of \emph{necessity}.
For instance, the containment relation~$\ell \in \rho$ holds only if the location~$\ell$ is an element of~$\rho$ for \emph{all} values that their variables could resolve to.
Note here that the metavariable~$\ell$ stands for either a type variable~$\alpha$ or a concrete location~$L \in \Locations$, and the metavariable~$\rho$ stands for any location set, including possibly a type variable.
\begin{mathparpagebreakable}
  \infer{~}{\ell \in \{\ell\}}
  \and
  \infer{\ell \in \rho_1}{\ell \in \rho_1 \cup \rho_2}
  \and
  \infer{\ell \in \rho_2}{\ell \in \rho_1 \cup \rho_2}
  \and
  \infer{~}{\ell \in \anyLoc}
\end{mathparpagebreakable}

The disjointness relation is defined as expected:
\[(\rho_1 \cap \rho_2 = \varnothing) \defeq \forall \ell \ldotp \neg (\ell \in \rho_1 \wedge \ell \in \rho_2) \]

To define the subset relation, we first note that we cannot use the na\"{i}ve definition in terms of the containment relation.
That is, $\forall \ell \ldotp \ell \in \rho_1 \Rightarrow \ell \in \rho_2$ would not serve as a correct definition for $\rho_1 \subseteq \rho_2$ in the presence of type variables.
Specifically, because~$\ell \notin \alpha$ for every location~$\ell$, with this definition we would have that~$\alpha \subseteq \rho$ for every set~$\rho$.
The relation should have the modality of necessity (i.e., be preserved under substitution),
but this example shows that this is not the case with the na\"{i}ve definition.
Instead, the subset relation must be defined inductively as follows.
\begin{mathparpagebreakable}
  \infer{~}{\varnothing \subseteq \rho}
  \and
  \infer{~}{\rho \subseteq \rho}
  \and
  \infer{\ell \in \rho}{\{\ell\} \subseteq \rho}
  \and
  \infer{\rho \subseteq \rho_1}{\rho \subseteq \rho_1 \cup \rho_2}
  \and
  \infer{\rho \subseteq \rho_2}{\rho \subseteq \rho_1 \cup \rho_2}
  \and
  \infer{\rho_1 \subseteq \rho \\
  \rho_2 \subseteq \rho}
  {\rho_1 \cup \rho_2 \subseteq \rho}
  \and
  \infer{~}{\rho \subseteq \anyLoc}
\end{mathparpagebreakable}

Finally, the set subtraction function~$\rho_1 \setminus \rho_2$ is defined inductively on~$\rho_1$---taking advantage of the relations defined above---to remove all components from~$\rho_1$ that appear in~$\rho_2$.
\begin{mathparpagebreakable}
  \varnothing \setminus \rho = \varnothing
  \and
  \alpha \setminus \rho =
    \begin{cases}
      \varnothing & \alpha \subseteq \rho \\
      \alpha & \alpha \not\subseteq \rho
    \end{cases}
  \and
  \{\ell\} \setminus \rho =
    \begin{cases}
      \varnothing & \ell \in \rho \\
      \{\ell\} & \ell \notin \rho
    \end{cases}
  \and
  (\rho_1 \cup \rho_2) \setminus \rho = (\rho_1 \setminus \rho) \cup (\rho_2 \setminus \rho)
  \and
  \anyLoc \setminus \rho =
    \begin{cases}
        \varnothing & \anyLoc \subseteq \rho \\
        \anyLoc & \anyLoc \not\subseteq \rho
      \end{cases}
\end{mathparpagebreakable}

\subsection{Choreography Operational Semantics}
\label{sec:full-chor-sem}
\begin{mathparpagebreakable}[\rulefiguresize]
  \SetRuleLabelLoc{lab}
  \CCtxRule
  \and
  \CDoneRule
  \and
  \CAppRule
  \and
  \CTAppRule
  \and
  \CUnfoldFoldRule
  \and
  \CFstPairRule
  \and
  \CSndPairRule
  \and
  \CCaseInlRule
  \and
  \CCaseInrRule
  \and
  \CLocalCaseInlRule
  \and
  \CLocalCaseInrRule
  \and
  \CLetVRule
  \and
  \CTyLetVRule
  \and
  \CSendVRule
  \and
  \CSyncRule
  \and
  \CSyncIRule
  \and
  \CCaseIRule
  \and
  \CLocalCaseIRule
  \and
  \CAppIRule
  \and
  \CPairIRule
  \and
  \CLetIRule
  \and
  \CTyLetIRule
  \and
  \CForkRule
  \and
  \CForkIRule
  \and
  \CKillRule
  \and
  \CKillIRule
\end{mathparpagebreakable}

\section{Static Semantics}
\label{sec:full-types}

\subsection[Lambda-Fork Kinding System]{\lamfork Kinding System}
\label{sec:full_kinds}
First we note that, in order to prevent kind dependency, the kinding
context should be split into two contexts---$\Gamma_\ell$ for locations and location sets of kind $\kappa_\ell \in \{\locKnd, \setKnd, \finsetKnd\}$,
and $\Gamma$ for local and program kinds.
We have elided this detail from the presentation of the kinding and typing rules for simplicity;
the full development addresses these details separately.
Our kinding system is designed so that the subkinding relationship~$\finsetKnd <: \setKnd$
and the relevant subkinding rule---while not explicitly given---is admissible.
The finite-set kind is required exactly where an operation must enumerate or communicate with its participants:
the located-type constructor $t_e @ \rho$,
the participant annotations of location and location-set binders,
send and synchronization targets, and an instantiation of a type function's argument.
Ordinary participant bounds, function/application annotations, recursive types, and
the bound set in a location-polymorphic type need only have kind $\setKnd$.

\begin{mathparpagebreakable}[\rulefiguresize]
  \KVarRule
  \and
  \KSubVarRule
  \and
  \KLocRule
  \and
  \KSngRule
  \and
  \KSngFinRule
  \and
  \KUnionRule
  \and
  \KUnionFinRule
  \and
  \KAnySetRule
  \and
  \KLocalRule
  \and
  \KAtRule
  \and
  \KArrowRule
  \and
  \KProdRule
  \and
  \KSumRule
  \and
  \KRecRule
  \and
  \KAllLocRule
  \and
  \KAllRule
\end{mathparpagebreakable}

\subsection{Context Projection}
For a context $\ctx = \Gamma_\ell;\Gamma;\Delta_e;\Delta$, define its projection~$\proj{\ctx}{\rho}$
to a (possibly open) location set~$\rho$ to be the local context~$\proj{\ctx}{\rho} = \Gamma_\ell; \proj{\Gamma}{\rho}; \proj{\Delta_e}{\rho}$,
where~$\proj{\Gamma}{\rho}$ and~$\proj{\Delta_e}{\rho}$ are defined as follows.
\begin{mathparpagebreakable}[\rulefiguresize]
  \proj{\cdot}{\rho} \defeq \cdot
  \and
  \proj{(\Gamma, \alpha \knd \kappa_{\rho'})}{\rho} \defeq \proj{\Gamma}{\rho}
  \and
  \proj{(\Gamma, \alpha \knd \kappa_e)}{\rho} \defeq \proj{\Gamma}{\rho}, \alpha \knd \kappa_e
  \and
  \proj{(\Delta_e, \rho'.x \ty t_e)}{\rho} \defeq
  \begin{cases}
    \proj{\Delta_e}{\rho}, x \ty t_e & \rho \subseteq \rho' \\
    \proj{\Delta_e}{\rho} & \rho \not\subseteq \rho'
  \end{cases}
\end{mathparpagebreakable}

\subsection[Lambda-Fork Type System]{\lamfork Type System}
\label{sec:full_chor_types}
In these rules we denote $\ctx = \Gamma;\Delta_e;\Delta$ for brevity,
and abuse notation to add variables to and use the kinding judgment on the relevant sub-contexts of $\ctx$
as appropriate.
In the subsequent sections we may use the shorthand $\tloc{\tau}$ to denote the set $\rho$ of locations such that $\ctx \proves \tau \knd \tyknd{\rho}$.

\begin{mathparpagebreakable}[\rulefiguresize]
  \TVarRule
  \and
  \TDoneRule
  \and
  \TFunRule
  \and
  \TAppRule
  \and
  \TTFunLocRule
  \and
  \TTAppLocRule
  \and
  \TTFunRule
  \and
  \TTAppRule
  \and
  \TPairRule
  \and
  \TFstRule
  \and
  \TSndRule
  \and
  \TInlRule
  \and
  \TInrRule
  \and
  \TCaseRule
  \and
  \TLocalCaseRule
  \and
  \TFoldRule
  \and
  \TUnfoldRule
  \and
  \TLetLocalRule
  \and
  \TLetLocRule
  \and
  \TLetLocSetRule
  \and
  \TSendRule
  \and
  \TSyncRule
  \and
  \TForkRule
  \and
\end{mathparpagebreakable}

\subsection{Augmented Type System}
\label{sec:full_chor_wellscoped}
\begin{mathparpagebreakable}[\rulefiguresize]
  \SVarRule
  \and
  \SDoneRule
  \and
  \SFunRule
  \and
  \SAppRule
  \and
  \STFunRule
  \and
  \STAppRule
  \and
  \STFunLocRule
  \and
  \STAppLocRule
  \and
  \SPairRule
  \and
  \SFstRule
  \and
  \SSndRule
  \and
  \SInlRule
  \and
  \SInrRule
  \and
  \SCaseRule
  \and
  \SLocalCaseRule
  \and
  \SFoldRule
  \and
  \SUnfoldRule
  \and
  \SLetLocalRule
  \and
  \SLetLocRule
  \and
  \SLetLocSetRule
  \and
  \SSendRule
  \and
  \SSyncRule
  \and
  \SForkRule
  \and
  \SKillRule
\end{mathparpagebreakable}

\section{Network Language}
\subsection{Network Language Expressions}
\label{sec:full-ntwk-syntax}
\begin{syntax}
  \category[Network Program]{E}
  \alternative{X} \alternative{\NtwkUnit} \alternative{\Ret{e}} \alternative{E_1 \NtwkSeq E_2} \\
  \alternative{\NtwkFun{F}{X}{E}} \alternative{E_1~E_2} \alternative{\NtwkTFun{F}{\alpha}{E}} \alternative{E~t} \\
  \alternative{(E_1,E_2)} \alternative{\NtwkFst{E}} \alternative{\NtwkSnd{E}}\\
  \alternative{\NtwkInl{E}} \alternative{\NtwkInr{E}} \alternative{\NtwkCase{E}{X}{E_1}{Y}{E_2}}\\
  \alternative{\NtwkLocalCase{E}{x}{E_1}{y}{E_2}}\\
  \alternative{\NtwkFold{E}} \alternative{\NtwkUnfold{E}}\\
  \alternative{\SendTo{E}{\rho}} \alternative{\RecvFrom{\ell}} \\
  \alternative{\NtwkLetIn{x}{E_1}{E_2}} \alternative{\NtwkLetIn{\alpha \knd \kappa}{E_1}{E_2}} \\
  \alternative{\ChooseFor{d}{\rho}{E}} \\
  \alternative{\AllowChoice{\ell}{{E_1}_\bot}{{E_2}_\bot}} \\
  \alternative{\AmIIn{\rho}{E_1}{E_2}}
  \alternative{\PID{\alpha}{E}} \\
  \alternative{\NtwkFork{(\alpha,x)}{E_1}{E_2}}
  \alternative{\NtwkExit}

  \category[Network Values]{V}
  \alternative{X}
  \alternative{\NtwkUnit} \alternative{\Ret{v}} \alternative{\NtwkFun{F}{X}{E}}
  \alternative{\NtwkTFun{F}{\alpha}{E}}\\
  \alternative{(V_1,V_2)} \alternative{\NtwkInl{V}} \alternative{\NtwkInr{V}} \alternative{\NtwkFold{V}}
\end{syntax}

\subsection{Transition Labels and Evaluation Contexts}
\begin{syntax}
  \category[Transition Labels]{l}
  \alternative{\iota}
  \alternative{\RSendNtwk{m}{\rho}}
  \alternative{\RRecvNtwk{L}{m}}
  \alternative{\RForkNtwk{L}{E}}
  \alternative{\RExitNtwk}
  \\
  \category[Evaluation Contexts]{\eta}
  \alternative{\hole}
  \alternative{\eta \NtwkSeq E}
  \alternative{\eta~E}
  \alternative{V~\eta}
  \alternative{\eta~t}
  \alternative{\NtwkFold{\eta}}
  \alternative{\NtwkUnfold{\eta}}\\
  \alternative{(\eta,E)}
  \alternative{(V,\eta)}
  \alternative{\NtwkFst{\eta}}
  \alternative{\NtwkSnd{\eta}}
  \alternative{\NtwkInl{\eta}}
  \alternative{\NtwkInr{\eta}}\\
  \alternative{\NtwkCase{\eta}{X}{E_1}{Y}{E_2}}\\
  \alternative{\NtwkLocalCase{\eta}{x}{E_1}{y}{E_2}}\\
  \alternative{\SendTo{\eta}{\rho}}
  \alternative{\NtwkLetIn{x}{\eta}{E}}
  \alternative{\NtwkLetIn{\alpha \knd \kappa}{\eta}{E}}
\end{syntax}

\subsection{Network Language Operational Semantics}
\label{sec:full-ntwk-sem}
\begin{mathparpagebreakable}[\rulefiguresize]
  \NCtxRule
  \and
  \NRetRule
  \and
  \NSeqRule
  \and
  \NAppRule
  \and
  \NTAppRule
  \and
  \NUnfoldFoldRule
  \and
  \NFstPairRule
  \and
  \NSndPairRule
  \and
  \NCaseInlRule
  \and
  \NCaseInrRule
  \and
  \NLocalCaseInlRule
  \and
  \NLocalCaseInrRule
  \and
  \NLetRule
  \and
  \NTyLetRule
  \and
  \NSendRule
  \and
  \NRecvRule
  \and
  \NChooseRule
  \and
  \NAllowLRule
  \and
  \NAllowRRule
  \and
  \NIAmInRule
  \and
  \NIAmNotInRule
  \and
  \NPIDRule
  \and
  \NForkRule
  \and
  \NExitRule
\end{mathparpagebreakable}

\section{Compilation}
\subsection{Network Program Merging}
\label{sec:ntwk-merge-def}
We show the patterns for which $E_1 \ntwkmerge E_2$ is defined; if there is no matching pattern, then $E_1 \ntwkmerge E_2$ is undefined.
The first three patterns concern only the optional branches of~\AllowChoiceN;
if~$E_1$ or~$E_2$ is itself undefined (e.g., a failed projection), then so is $E_1 \ntwkmerge E_2$.
\begingroup
\allowdisplaybreaks
\rulefiguresize
\begin{align*}
  \NtwkNone \ntwkmerge \NtwkNone &\defeq \NtwkNone
  \\[\vertrulegap]
  \NtwkNone \ntwkmerge E_2 &\defeq E_2
  \\[\vertrulegap]
  E_1 \ntwkmerge \NtwkNone &\defeq E_1
  \\[\vertrulegap]
  X \ntwkmerge X &\defeq X
  \\[\vertrulegap]
  \NtwkUnit \ntwkmerge \NtwkUnit &\defeq \NtwkUnit
  \\[\vertrulegap]
  \Ret{e} \ntwkmerge \Ret{e} &\defeq \Ret{e}
  \\[\vertrulegap]
  (E_{1,1} \NtwkSeq E_{1,2}) \ntwkmerge (E_{2,1} \NtwkSeq E_{2,2}) &\defeq E_{1,1} \ntwkmerge E_{2,1} \NtwkSeq E_{1,2} \ntwkmerge E_{2,2}
  \\[\vertrulegap]
  (\NtwkFun{F}{X}{E_1}) \ntwkmerge (\NtwkFun{F}{X}{E_2}) &\defeq \NtwkFun{F}{X}{(E_1 \ntwkmerge E_2)}
  \\[\vertrulegap]
  (E_{1,1}~E_{1,2}) \ntwkmerge (E_{2,1}~E_{2,2}) &\defeq (E_{1,1} \ntwkmerge E_{2,1})~(E_{1,2} \ntwkmerge E_{2,2})
  \\[\vertrulegap]
  (\NtwkTFun{F}{\alpha}{E_1}) \ntwkmerge (\NtwkTFun{F}{\alpha}{E_2}) &\defeq \NtwkTFun{F}{\alpha}{(E_1 \ntwkmerge E_2)}
  \\[\vertrulegap]
  (E_1~t) \ntwkmerge (E_2~t) &\defeq (E_1 \ntwkmerge E_2)~t
  \\[\vertrulegap]
  (\NtwkFold{E_1}) \ntwkmerge (\NtwkFold{E_2}) &\defeq \NtwkFold{(E_1 \ntwkmerge E_2)}
  \\[\vertrulegap]
  (\NtwkUnfold{E_1}) \ntwkmerge (\NtwkUnfold{E_2}) &\defeq \NtwkUnfold{(E_1 \ntwkmerge E_2)}
  \\[\vertrulegap]
  (E_{1,1},E_{1,2}) \ntwkmerge (E_{2,1},E_{2,2}) &\defeq ((E_{1,1} \ntwkmerge E_{2,1}),(E_{1,2} \ntwkmerge E_{2,2}))
  \\[\vertrulegap]
  (\NtwkFst{E_1}) \ntwkmerge (\NtwkFst{E_2}) &\defeq \NtwkFst{(E_1 \ntwkmerge E_2)}
  \\[\vertrulegap]
  (\NtwkSnd{E_1}) \ntwkmerge (\NtwkSnd{E_2}) &\defeq \NtwkSnd{(E_1 \ntwkmerge E_2)}
  \\[\vertrulegap]
  (\NtwkInl{E_1}) \ntwkmerge (\NtwkInl{E_2}) &\defeq \NtwkInl{(E_1 \ntwkmerge E_2)}
  \\[\vertrulegap]
  (\NtwkInr{E_1}) \ntwkmerge (\NtwkInr{E_2}) &\defeq \NtwkInr{(E_1 \ntwkmerge E_2)}
  \\[\vertrulegap]
  \left(\NtwkCase*{E_{1,1}}{X}{E_{1,2}}{Y}{E_{1,3}}\right) \ntwkmerge \left(\NtwkCase*{E_{2,1}}{X}{E_{2,2}}{Y}{E_{2,3}}\right) &\defeq \NtwkCase*{(E_{1,1} \ntwkmerge E_{2,1})}{X}{E_{1,2} \ntwkmerge E_{2,2}}{Y}{E_{1,3} \ntwkmerge E_{2,3}}
  \\[\vertrulegap]
  \left(\NtwkLocalCase*{E_{1,1}}{x}{E_{1,2}}{y}{E_{1,3}}\right) \ntwkmerge \left(\NtwkLocalCase*{E_{2,1}}{x}{E_{2,2}}{y}{E_{2,3}}\right) &\defeq \NtwkLocalCase*{(E_{1,1} \ntwkmerge E_{2,1})}{x}{E_{1,2} \ntwkmerge E_{2,2}}{y}{E_{1,3} \ntwkmerge E_{2,3}}
  \\[\vertrulegap]
  (\NtwkLetIn{x}{E_{1,1}}{E_{1,2}}) \ntwkmerge (\NtwkLetIn{x}{E_{2,1}}{E_{2,2}}) &\defeq \NtwkLetIn{x}{(E_{1,1} \ntwkmerge E_{2,1})}{(E_{1,2} \ntwkmerge E_{2,2})}
  \\[\vertrulegap]
  (\NtwkLetIn{\alpha \knd \kappa}{E_{1,1}}{E_{1,2}}) \ntwkmerge (\NtwkLetIn{\alpha \knd \kappa}{E_{2,1}}{E_{2,2}}) &\defeq \NtwkLetIn{\alpha \knd \kappa}{(E_{1,1} \ntwkmerge E_{2,1})}{(E_{1,2} \ntwkmerge E_{2,2})}
  \\[\vertrulegap]
  (\SendTo{E_1}{\rho}) \ntwkmerge (\SendTo{E_2}{\rho}) &\defeq \SendTo{(E_1 \ntwkmerge E_2)}{\rho}
  \\[\vertrulegap]
  (\RecvFrom{\ell}) \ntwkmerge (\RecvFrom{\ell}) &\defeq \RecvFrom{\ell}
  \\[\vertrulegap]
  (\ChooseFor{d}{\ell}{E_1}) \ntwkmerge (\ChooseFor{d}{\ell}{E_2}) &\defeq \ChooseFor{d}{\ell}{(E_1 \ntwkmerge E_2)}
  \\[\vertrulegap]
  \left(\AllowChoice*{\ell}{E_{1,1}}{E_{1,2}}\right) \ntwkmerge \left(\AllowChoice*{\ell}{E_{2,1}}{E_{2,2}}\right) &\defeq \AllowChoice*{\ell}{E_{1,1} \ntwkmerge E_{2,1}}{E_{1,2} \ntwkmerge E_{2,2}}
  \\[\vertrulegap]
  \AllowOneChoice{\ell}{\Left}{E_1} \ntwkmerge \AllowChoice*{\ell}{E_1'}{E_2'} &\defeq
    \AllowChoice*{\ell}{E_1 \ntwkmerge E_1'}{E_2'}
  \\[\vertrulegap]
  \AllowChoice*{\ell}{E_1'}{E_2'} \ntwkmerge \AllowOneChoice{\ell}{\Left}{E_1} &\defeq
    \AllowChoice*{\ell}{E_1' \ntwkmerge E_1}{E_2'}
  \\[\vertrulegap]
  \AllowOneChoice{\ell}{\Right}{E_2} \ntwkmerge \AllowChoice*{\ell}{E_1'}{E_2'} &\defeq
    \AllowChoice*{\ell}{E_1'}{E_2 \ntwkmerge E_2'}
  \\[\vertrulegap]
  \AllowChoice*{\ell}{E_1'}{E_2'} \ntwkmerge \AllowOneChoice{\ell}{\Right}{E_2} &\defeq
    \AllowChoice*{\ell}{E_1'}{E_2' \ntwkmerge E_2}
  \\[\vertrulegap]
  \left(\AmIIn*{\rho}{E_{1,1}}{E_{1,2}}\right) \ntwkmerge \left(\AmIIn*{\rho}{E_{2,1}}{E_{2,2}}\right) &\defeq \AmIIn*{\rho}{(E_{1,1} \ntwkmerge E_{2,1})}{(E_{1,2} \ntwkmerge E_{2,2})}
  \\[\vertrulegap]
  \left(\PID{\alpha}{E_1}\right) \ntwkmerge \left(\PID{\alpha}{E_2}\right) &\defeq \PID{\alpha}{(E_1 \ntwkmerge E_2)}
  \\[\vertrulegap]
  \left(\NtwkFork*{(\alpha,x)}{E_{1,1}}{E_{1,2}}\right) \ntwkmerge \left(\NtwkFork*{(\alpha,x)}{E_{2,1}}{E_{2,2}}\right) &\defeq \NtwkFork*{(\alpha,x)}{E_{1,1} \ntwkmerge E_{2,1}}{(E_{1,2} \ntwkmerge E_{2,2})}
  \\[\vertrulegap]
  \NtwkExit \ntwkmerge \NtwkExit &\defeq \NtwkExit
\end{align*}
\endgroup

\subsection{Endpoint Projection}
\label{sec:proj-def}
Note that $\AmI{\ell}{E_1}{E_2}$ is shorthand for $\AmIIn{\{\ell\}}{E_1}{E_2}$.
The clauses for~$\setKnd$ apply verbatim to~$\finsetKnd$.
The formalization also contains a second projection function, used in the proofs, which
represents a projection target by a location~$\ell$ and a list~$A$ of location-set variables assumed to contain~$\ell$.
In particular, the projection of a location-set binder~$\alpha$ records whether the bound set contains the target
(the two branches of the test on~$\alpha$ use the assumptions $\alpha::A$ and $A$).
That is, the alternate definition does not perform a syntactic substitution of $\{\ell\}\cup\alpha$;
the two projections therefore agree only up to this substitution, i.e., up to set equality of location sets.
\begingroup
\allowdisplaybreaks
\rulefiguresize
\begin{align*}
  \epp{X_\rho}{\ell} & =
    \begin{cases}
      X & \ifText \ell \in \rho \\
      \NtwkUnit & \owText
    \end{cases}
  \\[\vertrulegap]
  \epp{\rho.e}{\ell} & =
    \begin{cases}
      \Ret{e} & \ifText \ell \in \rho \\
      \NtwkUnit & \owText
    \end{cases}
  \\[\vertrulegap]
  \epp{\Fun{\rho}{F}{X}{C}}{\ell} & \defeq \begin{cases}
      \mathrlap{\NtwkFun{F}{X}{\AmI{\ell}{\epp{C}{\ell}}{(\PID{\alpha}{\epp{C}{\alpha}})}} \quad \ifText \ell \in \rho} & \\
      \NtwkUnit & \qquad \ifText \ell \notin \rho \andText \epp{C}{\ell} \neq \Undef \\
      \Undef & \qquad \owText
    \end{cases}
  \\[\vertrulegap]
  \epp{C_1 \appchor{\rho} C_2}{\ell} & =
    \begin{cases}
      \epp{C_1}{\ell}~\epp{C_2}{\ell} & \ifText \ell \in \rho \\
      \epp{C_1}{\ell} \seqfun \epp{C_2}{\ell} \seqfun \NtwkUnit & \owText
    \end{cases}
  \\[\vertrulegap]
  \epp{\TFunLoc{F}{\alpha \knd \locKnd}{C}}{\ell} &=
    \NtwkTFun*{F}{\alpha}{
          \quad\AmI*{\ell}
            {(\AmI{\alpha}{\epp{\subst{C}{\alpha}{\ell}}{\ell}}{\epp{C}{\ell}})}
            {(\PID{\beta}{(\AmI{\alpha}{\epp{\subst{C}{\alpha}{\beta}}{\beta}}{\epp{C}{\beta}})})}}
  \\[\vertrulegap]
  \epp{\TFunLoc{F}{\alpha \knd \setKnd}{C}}{\ell} &=
    \NtwkTFun*{F}{\alpha}{
          \quad\AmI*{\ell}
            {(\AmIIn{\alpha}{\epp{\subst{C}{\alpha}{\{\ell\} \cup \alpha}}{\ell}}{\epp{C}{\ell}})}
            {(\PID{\beta}{(\AmIIn{\alpha}{\epp{\subst{C}{\alpha}{\{\beta\} \cup \alpha}}{\beta}}{\epp{C}{\beta}})})}}
  \\[\vertrulegap]
  \epp{\TFun{\rho}{F}{\alpha \knd \kappa}{C}}{\ell} & =
  \begin{cases}
    \mathrlap{\NtwkTFun{F}{\alpha}{\AmI{\ell}{\epp{C}{\ell}}{(\PID{\beta}{\epp{C}{\beta}})}} \quad \ifText \ell \in \rho} & \\
    \NtwkUnit & \ifText \ell \notin \rho \andText \epp{C}{\ell} \neq \Undef \\
    \Undef & \owText
  \end{cases}
  \\[\vertrulegap]
  \epp{C \appchor{\rho} t}{\ell} &=
    \begin{cases}
      \epp{C}{\ell}~t & \ifText \ell \in \rho \\
      \epp{C}{\ell} \seqfun \NtwkUnit & \owText
    \end{cases}
  \\[\vertrulegap]
  \epp{\Fold{\rho}{C}}{\ell} & =
    \begin{cases}
      \NtwkFold{\epp{C}{\ell}} & \ifText \ell \in \rho \\
      \epp{C}{\ell} & \owText
    \end{cases}
  \\[\vertrulegap]
  \epp{\Unfold{\rho}{C}}{\ell} & =
    \begin{cases}
      \NtwkUnfold{\epp{C}{\ell}} & \ifText \ell \in \rho \\
      \epp{C}{\ell} \seqfun \NtwkUnit & \owText
    \end{cases}
  \\[\vertrulegap]
  \epp{(C_1,C_2)_\rho}{\ell} & =
    \begin{cases}
      (\epp{C_1}{\ell},\epp{C_2}{\ell}) & \ifText \ell \in \rho \\
      \epp{C_1}{\ell} \seqfun \epp{C_2}{\ell} \seqfun \NtwkUnit & \owText
    \end{cases}
  \\[\vertrulegap]
  \epp{\Fst{\rho}{C}}{\ell} & =
    \begin{cases}
      \NtwkFst{\epp{C}{\ell}} & \ifText \ell \in \rho \\
      \epp{C}{\ell} \seqfun \NtwkUnit & \owText
    \end{cases}
  \\[\vertrulegap]
  \epp{\Snd{\rho}{C}}{\ell} & =
    \begin{cases}
      \NtwkSnd{\epp{C}{\ell}} & \ifText \ell \in \rho \\
      \epp{C}{\ell} \seqfun \NtwkUnit & \owText
    \end{cases}
  \\[\vertrulegap]
  \epp{\Inl{\rho}{C}}{\ell} & =
    \begin{cases}
      \NtwkInl{\epp{C}{\ell}} & \ifText \ell \in \rho \\
      \epp{C}{\ell} & \owText
    \end{cases}
  \\[\vertrulegap]
  \epp{\Inr{\rho}{C}}{\ell} & =
    \begin{cases}
      \NtwkInr{\epp{C}{\ell}} & \ifText \ell \in \rho \\
      \epp{C}{\ell} & \owText
    \end{cases}
  \\[\vertrulegap]
  \epp{\Case*{\rho}{C}{X}{C_1}{Y}{C_2}}{\ell} & =
    \begin{cases}
      \NtwkCase*{\epp{C}{\ell}}{X}{\epp{C_1}{\ell}}{Y}{\epp{C_2}{\ell}} & \ifText \ell \in \rho \\
      \epp{C}{\ell} \seqfun (\epp{C_1}{\ell} \ntwkmerge \epp{C_2}{\ell}) & \ifText \ell \notin \rho \andText X \notin \fv{\epp{C_1}{\ell}} \andText Y \notin \fv{\epp{C_2}{\ell}} \\
      \NtwkNone & \owText
    \end{cases}
    \\[\vertrulegap]
  \epp{\LocalCase*{\rho}{C}{x}{C_1}{y}{C_2}}{\ell} & =
    \begin{cases}
      \NtwkLocalCase*{\epp{C}{\ell}}{x}{\epp{C_1}{\ell}}{y}{\epp{C_2}{\ell}} & \ifText \ell \in \rho \\
      \epp{C}{\ell} \seqfun (\epp{C_1}{\ell} \ntwkmerge \epp{C_2}{\ell}) & \ifText \ell \notin \rho \andText x \notin \fv{\epp{C_1}{\ell}} \andText y \notin \fv{\epp{C_2}{\ell}} \\
      \NtwkNone & \owText
    \end{cases}
  \\[\vertrulegap]
  \epp{\LetIn{\rho.x \ty t_e}{C_1}{C_2}}{\ell} & =
  \begin{cases}
    \NtwkLetIn{x}{\epp{C_1}{\ell}}{\epp{C_2}{\ell}} & \ifText \ell \in \rho\\
    \epp{C_1}{\ell} \seqfun \epp{C_2}{\ell} & \ifText \ell \notin \rho \andText x \notin \fv{\epp{C_2}{\ell}}\\
    \NtwkNone & \owText
  \end{cases}
  \\[\vertrulegap]
  \epp{\LetIn{\rho.\alpha \knd *_e}{C_1}{C_2}}{\ell} & =
  \begin{cases}
    \NtwkLetIn{\alpha}{\epp{C_1}{\ell}}{\epp{C_2}{\ell}} & \ifText \ell \in \rho\\
    \epp{C_1}{\ell} \seqfun \epp{C_2}{\ell} & \ifText \ell \notin \rho \andText \alpha \notin \fv{\epp{C_2}{\ell}}\\
    \NtwkNone & \owText
  \end{cases}
  \\[\vertrulegap]
  \addtocounter{numlevels}{1}
  \epp{\LetIn{\rho.\alpha \knd \locKnd}{C_1}{C_2}}{\ell} & =
  \begin{cases}
    \NtwkLetIn*{\alpha}{\epp{C_1}{\ell}}{\AmI{\alpha}{\epp{\subst{C_2}{\alpha}{\ell}}{\ell}}{\epp{C_2}{\ell}}} & \ifText \ell \in \rho\\
    \epp{C_1}{\ell} \seqfun \epp{C_2}{\ell} & \ifText \ell \notin \rho  \andText  \alpha \notin \fv{\epp{C_2}{\ell}}\\
    \Undef & \owText
  \end{cases}
  \\[\vertrulegap]
  \epp{\LetIn{\rho.\alpha \knd \setKnd}{C_1}{C_2}}{\ell} & =
  \begin{cases}
    \NtwkLetIn*{\alpha}{\epp{C_1}{\ell}}{\AmIIn*{\alpha}{\epp{\subst{C_2}{\alpha}{\{\ell\} \cup \alpha}}{\ell}}{\epp{C_2}{\ell}}} & \ifText \ell \in \rho\\
    \epp{C_1}{\ell} \seqfun \epp{C_2}{\ell} & \ifText \ell \notin \rho  \andText  \alpha \notin \fv{\epp{C_2}{\ell}}\\
    \NtwkNone & \owText
  \end{cases}
  \addtocounter{numlevels}{-1}
  \\[\vertrulegap]
  \epp{C \ChorSend[\ell'] \rho}{\ell} & =
    \begin{cases}
      \SendTo{\epp{C}{\ell}}{\rho} & \ifText \ell = \ell' \\
      \epp{C}{\ell} \seqfun \RecvFrom{\ell'} & \ifText \ell \neq \ell'  \andText \ell \in \rho\\
      \epp{C}{\ell} & \owText
    \end{cases}
  \\[\vertrulegap]
  \epp{\syncs{\ell'}{d}{\rho} \seq C}{\ell} & =
    \begin{cases}
      \ChooseFor{d}{\rho}{\epp{C}{\ell}} & \ifText \ell = \ell' \\
      \AllowOneChoice{\ell'}{\Left}{\epp{C}{\ell}} & \ifText \ell \neq \ell' \andText  \ell \in \rho  \andText  d = \Left \\
      \AllowOneChoice{\ell'}{\Right}{\epp{C}{\ell}} & \ifText \ell \neq \ell' \andText  \ell \in \rho  \andText  d = \Right \\
      \epp{C}{\ell} & \owText
  \end{cases}
  \\[\vertrulegap]
  \epp{\Fork{(\alpha,x)}{\ell'}{C}}{\ell} & =
    \begin{cases}
      \NtwkFork{(\alpha,x)}{\epp{C}{\alpha}}{\epp{C}{\ell}} & \ifText \ell = \ell' \\
      \epp{C}{\ell} & \ifText \ell \neq \ell' \andText  \alpha, x \notin \fv{\epp{C}{\ell}} \\
      \NtwkNone & \owText
    \end{cases}
  \\[\vertrulegap]
  \epp{\KillAfter{L}{C}}{\ell} & =
    \begin{cases}
      \Undef & \ifText \ell = L \\
      \epp{C}{\ell} & \owText
    \end{cases}
\end{align*}
\endgroup

\subsection{Endpoint Projection for Spawned Threads}
\label{sec:proj-def-thread}
We first define the function~$\tbodies{C}{L}$ that collects the list of all bodies~$C'$ of the
subterms in~$C$ of the form~$\KillAfter{L}{C'}$ (that do not appear inside a closure).
When a program has multiple subterms, a thread can only be spawned in one
subterm (except for case-expressions). 
Consequently, $S_1\bowtie S_2$ selects the nonempty list, or returns the empty list if both lists are nonempty.
Function and type-function bodies are closures, so the collector does not descend into them.
\begingroup
\allowdisplaybreaks
\rulefiguresize
\begin{align*}
  S_1 \bowtie S_2 &=
    \begin{cases}
      S_2 & S_1 = [] \\
      S_1 & S_2 = [] \\
      [] & S_1 \neq [] \andText S_2 \neq []
    \end{cases}
  \\[\vertrulegap]
  \tbodies{X_\rho}{L} &= \tbodies{\rho.e}{L} = \tbodies{\Fun{\rho}{F}{X}{C}}{L} = \tbodies{\TFun{\rho}{F}{\alpha \knd \kappa}{C}}{L} = []
  \\[\vertrulegap]
  \tbodies{C \appchor{\rho} t}{L} &= \tbodies{\Fst{\rho}{C}}{L} = \tbodies{\Snd{\rho}{C}}{L} = \tbodies{\Inl{\rho}{C}}{L} = \tbodies{\Inr{\rho}{C}}{L} \\
    &= \tbodies{\Fold{\rho}{C}}{L} = \tbodies{\Unfold{\rho}{C}}{L} \\
    &= \tbodies{C \ChorSend[\ell] \rho}{L} = \tbodies{\syncs{\ell}{d}{\rho} \seq C}{L} \\
    &= \tbodies{\Fork{(\alpha,x)}{\ell}{C}}{L} \\
    &= \tbodies{C}{L}
  \\[\vertrulegap]
  \tbodies{C_1 \appchor{\rho} C_2}{L} &= \tbodies{(C_1,C_2)_\rho}{L} \\
    &= \tbodies{\LetIn{\rho.x}{C_1}{C_2}}{L} = \tbodies{\LetIn{\rho.\alpha}{C_1}{C_2}}{L} \\
    &= \tbodies{C_1}{L} \bowtie \tbodies{C_2}{L}
  \\[\vertrulegap]
  \tbodies{\Case*{\rho}{C}{X}{C_1}{Y}{C_2}}{L} &=
      \tbodies{C}{L} \bowtie (\tbodies{C_1}{L} \mathbin{+\!\!+} \tbodies{C_2}{L})
  \\[\vertrulegap]
  \tbodies{\LocalCase*{\rho}{C}{x}{C_1}{y}{C_2}}{L} &=
    \tbodies{C}{L} \bowtie (\tbodies{C_1}{L} \mathbin{+\!\!+} \tbodies{C_2}{L})
  \\[\vertrulegap]
  \tbodies{\KillAfter{L'}{C}}{L} & =
    \begin{cases}
      [C] & \ifText L = L'\\
      \tbodies{C}{L} & \owText
    \end{cases}
\end{align*}
Using this function we can then define~$\eppfork{C}{L}$ to project non-threads normally,
and to project to the projected then merged~\KillAfterN bodies followed by an~\NtwkExit for threads.
\begin{mathpar}
  \eppfork{C}{L} =
    \begin{cases}
      \epp{C}{L} & \ifText \tbodies{C}{L} = [] \\
      (\epp{C_1}{L} \ntwkmerge \ldots \ntwkmerge \epp{C_n}{L})
      \seqfun \NtwkExit
      & \ifText \tbodies{C}{L} = [C_1,\ldots,C_n]~\text{where}~n > 0
    \end{cases}
  \end{mathpar}
\endgroup

\subsection{Locations Named by a Type or Choreography}
\label{sec:pn-def}
The set of named locations in a choreography (resp. type)~$\nl{C}$ (resp.~$\nl{t}$) is defined as follows.
\begingroup
\allowdisplaybreaks
\rulefiguresize
\begin{align*}
  \namedlocs{\alpha} &= \varnothing
  \\[\vertrulegap]
  \namedlocs{L} &= \{L\}
  \\[\vertrulegap]
  \namedlocs{\{\ell\}} &= \namedlocs{\ell}
  \\[\vertrulegap]
  \namedlocs{\rho_1 \cup \rho_2} &= \namedlocs{\rho_1} \cup \namedlocs{\rho_2}
  \\[\vertrulegap]
  \namedlocs{\anyLoc} &= \varnothing
  \\[\vertrulegap]
  \namedlocs{t_e @ \rho} &= \namedlocs{t_e} \cup \namedlocs{\rho}
  \\[\vertrulegap]
  \namedlocs{\tau_1 \arr{\rho} \tau_2} &= \namedlocs{\tau_1} \cup \namedlocs{\tau_2} \cup \namedlocs{\rho}
  \\[\vertrulegap]
  \namedlocs{\tau_1 +_{\rho} \tau_2} &= \namedlocs{\tau_1} \cup \namedlocs{\tau_2} \cup \namedlocs{\rho}
  \\[\vertrulegap]
  \namedlocs{\tau_1 \times \tau_2} &= \namedlocs{\tau_1} \cup \namedlocs{\tau_2}
  \\[\vertrulegap]
  \namedlocs{\mu_\rho \alpha \ldotp \tau} &= \namedlocs{\rho} \cup \namedlocs{\tau}
  \\[\vertrulegap]
  \namedlocs{\allty{\alpha \knd \kappa}{\rho}{\tau}} &= \namedlocs{\rho} \cup \namedlocs{\tau}
  \\[\vertrulegap]
  \namedlocs{X\!:\!\tau} &= \namedlocs{\tau}
  \\[\vertrulegap]
  \namedlocs{\rho.e} &= \namedlocs{\rho}
  \\[\vertrulegap]
  \namedlocs{\Fun{\rho}{F}{X}{C}} &= \namedlocs{\rho} \cup \namedlocs{C}
  \\[\vertrulegap]
  \namedlocs{C_1 \appchor{\rho} C_2} &= \namedlocs{\rho} \cup \namedlocs{C_1} \cup \namedlocs{C_2}
  \\[\vertrulegap]
  \namedlocs{\TFun{\rho}{F}{\alpha \knd \kappa}{C}} &= \namedlocs{\rho} \cup \namedlocs{C}
  \\[\vertrulegap]
  \namedlocs{C \appchor{\rho} t} &= \namedlocs{C} \cup \namedlocs{\rho} \cup \namedlocs{t}
  \\[\vertrulegap]
  \namedlocs{\Fold{\rho}{C}} &= \namedlocs{\rho} \cup \namedlocs{C}
  \\[\vertrulegap]
  \namedlocs{\Unfold{\rho}{C}} &= \namedlocs{\rho} \cup \namedlocs{C}
  \\[\vertrulegap]
  \namedlocs{(C_1,C_2)_\rho} &= \namedlocs{\rho} \cup \namedlocs{C_1} \cup \namedlocs{C_2}
  \\[\vertrulegap]
  \namedlocs{\Fst{\rho}{C}} &= \namedlocs{\rho} \cup \namedlocs{C}
  \\[\vertrulegap]
  \namedlocs{\Snd{\rho}{C}} &= \namedlocs{\rho} \cup \namedlocs{C}
  \\[\vertrulegap]
  \namedlocs{\Inl{\rho}{C}} &= \namedlocs{\rho} \cup \namedlocs{C}
  \\[\vertrulegap]
  \namedlocs{\Inr{\rho}{C}} &= \namedlocs{\rho} \cup \namedlocs{C}
  \\[\vertrulegap]
  \operatorname{NL}\left(\Case*{\rho}{C}{X}{C_1}{Y}{C_2}\right) &= \namedlocs{\rho} \cup \namedlocs{C} \cup \namedlocs{C_1} \cup \namedlocs{C_2}
  \\[\vertrulegap]
  \operatorname{NL}\left(\LocalCase*{\rho}{C}{x}{C_1}{y}{C_2}\right) &= \namedlocs{\rho} \cup \namedlocs{C} \cup \namedlocs{C_1} \cup \namedlocs{C_2}
  \\[\vertrulegap]
  \namedlocs{\LetIn{\rho.x}{C_1}{C_2}} &= \namedlocs{\rho} \cup \namedlocs{C_1} \cup \namedlocs{C_2}
  \\[\vertrulegap]
  \namedlocs{\LetIn{\rho.\alpha}{C_1}{C_2}} &= \namedlocs{\rho} \cup \namedlocs{C_1} \cup \namedlocs{C_2}
  \\[\vertrulegap]
  \namedlocs{C \ChorSend[\ell] \rho} &= \namedlocs{\ell} \cup \namedlocs{\rho} \cup \namedlocs{C}
  \\[\vertrulegap]
  \namedlocs{\syncs{\ell}{d}{\rho} \seq C} &= \namedlocs{\ell} \cup \namedlocs{\rho} \cup \namedlocs{C}
  \\[\vertrulegap]
  \namedlocs{\Fork{(\alpha,x)}{\ell}{C}} &= \namedlocs{\ell} \cup \namedlocs{C}
  \\[\vertrulegap]
  \namedlocs{\KillAfter{L}{C}} &= \{L\} \cup \namedlocs{C}
\end{align*}
\endgroup

\subsection{Spawned Locations in a Choreography}
\begingroup
\allowdisplaybreaks
\rulefiguresize
\begin{align*}
  \spawnedlocs{X} &= \varnothing
  \\[\vertrulegap]
  \spawnedlocs{\rho.e} &= \varnothing
  \\[\vertrulegap]
  \spawnedlocs{\Fun{\rho}{F}{X}{C}} &= \spawnedlocs{C}
  \\[\vertrulegap]
  \spawnedlocs{C_1 \appchor{\rho} C_2} &= \spawnedlocs{C_1} \cup \spawnedlocs{C_2}
  \\[\vertrulegap]
  \spawnedlocs{\TFunLoc{F}{\alpha}{C}} &= \spawnedlocs{C}
  \\[\vertrulegap]
  \spawnedlocs{C \appchor{\rho} t} &= \spawnedlocs{C}
  \\[\vertrulegap]
  \spawnedlocs{\Fold{\rho}{C}} &= \spawnedlocs{C}
  \\[\vertrulegap]
  \spawnedlocs{\Unfold{\rho}{C}} &= \spawnedlocs{C}
  \\[\vertrulegap]
  \spawnedlocs{(C_1,C_2)_\rho} &= \spawnedlocs{C_1} \cup \spawnedlocs{C_2}
  \\[\vertrulegap]
  \spawnedlocs{\Fst{\rho}{C}} &= \spawnedlocs{C}
  \\[\vertrulegap]
  \spawnedlocs{\Snd{\rho}{C}} &= \spawnedlocs{C}
  \\[\vertrulegap]
  \spawnedlocs{\Inl{\rho}{C}} &= \spawnedlocs{C}
  \\[\vertrulegap]
  \spawnedlocs{\Inr{\rho}{C}} &= \spawnedlocs{C}
  \\[\vertrulegap]
  \operatorname{SL}\left(\Case*{\rho}{C}{X}{C_1}{Y}{C_2}\right) &= \spawnedlocs{C} \cup \spawnedlocs{C_1} \cup \spawnedlocs{C_2}
  \\[\vertrulegap]
  \operatorname{SL}\left(\LocalCase*{\rho}{C}{x}{C_1}{y}{C_2}\right) &= \spawnedlocs{C} \cup \spawnedlocs{C_1} \cup \spawnedlocs{C_2}
  \\[\vertrulegap]
  \spawnedlocs{\LetIn{\rho.x}{C_1}{C_2}} &= \spawnedlocs{C_1} \cup \spawnedlocs{C_2}
  \\[\vertrulegap]
  \spawnedlocs{\LetIn{\rho.\alpha}{C_1}{C_2}} &= \spawnedlocs{C_1} \cup \spawnedlocs{C_2}
  \\[\vertrulegap]
  \spawnedlocs{C \ChorSend[\ell] \rho} &= \spawnedlocs{C}
  \\[\vertrulegap]
  \spawnedlocs{\syncs{\ell}{d}{\rho} \seq C} &= \spawnedlocs{C}
  \\[\vertrulegap]
  \spawnedlocs{\Fork{(\alpha,x)}{\ell}{C}} &= \spawnedlocs{C}
  \\[\vertrulegap]
  \spawnedlocs{\KillAfter{L}{C}} &= \{L\} \cup \spawnedlocs{C}
\end{align*}
\endgroup

\subsection{The Less-Than Relation}
\label{sec:less-nonderm-def}
The relation is indexed by the name~$\ell$ of the currently running location: $E_1 \lessthan[\ell] E_2$.
Here~$\ell$ is either a concrete location or a location variable.
We omit this index in the rules below when it is unchanged between the premise(s) and the conclusion.

\begin{mathparpagebreakable}
  \infer{E_1 \lessthan E_2 \\ \NtwkVal{V}}
  {E_1 \lessthan V \NtwkSeq E_2}
  \and
  \infer[]{~}
  {X \lessthan X}
  \and
  \infer[]{~}
  {\NtwkUnit \lessthan \NtwkUnit}
  \and
  \infer{\NtwkVal{V} \\ V \lessthan E}
  {\NtwkUnit \lessthan E}
  \and
  \infer[]{~}
  {\Ret{e} \lessthan \Ret{e}}
  \and
  \infer[]{E_{1,1} \lessthan E_{2,1} \\
    E_{1,2} \lessthan E_{2,2}}
  {E_{1,1} \NtwkSeq E_{1,2} \lessthan E_{2,1} \NtwkSeq E_{2,2}}
  \and
  \infer[]{E_1 \lessthan E_2}
  {\NtwkFun{F}{X}{E_1} \lessthan \NtwkFun{F}{X}{E_2}}
  \and
  \infer[]{E_{1,1} \lessthan E_{2,1} \\
  E_{1,2} \lessthan E_{2,2}}
  {E_{1,1}~E_{1,2} \lessthan E_{2,1}~E_{2,2}}
  \and
  \infer[]{E_1 \lessthan E_2}
  {\NtwkTFun{F}{\alpha}{E_1} \lessthan \NtwkTFun{F}{\alpha}{E_2}}
  \and
  \infer[]{E_1 \lessthan E_2}
  {E_1~t \lessthan E_2~t}
  \and
  \infer[]{E_1 \lessthan E_2}
  {\NtwkFold{E_1} \lessthan \NtwkFold{E_2}}
  \and
  \infer[]{E_1 \lessthan E_2}
  {\NtwkUnfold{E_1} \lessthan \NtwkUnfold{E_2}}
  \and
  \infer[]{E_{1,1} \lessthan E_{2,1} \\
    E_{1,2} \lessthan E_{2,2}}
  {(E_{1,1},E_{1,2}) \lessthan (E_{2,1},E_{2,2})}
  \and
  \infer[]{E_1 \lessthan E_2}
  {\NtwkFst{E_1} \lessthan \NtwkFst{E_2}}
  \and
  \infer[]{E_1 \lessthan E_2}
  {\NtwkSnd{E_1} \lessthan \NtwkSnd{E_2}}
  \and
  \infer[]{E_1 \lessthan E_2}
  {\NtwkInl{E_1} \lessthan \NtwkInl{E_2}}
  \and
  \infer[]{E_1 \lessthan E_2}
  {\NtwkInr{E_1} \lessthan \NtwkInr{E_2}}
  \and
  \infer[]{E_{1,1} \lessthan E_{2,1} \\
    E_{1,2} \lessthan E_{2,2} \\
    E_{1,3} \lessthan E_{2,3}}
  {\NtwkCase*{E_{1,1}}{X}{E_{1,2}}{Y}{E_{1,3}} \lessthan \NtwkCase*{E_{2,1}}{X}{E_{2,2}}{Y}{E_{2,3}}}
  \and
  \infer[]{E_{1,1} \lessthan E_{2,1} \\
    E_{1,2} \lessthan E_{2,2} \\
    E_{1,3} \lessthan E_{2,3}}
  {\NtwkLocalCase*{E_{1,1}}{x}{E_{1,2}}{y}{E_{1,3}} \lessthan \NtwkLocalCase*{E_{2,1}}{x}{E_{2,2}}{y}{E_{2,3}}}
  \and
  \infer[]{E_{1,1} \lessthan E_{2,1} \\
    E_{1,2} \lessthan E_{2,2}}
  {\NtwkLetIn{x}{E_{1,1}}{E_{1,2}} \lessthan \NtwkLetIn{x}{E_{2,1}}{E_{2,2}}}
  \and
  \infer[]{E_{1,1} \lessthan E_{2,1} \\
    E_{1,2} \lessthan E_{2,2}}
  {\NtwkLetIn{\alpha \knd \kappa}{E_{1,1}}{E_{1,2}} \lessthan \NtwkLetIn{\alpha \knd \kappa}{E_{2,1}}{E_{2,2}}}
  \and
  \infer[]{E_1 \lessthan E_2}
  {\SendTo{E_1}{\rho} \lessthan \SendTo{E_2}{\rho}}
  \and
  \infer[]{~}
  {\RecvFrom{\ell} \lessthan \RecvFrom{\ell}}
  \and
  \infer[]{E_1 \lessthan E_2}
  {\ChooseFor{d}{\ell}{E_1} \lessthan \ChooseFor{d}{\ell}{E_2}}
  \and
  \infer[]{E_{1,1} \lessthan E_{2,1} \\
    E_{1,2} \lessthan E_{2,2}}
  {\AllowChoice*{\ell}{E_{1,1}}{E_{1,2}} \lessthan \AllowChoice*{\ell}{E_{2,1}}{E_{2,2}}}
  \and
  \infer[]{E_1 \lessthan E_1'}
  {\AllowOneChoice{\ell}{\Left}{E_1} \lessthan \AllowChoice*{\ell}{E_1'}{E_2'}}
  \and
  \infer[]{E_2 \lessthan E_2'}
  {\AllowOneChoice{\ell}{\Right}{E_2} \lessthan \AllowChoice*{\ell}{E_1'}{E_2'}}
  \and
  \infer{E_1 \lessthan E_1' \\ E_2 \lessthan E_2'}
  {\AmIIn{\rho}{E_1}{E_2} \lessthan \AmIIn{\rho}{E_1'}{E_2'}}
  \and
  \infer{
    E_1 \lessthan[\ell] E_2' \\
    \ell \neq \ell'
  }{
    \AmI{\ell}{E_1}{E_2} \lessthan[\ell] \AmI{\ell'}{E_1'}{E_2'}
  }
  \and
  \infer[]{E_1 \lessthan[\alpha] E_1' \\
    E_2 \lessthan[\ell] E_2'}
  {\NtwkFork{(\alpha,x)}{E_1}{E_2} \lessthan[\ell] \NtwkFork{(\alpha,x)}{E_1'}{E_2'}}
  \and
  \infer[]{~}
  {\NtwkExit \lessthan \NtwkExit}
  \and
  \infer[]{E \lessthan E'[\alpha \mapsto \ell]}
  {E \lessthan[\ell] \PID{\alpha}{E'}}
  \and
  \infer[]{E \lessthan[\alpha] E'}
  {\PID{\alpha}{E} \lessthan[\ell] \PID{\alpha}{E'}}
\end{mathparpagebreakable}

\end{document}
\typeout{get arXiv to do 4 passes: Label(s) may have changed. Rerun}